%% file: legs.tex
\newcommand{\paperOption}{onecolumn,11pt}
\newcommand{\hyperrefOption}{%
 pagebackref,pdfstartview=FitH,hidelinks%
}
\newcommand{\omitstyle}{\color[rgb]{0.25,0,0.75}}
\newcommand{\omitted}[1]{{\omitstyle#1}}
\newcommand{\newChange}[1]{{\color{blue}#1}}
\InputIfFileExists{option.aux}{%
 \typeout{Compiling options activated!}%
}{}

\documentclass[\paperOption]{IEEEtran}

\usepackage{clproject}
\usepackage[\hyperrefOption]{hyperref}
\usepackage[noadjust]{cite}
\usepackage{graphicx}

\usepackage[cmex10]{amsmath}
\usepackage{amsopn,amssymb}
\usepackage{bm}
\usepackage{cases}

\usepackage{color}

\setProjectProperty{paper.properties}

\newtheorem{theorem}{Theorem}[section]
\newtheorem{lemma}[theorem]{Lemma}
\newtheorem{proposition}[theorem]{Proposition}
\newtheorem{corollary}[theorem]{Corollary}
\newtheorem{definition}[theorem]{Definition}
\newtheorem{example}[theorem]{Example}
\newtheorem{remark}[theorem]{Remark}

\newenvironment{proofof}[2][Proof of]{%
 \begin{IEEEproof}[#1 {#2}]%
}{\end{IEEEproof}}

\newcommand{\relvar}[2]{\buildrel {#2} \over {#1}}
\newcommand{\eqvar}[1]{\relvar{=}{\mathrm{#1}}}

\newcommand{\levar}[1]{\relvar{\le}{\mathrm{#1}}}

\newcommand{\eqdef}{:=}

\newlength{\nullrellength}
\newcommand{\breakop}[1]{\hspace{-0.2222em}{}#1}

\DeclareMathOperator*{\plim}{p-lim}

\newcommand{\integers}{\mathbb{Z}}
\newcommand{\pintegers}{\mathbb{N}}
\newcommand{\nnintegers}{\mathbb{N}_0}
\newcommand{\realnumbers}{\mathbb{R}}
\newcommand{\complexnumbers}{\mathbb{C}}
\newcommand{\unitsof}[1]{{#1^{\times}}}

\newcommand{\field}[1]{\mathbb{F}_{#1}}
\newcommand{\fieldstar}[1]{\unitsof{\field{#1}}}

\newcommand{\symmetricgroup}[1]{\mathrm{S}_{#1}}

\newcommand{\seq}[1]{\mathbf{#1}}
\newcommand{\seqx}[1]{\bm{#1}}
\newcommand{\mat}[1]{\mathbf{#1}}
\newcommand{\matentry}[1]{\mathrm{#1}}
\newcommand{\transpose}[1]{#1^{\mathsf{T}}}
\newcommand{\rank}[1]{\mathrm{rank}{\left({#1}\right)}}
\newcommand{\rdist}{\mathrm{d}_{\mathrm{R}}}

\DeclareMathOperator{\order}{\Theta}

\newcommand{\ceil}[1]{\left\lceil {#1} \right\rceil}

\newcommand{\pprod}{\odot}
\newcommand{\bigpprod}{\bigodot}
\newcommand{\rtilde}[1]{#1^{\sim}}

\newcommand{\pr}{\mathsf{P}}
\newcommand{\av}{\mathsf{E}}
\newcommand{\spec}{\mathsf{S}}
\newcommand{\avS}{\overline{\spec}}
\newcommand{\gf}{\mathcal{G}}
\newcommand{\avG}{\overline{\gf}}

\newcommand{\codevar}[2]{f^{\mathsf{#1}}_{#2}}
\newcommand{\rcodevar}[2]{F^{\mathsf{#1}}_{#2}}

\newcommand{\rlccode}[1]{\rcodevar{RLC}{#1}}

\newcommand{\rmcode}[1]{\rcodevar{RM}{#1}}
\newcommand{\ldcode}[1]{\rcodevar{LD}{#1}}
\newcommand{\repcode}[1]{\codevar{REP}{#1}}
\newcommand{\rrepcode}[1]{\rcodevar{REP}{#1}}
\newcommand{\chkcode}[1]{\codevar{CHK}{#1}}
\newcommand{\rchkcode}[1]{\rcodevar{CHK}{#1}}

\newcommand{\citeSpectrumProperty}{%
 \cite[Proposition~2.4]{JSCC:Yang200904}%
}
\newcommand{\citeGoodLinearCode}{%
 \cite[Proposition~2.5]{JSCC:Yang200904}%
}

\begin{document}

\title{Constructing Linear Encoders with Good Spectra}
\author{%
 Shengtian~Yang,~\IEEEmembership{Member,~IEEE,}
 Thomas~Honold,~\IEEEmembership{Member,~IEEE,}
 Yan~Chen,~\IEEEmembership{Member,~IEEE,}
 Zhaoyang~Zhang,~\IEEEmembership{Member,~IEEE,}
 Peiliang~Qiu,~\IEEEmembership{Member,~IEEE}
 \thanks{
  Version: \exactProjectVersion\ (no.~\releaseNumber).
  This file was generated on \today\ by \texCompilerName{} with format
  \fmtname\ \fmtversion.\par
  \textsf{IEEEtran} paper options: $\mathrm{\paperOption}$.
  Options for \textsf{hyperref}: $\mathrm{\hyperrefOption}$.\par
  \textsf{Color convention:} 1) stuff published;
   2) \omitted{stuff not published};
   3) and \newChange{new changes relative to the published version}.
 }
 \thanks{
  This paper was presented in part at the 2008
  International Conference on Communications and Networking in China
  (CHINACOM'08), Hangzhou, China, August 2008.
 }
 \thanks{
  This work was supported in part by the National Natural Science
  Foundation of China under Grants 60772093, 60802014, and 60872063;
  in part by the Chinese Specialized Research Fund for the Doctoral
  Program of Higher Education under Grants 200803351023 and
  200803351027;
  in part by the National Key Basic Research Program of China under
  Grant 2012CB316104;
  and in part by the Zhejiang Provincial Natural Science Foundation of
  China under Grants Y106068 and LR12F01002.}
 \thanks{
  S.~Yang was with the Department of Information Science and
  Electronic Engineering, Zhejiang University, Hangzhou 310027, China
  and also with Zhejiang Provincial Key Laboratory of Information
  Network Technology, Zhejiang University, Hangzhou 310027, China.
  He now resides at Zhengyuan Xiaoqu 10-2-101, Fengtan Road, Hangzhou
  310011, China (e-mail: \url{yangst@codlab.net}).
 }
 \thanks{
  T.~Honold, Z.~Zhang, and P.~Qiu are with the Department of
  Information Science and Electronic Engineering, Zhejiang University,
  Hangzhou 310027, China (e-mail: \url{honold@zju.edu.cn};
  \url{ning_ming@zju.edu.cn}; \url{qiupl@zju.edu.cn}).}
 \thanks{
  Y.~Chen was with the Department of Information Science and
  Electronic Engineering, Zhejiang University, Hangzhou 310027, China.
  She is now with Huawei Technologies Co., Ltd (Shanghai), Shanghai
  201206, China (e-mail: \url{bigbird.chenyan@huawei.com}).
 }
}

\markboth{%
 Accepted for publication in IEEE Transactions on Information Theory
 \omitted{(extended version)}
}{}


\maketitle

\begin{abstract}
  Linear encoders with good joint spectra are suitable candidates for
  optimal lossless joint source-channel coding (JSCC), where the joint
  spectrum is a variant of the input-output complete weight
  distribution and is considered good if it
  is close to the average joint spectrum of all linear encoders (of
  the same coding rate).
  In spite of their existence, little is known on how to construct
  such encoders in practice.
  This paper is devoted to their construction.
  In particular, two families of linear encoders are presented and
  proved to have good joint spectra.
  The first family is derived from Gabidulin codes, a class
  of maximum-rank-distance codes.
  The second family is constructed using a serial concatenation of an
  encoder of a low-density parity-check code (as outer encoder) with a
  low-density generator matrix encoder (as inner encoder).
  In addition, criteria for good linear encoders are defined for three
  coding applications: lossless source coding, channel coding, and
  lossless JSCC.
  In the framework of the code-spectrum approach, these three
  scenarios correspond to the problems of constructing linear
  encoders with good kernel spectra, good image spectra, and good
  joint spectra, respectively.
  Good joint spectra imply both good kernel spectra and good image
  spectra, and for every linear encoder having a good kernel (resp.,
  image) spectrum, it is proved that there exists a linear encoder not
  only with the same kernel (resp., image) but also with a good joint
  spectrum.
  Thus a good joint spectrum is the most important feature of a linear
  encoder.
\end{abstract}


\IEEEpeerreviewmaketitle

\section{Introduction}

Linear codes, owing to their good structure, are widely applied in the
areas of channel coding, source coding, and joint source-channel
coding (JSCC).
A variety of good linear codes such as Turbo codes
\cite{JSCC:Berrou199305} and low-density parity-check (LDPC) codes
\cite{JSCC:Gallager196300, JSCC:MacKay199903} have been constructed
for channel coding.
In the past decade, the parity-check matrices of good linear codes for
channel coding have also been employed as encoders for distributed
source coding.
They proved good both in theory
\cite{JSCC:Muramatsu200510,JSCC:Yang200503} and practice
\cite{JSCC:Pradhan200100,JSCC:Stankovic200604,JSCC:Coleman200608}.
However, for the general case of lossless JSCC (based on linear
codes), there is still no mature and complete solution (see
\cite{JSCC:Yang200904,JSCC:Yang200909} and the references therein for
background information).

\omitted{%
We do not even know how to design an implementable
optimal JSCC scheme based on linear encoders for arbitrary sources and
channels.  For instance, much work on practical designs of lossless
JSCC based on linear codes has been done for transmission of specific
correlated sources over specific multiple-access channels
(e.g., \cite{JSCC:Garcia200103, JSCC:Zhong200505, JSCC:Zhao200611a,
  JSCC:Murugan200408, JSCC:Garcia200709, JSCC:Zhao200611b}%
), but it is still not clear how to construct an implementable optimal
lossless JSCC scheme for the general case.  The same problem
occurs in the case of point-to-point transmission, but
since traditional (nonlinear) source coding techniques combined with
joint source-channel decoding work well in this case,
linear-encoder-based schemes are less important
here than in distributed JSCC.
One exception is the demand for a simple universal encoder, that is, the
encoder does not require any knowledge of the source statistics.  For
background information on lossless JSCC in the point-to-point case, we
refer to \cite{JSCC:Zhu200408,JSCC:Jaspar200706} and the references
therein.
}

Recently, for lossless transmission of correlated sources over
multiple-access channels (MACs), we proposed a general scheme based on
linear encoders \cite{JSCC:Yang200904}.
It was proved optimal if good linear encoders and good conditional
probability distributions are chosen.%
\footnote{In \cite{JSCC:Yang200904} a linear encoder is called a
``linear code'', which in fact conflicts with the traditional meaning
of the term ``linear code''.
To avoid such conflicts as well as possible misunderstanding, we use
the term ``linear encoder'' in this paper.}
Fig.~\ref{fig:Scheme1} illustrates the mechanism of this encoding
scheme (see \cite[Sec.~III-C]{JSCC:Yang200904}), which can also be
formulated as
\[
\Phi(\seq{x}) \eqdef \varphi(\seq{x}, \Sigma_m(F(\Sigma_n(\seq{x})))
 + \bar{Y}^m) \qquad \forall \seq{x} \in \mathcal{X}^n.
\]
Roughly speaking, the scheme consists of two steps.  First, the source
sequence $\seq{x}$ is processed by a special kind of random affine
mapping, $(\Sigma_m \circ F \circ \Sigma_n)(\seq{x}) + \bar{Y}^m$,
where $F$ is a random linear encoder from $\mathcal{X}^n$ to
$\mathcal{Y}^m$, $\Sigma_m$ and $\Sigma_n$ are uniform random
interleavers, $\bar{Y}^m$ is a uniform random vector, and all
of them are independent.
Second, the output of the
first step, together with the source sequence, is fed into the
quantization map $\varphi$ to yield the final output.
The first step is to
generate uniformly distributed output with the so-called
pairwise-independence property, while the second step is to shape the
output so that it is suitable for a given channel.
From this, two main issues arise:
how can we design a good linear encoder and a good quantization map to
fulfill the above two goals, respectively?
About the former, we proved in \cite{JSCC:Yang200904} that linear
encoders with joint spectra (a variant of input-output complete weight
distribution) close to the average joint spectrum over all linear
encoders (from $\mathcal{X}^n$ to $\mathcal{Y}^m$) are good
candidates.
We say that such linear encoders have good joint spectra.
Hence, for designing a lossless JSCC scheme in practice, the crucial
problem is how to construct linear encoders with good joint spectra.
To our knowledge, however, this problem has never been studied before.
\begin{figure*}[tbp]
\centering
\includegraphics{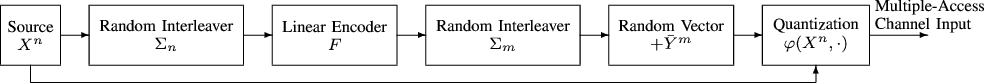}
\caption{The proposed lossless joint source-channel encoding scheme
based on linear encoders for multiple-access channels in
\cite{JSCC:Yang200904}.}
\label{fig:Scheme1}
\end{figure*}

In this paper, we shall give a thorough investigation of this problem.
Our main tool is the code-spectrum approach \cite{JSCC:Yang200904}.
As we shall see, the spectra of a linear encoder, including kernel
spectrum, image spectrum, and especially joint spectrum, provide an
important characterization of its performance for most applications
(see Definitions~\ref{df:KernelSpectrumCondition}--%
\ref{df:JointSpectrumCondition}, which form the
base of this paper).
The rest of the paper is organized as follows.

In Section~\ref{sec:BasicsOfCodeSpectrumApproach}, the code-spectrum
approach is briefly reviewed.
In Section~\ref{sec:ConceptsOfGoodLinearCodes}, three kinds of good
linear encoders are defined for lossless source coding, channel
coding, and lossless JSCC, respectively.
They are called
\emph{$\delta$-asymptotically good linear source encoders} (LSEs),
\emph{$\delta$-asymptotically good linear channel encoders} (LCEs),
and \emph{$\delta$-asymptotically good linear source-channel encoders}
(LSCEs), respectively.
We show that, under some conditions, good LSCEs are also good as LSEs
and LCEs.
Thus the problem of constructing good LSCEs (i.e., linear encoders
with good joint spectra) is of particular interest and importance.

Based on this observation, in
Section~\ref{sec:GeneralPrinciplesForConstructingGoodLSCC}, we proceed
to study the general principles for constructing good LSCEs.
In Section~\ref{subsec:MRDLSCC}, we provide a family of good LSCEs
derived from so-called maximum-rank-distance (MRD) codes.
In Section~\ref{subsec:ConstructingGoodLSCCGPA}, we investigate the
problem of how to construct a good LSCE with the same kernel (resp.,
image) as a given good LSE (resp., LCE).
In Section~\ref{subsec:ConstructingGoodLSCCGPB}, we propose a general
serial concatenation scheme for constructing good LSCEs.
In light of this general scheme, in
Section~\ref{sec:ExplicitConstruction}, we turn to the analysis of
joint spectra of regular low-density generator matrix (LDGM) encoders.
We show that the joint spectra of regular LDGM encoders with
appropriate parameters are approximately $\delta$-asymptotically good.
Based on this fact, we finally construct a family of good LSCEs by
means of a serial concatenation of an inner LDGM encoder and an outer
encoder of an LDPC code.

Some advanced tools of the code-spectrum approach are developed in
Section~\ref{sec:NewMethodsAndResultsOfCodeSpectrum}
in order to prove the results in the preceding sections.
Section~\ref{sec:Conclusion} concludes the paper.

Since this paper is highly condensed, we refer the reader to
\cite{JSCC:Yang200909} for concrete examples and important remarks.
In particular, an example is provided there to show how the binary
$[7,4,3]$ Hamming code is beaten by a $[7,4,1]$ linear code in the
case of lossless JSCC if the encoder (or generator matrix) is not
carefully chosen.
Recall that an $[n,k,d]$ linear code over a finite field $F$ is a
$k$-dimensional subspace of $F^n$ with minimum (Hamming) distance $d$.

We close this section with some basic notations and conventions used
throughout the paper.
In general, mathematical objects such as real variables and
deterministic mappings are denoted by lowercase letters.
Conventionally, sets, matrices, and random elements are denoted by
capital letters, and alphabets are denoted by script capital letters.

The symbols $\integers$, $\pintegers$, $\nnintegers$, $\realnumbers$,
$\complexnumbers$ denote the ring of integers, the set of positive
integers, the set of nonnegative integers, the field of real numbers,
and the field of complex numbers, respectively.
For a prime power $q>1$ the finite field of order $q$ is denoted by
$\field{q}$.
The multiplicative subgroup of nonzero elements of $\field{q}$ is
denoted by $\fieldstar{q}$ (and similarly for other fields).
The index set $\{1,2,\ldots,n\}$ for $n\in\pintegers$ is denoted by
$\mathcal{I}_n$.

A sequence (or vector) in $\mathcal{X}^n$ is denoted by
$\seq{x} = x_1x_2\cdots x_n$, with $x_i$ denoting the $i$th component.
The length of $\seq{x}$ is denoted by $|\seq{x}|$.
For the $l$-fold repetition of a single symbol $a\in\mathcal{X}$, we
write $a^l$ for brevity.
For any set $A = \{a_1, a_2, \dots, a_r\} \subseteq \mathcal{I}_n$
with $a_1 < a_2 < \dots < a_r$, we define the sequence
$(x_i)_{i \in A}$ or $x_A$ as $x_{a_1}x_{a_2}\cdots x_{a_r}$.

For any maps $f: \mathcal{X}_1\to\mathcal{Y}_1$ and
$g: \mathcal{X}_2\to\mathcal{Y}_2$, the cartesian product
$f\pprod g: \mathcal{X}_1\times\mathcal{X}_2 \to
 \mathcal{Y}_1\times\mathcal{Y}_2$
is given by $(x_1, x_2)\mapsto (f(x_1), g(x_2))$.
Given a map $f$ from a finite set $X$ to $\nnintegers$ with
$\sum_{x\in X} f(x)=n$, the multinomial coefficient
$n!/\prod_{x\in X}(f(x)!)$ is denoted by ${n\choose f}$.

\omitted{%
The function $1\{\cdot\}$ is a mapping defined by
$1\{\mathrm{true}\} = 1$ and $1\{\mathrm{false}\} = 0$.
Then the indicator function of a subset $A$ of a set $X$ can be
written as $1\{x\in A\}$.
For $x\in\realnumbers$, $\ceil{x}$ denotes the smallest integer
$\ge x$.
For any real-valued functions $f(n)$ and $g(n)$ with domain
$\pintegers$, the statement $f(n)=\order(g(n))$ refers to the
existence of positive constants $c_1$ and $c_2$ such that
$c_1g(n)\le f(n)\le c_2g(n)$ for sufficiently large $n$.
}

By default, all vectors are regarded as row vectors.
An $m\times n$ matrix is denoted by
$\mat{M}=(\matentry{M}_{i,j})_{i\in\mathcal{I}_m, j\in\mathcal{I}_n}$,
with $\matentry{M}_{i,j}$ denoting the $(i,j)$-th entry.
The transpose of a matrix $\mat{M}$ is denoted by
$\transpose{\mat{M}}$.
The set of all $m \times n$-matrices over a field $F$ is denoted by
$F^{m \times n}$.

When performing probabilistic analysis, all objects of study are
related to a basic probability space $(\Omega, \mathcal{A}, \pr)$
where $\mathcal{A}$ is a $\sigma$-algebra in $\Omega$ and $\pr$ is a
probability measure on $(\Omega, \mathcal{A})$.
For any event $A \in \mathcal{A}$, $\pr A = \pr(A)$ is called the
probability of $A$.
A random element is a measurable mapping of $\Omega$ into some
measurable space $(\Omega', \mathcal{B})$.
\emph{In this paper, distinct random elements are assumed to be
independent, and cartesian products of the same random sets (or maps)
are also regarded as cartesian products of independent copies.}

For a (discrete) probability distribution $P$ on $\mathcal{X}$, the entropy $H(P)$
is given by $-\sum_{a\in\mathcal{X}} P(a)\ln P(a)$, with
$0\ln 0\eqdef 0$.
For probability distributions $P$ and $Q$ on $\mathcal{X}$ with
$P\ll Q$ (i.e., $P$ absolutely continuous with respect to $Q$), the
information divergence $D(P\|Q)$ is given by
$\sum_{a\in\mathcal{X}} P(a) \ln(P(a)/Q(a))$.

\section{Basics of the Code-Spectrum Approach}
\label{sec:BasicsOfCodeSpectrumApproach}

In this section, we briefly introduce the basics of the code-spectrum
approach \cite{JSCC:Yang200904}, a variant of the weight-distribution
approach, e.g., \cite{JSCC:Divsalar199809, JSCC:MacWilliams197211}.

Let $\mathcal{X}$ and $\mathcal{Y}$ be two finite (additive) abelian
groups.  A \emph{linear encoder} is a homomorphism $f: \mathcal{X}^n
\to \mathcal{Y}^m$.
The image $f(\mathcal{X}^n)\subseteq\mathcal{Y}^m$ of such an $f$ is a
\emph{linear code} over $\mathcal{Y}$ in the usual sense, i.e., a
block code of length $m$ over $\mathcal{Y}$ that forms a subgroup of
$\mathcal{Y}^m$.%
\footnote{In the special case $\mathcal{X}=\mathcal{Y}=(\field{q},+)$,
where this concept refers to linearity over the prime field
$\field{p}$ of $\field{q}$, we rather speak of \emph{additive
encoders/codes}, so that there is no conflict with the stronger
concept of $\field{q}$-linearity.}
With each $f$ there are associated three kinds of rates:
First, the \emph{source transmission rate}
$R_s(f) \eqdef n^{-1}\ln|f(\mathcal{X}^n)|$.
Second, the \emph{channel transmission rate}
$R_c(f) \eqdef m^{-1}\ln|f(\mathcal{X}^n)|$.
Third, the coding rate $R(f)\eqdef n / m$.
There is a simple relation among these quantities, viz.,
$R(f)R_s(f) = R_c(f)$.
\omitted{%
If $f$ is injective, then $R_s(f) = \ln |\mathcal{X}|$ and
$R_c(f) = R(f) \ln |\mathcal{X}|$; if $f$ is surjective, then
$R_s(f) = \ln |\mathcal{Y}|/R(f)$ and $R_c(f) = \ln |\mathcal{Y}|$.}

For any linear encoder $f:\mathcal{X}^n\to\mathcal{Y}^m$, there exist
uniquely determined homomorphisms $f_{ij}:\mathcal{X}\to \mathcal{Y}$
($1\leq i\leq n$, $1\leq j\leq m$) such that
$f(\seq{x})
 =\left(\sum_{i=1}^nx_if_{i1},\dots,\sum_{i=1}^nx_if_{im}\right)
 =\seq{x}\mat{M}$,
where $x_if_{ij}\eqdef f_{ij}(x_i)$ and
$\mat{M}
 \eqdef (f_{ij})\in\mathrm{Hom}(\mathcal{X},\mathcal{Y})^{n\times m}$,
with $\mathrm{Hom}(\mathcal{X},\mathcal{Y})$ denoting the abelian
group of all homomorphisms from $\mathcal{X}$ to $\mathcal{Y}$ under
map addition.
In the special case of a prime field
$\mathcal{X}=\mathcal{Y}=\field{p}$, the usual representation of an
$\field{p}$-linear encoder $f:\field{p}^n\to\field{p}^m$ by its
generator matrix $\mat{M}\in\field{p}^{n\times m}$ is recovered, since
$\field{p}\cong\mathrm{End}(\field{p},+)
 \eqdef\mathrm{Hom}(\field{p},\field{p})$ via
$a\mapsto(\mu_a:\field{p}\to\field{p}$, $x\mapsto ax$).
For a general finite field $\field{q}$, the generator matrix
representation requires an $\field{q}$-linear encoder and is stronger.
Nevertheless, identifying a linear encoder with a generator matrix
over an appropriate finite field is still useful in general, since
every linear encoder with $\mathcal{X}=\mathcal{Y}=\field{q}$ has a
generator matrix representation over a certain subfield of
$\field{q}$.

A particularly simple class of linear encoders from $\mathcal{X}^n$ to
$\mathcal{X}^n$, including the identity map, is formed by the special
group automorphisms defined by
$\sigma(\seq{x}) \eqdef x_{\sigma^{-1}(1)} x_{\sigma^{-1}(2)} \cdots
 x_{\sigma^{-1}(n)}$
for each permutation $\sigma$ in $\symmetricgroup{n}$, the group of
all permutations of $\mathcal{I}_n$.
These linear encoders are called coordinate permutations or
interleavers.%
\footnote{For convenience, we shall slightly abuse the term
``permutation'' to refer to an induced coordinate permutation as long
as the exact meaning is clear from the context.}
Considering a random encoder uniformly distributed over all
permutations in $\symmetricgroup{n}$, we obtain a uniform random
permutation, denoted $\Sigma_n$.
We tacitly assume that different random permutations occurring in the
same expression are independent, and notation
such as $\Sigma_m$ and $\Sigma_n$ refers to
different random permutations even in the case $m = n$.

The \emph{type} of a sequence $\seq{x}$ in $\mathcal{X}^n$ is the
empirical distribution $P_{\seq{x}}$ on $\mathcal{X}$ defined by
\[
P_{\seq{x}}(a) \eqdef \frac{1}{|\seq{x}|}
\sum_{i=1}^{|\seq{x}|} 1\{x_i = a\}.
\]
For a (probability) distribution $P$ on $\mathcal{X}$, the set of
sequences of type $P$ in $\mathcal{X}^n$ is denoted by
$\mathcal{T}_P^n(\mathcal{X})$ or simply $\mathcal{T}_P^n$.
A distribution $P$ on $\mathcal{X}$ is called a type of sequences in
$\mathcal{X}^n$ if $\mathcal{T}_P^n \ne \varnothing$.
We denote by $\mathcal{P}_n(\mathcal{X})$ (or $\mathcal{P}_n$ if the
alphabet is clear from the context) the set of all types of sequences
in $\mathcal{X}^n$.
Since the set $\mathcal{P}_n\setminus\{P_{0^n}\}$ will be frequently
used, we denote it by $\mathcal{P}_n^*$ in the sequel.

The \emph{spectrum} of a nonempty set $A \subseteq \mathcal{X}^n$ is
the empirical distribution $\spec_{\mathcal{X}}(A)$ on
$\mathcal{P}_n(\mathcal{X})$ defined by
\[
\spec_{\mathcal{X}}(A)(P) \eqdef
\frac{|\{\seq{x} \in A : P_{\seq{x}} = P\}|}{|A|}
\qquad \forall P \in \mathcal{P}_n(\mathcal{X})
\]
and for convenience, we write $\spec(A)$, $\spec(A)(P)$, or further
$\spec_A(P)$ provided $P$ refers to an element of
$\mathcal{P}_n(\mathcal{X})$.
In other words, $\spec(A)$ is the empirical distribution of types of
sequences in $A$.
The spectrum of $A$ is closely related to the well-established
\emph{complete weight distribution} of $A$ (see e.g.,
\cite[Ch.~7.7]{JSCC:Huffman200300} or
\cite[Sec.~10]{JSCC:Pless199800}).
In fact, both distributions differ only by a scaling factor.
Nevertheless, introducing the spectrum as a new, independent concept
has its merits, see \cite{JSCC:Yang200909} for a detailed explanation.

Analogously, the \emph{joint spectrum}
$\spec_{\mathcal{X}\mathcal{Y}}(B)(P,Q)$ of a nonempty set
$B \subseteq \mathcal{X}^n \times \mathcal{Y}^m$ is the empirical
distribution of type pairs $(P_{\seq{x}},P_{\seq{y}})$ of sequence
pairs $(\seq{x},\seq{y})\in B$.
By considering the marginal and conditional distributions of
$\spec_{\mathcal{X}\mathcal{Y}}(B)$, we obtain the
\emph{marginal spectra} $\spec_{\mathcal{X}}(B)(P)$,
$\spec_{\mathcal{Y}}(B)(Q)$ and the \emph{conditional spectra}
$\spec_{\mathcal{Y}|\mathcal{X}}(B)(Q|P)$,
$\spec_{\mathcal{X}|\mathcal{Y}}(B)(P|Q)$.
These definitions can also be extended to the case of more than two
alphabets in the obvious way.
For convenience of notation, we sometimes write, e.g., $\spec(B)(P,Q)$,
$\spec(B)(P)$, $\spec(B)(Q|P)$, or further, $\spec_B(P,Q)$,
$\spec_B(P)$, $\spec_B(Q|P)$, provided $(P,Q)$ refers to an element of
$\mathcal{P}_n(\mathcal{X})\times\mathcal{P}_m(\mathcal{Y})$.

Furthermore, for any given function $f: \mathcal{X}^n \to
\mathcal{Y}^m$, we can define its \emph{joint spectrum}
$\spec_{\mathcal{X}\mathcal{Y}}(f)$,
\emph{forward conditional spectrum}
$\spec_{\mathcal{Y}|\mathcal{X}}(f)$,
and \emph{image spectrum} $\spec_{\mathcal{Y}}(f)$ as
$\spec_{\mathcal{X}\mathcal{Y}}(\mathrm{rl}(f))$,
$\spec_{\mathcal{Y}|\mathcal{X}}(\mathrm{rl}(f))$,
and $\spec_{\mathcal{Y}}(\mathrm{rl}(f))$, respectively, where
$\mathrm{rl}(f) \eqdef \{(\seq{x}, f(\seq{x})) : \seq{x}
\in \mathcal{X}^n\}$ is the graph of $f$.
In this case, the forward conditional spectrum $\spec(f)(Q|P)$ or
$\spec_f(Q|P)$ is given by $\spec_f(P, Q)/\spec_{\mathcal{X}^n}(P)$.
If $f$ is a linear encoder, we further define its \emph{kernel
spectrum} as $\spec(\ker f)$, where
$\ker f \eqdef \{\seq{x} \in \mathcal{X}^n : f(\seq{x}) = 0^m\}$, and
in this case, $\spec_\mathcal{Y}(f) = \spec(f(\mathcal{X}^n))$.

It is easy to see that coordinate permutations preserve the type and
hence the spectrum.
Two sets $A, B \subseteq \mathcal{X}^n$ are said to be
\emph{equivalent} (under permutation) if $\sigma(A) = B$ for some
$\sigma \in \symmetricgroup{n}$.
Two maps $f, g$ from $\mathcal{X}^n$ to $\mathcal{Y}^m$ are said to be
\emph{equivalent} (under permutation) if
$\sigma' \circ f \circ \sigma = g$ for some
$\sigma \in \symmetricgroup{n}$ and $\sigma' \in \symmetricgroup{m}$.
The notion of equivalence is extended to random sets and maps in the
obvious way.

For a random nonempty set $A\subseteq \mathcal{X}^n$, we define
$\alpha(A)(P)$ or
$\alpha_A(P) \eqdef \av[\spec_A(P)]/\spec_{\mathcal{X}^n}(P)$.
To simplify notation, we shall write, e.g., $\avS_A(P)$ in place of
$\av[\spec_A(P)]$.
Similarly, for a random map $F: \mathcal{X}^n \to \mathcal{Y}^m$, we
define
\begin{equation}\label{eq:DefinitionOfAlpha}
\alpha_F(P,Q)
\eqdef
 \frac{\avS_F(P,Q)}{\spec_{\mathcal{X}^n\times\mathcal{Y}^m}(P,Q)}.
\end{equation}
The definition of $\alpha$ is essentially a ratio of two spectra.
Its purpose is to measure the distance from the spectrum of a set to a
``random-like'' spectrum (e.g., $\spec_{\mathcal{X}^n}(P)$).
The reader may also compare it or its logarithm to the notion of
information divergence.
It is clear that equivalent sets or maps have the same $\alpha$.
As we shall see, $\alpha$ plays an important role in characterizing
the average behavior of an equivalence class of sets or maps
(Proposition~\ref{pr:GeneralSpectrumPropertyOfSets} as well as
\citeSpectrumProperty), and hence it provides a good criterion of code
performance for various applications
(Section~\ref{sec:ConceptsOfGoodLinearCodes}).

\section{Good Linear Encoders for Source Coding, Channel Coding, and
Joint Source-Channel Coding}\label{sec:ConceptsOfGoodLinearCodes}

In classical coding theory, a good linear code typically refers to a
set of codewords that has good performance for some family of channels
or has a large minimum Hamming distance close to one of the well-known
upper bounds such as the sphere packing, Plotkin, or Singleton bound.
Such a viewpoint may be sufficient for channel coding, but has its
limitations in the context of source coding and JSCC.
The main reason is that an approach focusing on linear codes (i.e., the
image of encoder map) cannot cover all coding-related properties of
linear encoders, and relevant criteria for good linear encoders are
generally different for different applications.

In \cite[Table~I]{JSCC:Yang200904}, we have only briefly reviewed the
criteria of good linear encoders in terms of spectrum requirements for
lossless source coding, channel coding, and lossless JSCC.
So in this section, we shall resume this discussion, including the
concepts of good linear encoders and the relations among different
kinds of good linear encoders.
Since our main constructions require the underlying alphabet to be a
finite field, we assume from now on (unless stated otherwise) that the
alphabet of any linear encoder is the finite field $\field{q}$, where
$q = p^r$ and $p$ is prime.

We begin with some concepts related to the asymptotic rate of an
encoder sequence.
Let
$\seqx{F} = \{F_k: \field{q}^{n_k} \to \field{q}^{m_k}\}_{k=1}^\infty$
be a sequence of random linear encoders.
If $R_s(F_k)$ converges in probability to a constant, we say that the
\emph{asymptotic source transmission rate} $R_s(\seqx{F})$ of
$\seqx{F}$ is $\plim_{k \to \infty} R_s(F_k)$.
Analogously, we define the \emph{asymptotic channel transmission rate}
and the \emph{asymptotic coding rate} of $\seqx{F}$ by
$R_c(\seqx{F}) \eqdef \plim_{k \to \infty} R_c(F_k)$ and
$R(\seqx{F}) \eqdef \lim_{k \to \infty} R(F_k)$, respectively.
When $R(F_k)$ does not necessarily converge, we define the
\emph{superior asymptotic coding rate} $\overline{R}(\seqx{F})$ and
\emph{inferior asymptotic coding rate} $\underline{R}(\seqx{F})$ by
taking the limit superior and limit inferior, respectively.
To simplify notation, in the rest of this section, when writing
$\seqx{F}$, we always mean a sequence $\{F_k\}_{k=1}^\infty$ of random
linear encoders $F_k: \field{q}^{n_k} \to \field{q}^{m_k}$.
To avoid some degenerate cases, we assume that
$\lim_{k \to \infty} |F_k(\field{q}^{n_k})| = \infty$.

Next, we introduce the definitions of good linear encoders for
lossless source coding, channel coding, and lossless JSCC,
respectively.
The rationale behind them is explained in \cite{JSCC:Yang200909} using
the ideas of
\cite{JSCC:Yang200503,JSCC:Gallager196800,JSCC:Shulman199909,
 JSCC:Bennatan200403,JSCC:Yang200904}.

\begin{definition}\label{df:KernelSpectrumCondition}
Let $\seqx{F}$ be a sequence of random linear encoders with the
asymptotic source transmission rate $R_s(\seqx{F})$.
If $\seqx{F}$ satisfies the \emph{kernel-spectrum condition}:
\begin{equation}\label{eq:DefinitionOfAsympGoodLSC}
\limsup_{k \to \infty} \max_{P \in \mathcal{P}_{n_k}^*}
 \frac{1}{n_k} \ln \alpha_{\ker F_k}(P)
\le \delta,
\end{equation}
then it is called a sequence of $\delta$-asymptotically good
\emph{linear source encoders} (LSEs) or is said to be
$\delta$-asymptotically SC-good (where the last ``C'' stands for
``coding'').
\end{definition}

\begin{definition}\label{df:ImageSpectrumCondition}
Let $\seqx{F}$ be a sequence of random linear encoders with the
asymptotic channel transmission rate $R_c(\seqx{F})$.
If $\seqx{F}$ satisfies the \emph{image-spectrum condition}:
\begin{equation}\label{eq:DefinitionOfAsympGoodLCC}
\limsup_{k \to \infty} \max_{Q \in \mathcal{P}_{m_k}^*}
 \frac{1}{m_k} \ln \alpha_{F_k(\field{q}^{n_k})}(Q)
\le \delta,
\end{equation}
then it is called a sequence of $\delta$-asymptotically good
\emph{linear channel encoders} (LCEs) or is said to be
$\delta$-asymptotically CC-good.
\end{definition}

\begin{definition}\label{df:JointSpectrumCondition}
Let $\seqx{F}$ be a sequence of random linear encoders.
If $\seqx{F}$ satisfies the \emph{joint-spectrum condition}:
\begin{equation}\label{eq:DefinitionOfAsympGoodLSCC}
\limsup_{k \to \infty}
 \max_{P \in \mathcal{P}_{n_k}^*, Q \in \mathcal{P}_{m_k}}
 \frac{1}{n_k} \ln \alpha_{F_k}(P, Q)
\le \delta,
\end{equation}
then it is called a sequence of $\delta$-asymptotically good
\emph{linear source-channel encoders} (LSCEs) or is said to be
$\delta$-asymptotically SCC-good.
\end{definition}

When $\delta = 0$, we use the simplified term ``asymptotically good'';
when talking about $\delta$-asymptotically good LSEs (resp., LCEs), we
tacitly assume that their asymptotic source (resp., channel)
transmission rates exist.
Since the conditions for $\delta$-asymptotically good LSEs or LCEs
only depend on the kernel or image of the linear encoders involved, we
also introduce the concepts of equivalence in the sense of LSE or
LCE.
Linear encoders $f: \field{q}^n \to \field{q}^m$ and $g:
\field{q}^n \to \field{q}^m$ are said to be \emph{SC-equivalent}
if their kernels are equivalent, and
\emph{CC-equivalent} if their images are equivalent.
For convenience, we define the function
\[
\rho(F)
\eqdef \max_{P \in \mathcal{P}_n^*, Q \in \mathcal{P}_{m}}
 \frac{1}{n} \ln \alpha_F(P, Q)
\]
of a random linear encoder $F: \field{q}^n \to \field{q}^{m}$.
Then condition~\eqref{eq:DefinitionOfAsympGoodLSCC} can be rewritten
as $\limsup_{k \to \infty} \rho(F_k) \le \delta$.

The following propositions relate these three kinds of good linear
encoders to each other.

\begin{proposition}\label{pr:RelationBetweenLSCandLCC}
For a sequence $\seqx{F}$ of $\delta$-asymptotically good LSEs, there
exists a sequence $\seqx{G} = \{G_k\}_{k=1}^\infty$ of
$\delta$-asymptotically good LCEs
$G_k: \field{q}^{l_k} \to \field{q}^{n_k}$ such that
$G_k(\field{q}^{l_k}) = \ker F_k$ for all $k \in \pintegers$.
\end{proposition}

\begin{proposition}\label{pr:RelationBetweenLCCandLSC}
For a sequence $\seqx{F}$ of $\delta$-asymptotically good LCEs, there
exists a sequence $\seqx{G} = \{G_k\}_{k=1}^\infty$ of
$\delta$-asymptotically good LSEs
$G_k: \field{q}^{m_k} \to \field{q}^{l_k}$ such that
$\ker G_k = F_k(\field{q}^{n_k})$ for all $k \in \pintegers$.
\end{proposition}

\begin{proposition}\label{pr:RelationBetweenLSCCandLSCLCC}
Let $\seqx{F}$ be a sequence of $\delta$-asymptotically good LSCEs.
Then $\seqx{F}$ is $\delta$-asymptotically SC-good whenever its
asymptotic source transmission rate exists, and $\seqx{F}$ is
$\delta\overline{R}(\seqx{F})$-asymptotically CC-good whenever its
asymptotic channel transmission rate exists.
\end{proposition}

These relations are all depicted in Fig.~\ref{fig:Relation1}.
\begin{figure}[htbp]
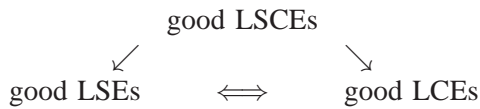

\[
\begin{array}{rcl}
&\mbox{good LSCEs} \\
\swarrow & &\searrow \\
\mbox{good LSEs} &\Longleftrightarrow &\mbox{good LCEs} \\
\end{array}
\]
\caption{The relations among different kinds of good linear encoders}
\label{fig:Relation1}
\end{figure}
Among them, the relations
(Proposition~\ref{pr:RelationBetweenLSCCandLSCLCC}) that good LSCEs
are good LSEs and LCEs deserve much more attention.  This fact
indicates the fundamental role of good LSCEs, and allows us to
concentrate on only one problem, namely, constructing linear encoders
with good joint spectra.
So in the following sections, we shall investigate this problem in
depth.
Note, however, that there are two arrows missing in
Fig.~\ref{fig:Relation1}, one from LSE to LSCE and the other from LCE
to LSCE.
This naturally leads to the following question: Are
$\delta$-asymptotically good LSEs (resp., LCEs)
$\delta$-asymptotically SCC-good, or can we construct
$\delta$-asymptotically good LSCEs which are SC-equivalent (resp.,
CC-equivalent) to given $\delta$-asymptotically good LSEs (resp.,
LCEs)?

As $\delta$ is close to zero, it is clear that $\delta$-asymptotically
good LSEs (or LCEs) may not be $\delta$-asymptotically SCC-good.
\cite[Theorem 4]{JSCC:Yang200503} shows that there exist
asymptotically good LSEs in the ensemble of low-density parity-check
matrices.
However, \cite[Corollary~4.2]{JSCC:Yang200904} proves that these
encoders are not asymptotically SCC-good because the matrices are
sparse.
It is also well known that for any linear code there exists a
CC-equivalent systematic linear encoder.
According to \cite[Corollary~4.1]{JSCC:Yang200904}, high-rate
systematic linear encoders are not asymptotically SCC-good.
So we can always find asymptotically good systematic LCEs which are
not asymptotically SCC-good.

It remains yet to decide whether, and if so, how we can construct
$\delta$-asymptotically good LSCEs which are SC-equivalent (resp.,
CC-equivalent) to given $\delta$-asymptotically good LSEs (resp.,
LCEs).
The answer is positive and will be given in
Section~\ref{sec:GeneralPrinciplesForConstructingGoodLSCC}.

We close this section with an easy result.

\begin{proposition}\label{pr:DiagonalArgument}
Let $\{\seqx{F}_i = \{F_{i,k}\}_{k=1}^\infty\}_{i=1}^\infty$ be a
family of sequences of random linear encoders $F_{i,k}:
\field{q}^{n_k} \to \field{q}^{m_k}$.
Suppose that $\seqx{F}_i$ is $\delta_i$-asymptotically SCC-good where
$\delta_i$ is nonincreasing in $i$ and converges to $\delta$ as
$i\to\infty$.
Define the random linear encoder $G_{i,k} \eqdef F_{i_0,k}$ where
$i_0 \eqdef \arg \min_{1 \le j \le i} \rho(F_{j,k})$.
Then $\{G_{k,k}\}_{k=1}^\infty$ is $\delta$-asymptotically SCC-good.
\end{proposition}

The proofs of this section are given in
Appendix~\ref{subsec:ProofOfConceptsOfGoodLinearCodes}.

\section{General Principles for Constructing Good Linear
 Source-Channel Encoders}
\label{sec:GeneralPrinciplesForConstructingGoodLSCC}

\subsection{A Class of SCC-Good Encoders Derived from Certain
 Maximum-Rank-Distance Codes}\label{subsec:MRDLSCC}

Recall that in \cite[Sec.~III]{JSCC:Yang200904}, a random linear
encoder $F: \field{q}^n \to \field{q}^m$ was said to be good for JSCC
if it satisfies
\begin{equation}\label{eq:SuperGoodLinearCodes1}
\alpha_F(P, Q) = 1
\end{equation}
for all $P \in \mathcal{P}_n^*$ and $Q \in \mathcal{P}_m$.
To distinguish good linear encoders for JSCC from good linear encoders
in other contexts, we say that $F$ is SCC-good.
It is clear that SCC-good linear encoders are asymptotically SCC-good.
By \citeSpectrumProperty, we also have an alternative condition:
\begin{equation}\label{eq:SuperGoodLinearCodes2}
\pr\{\tilde{F}(\seq{x}) = \seq{y}\} = q^{-m}
\quad \forall \seq{x} \in \field{q}^n \setminus \{0^n\},
 \seq{y} \in \field{q}^m.
\end{equation}

Let $\rlccode{n,m}: \field{q}^n\to\field{q}^m$ be a random linear
encoder uniformly distributed over $\field{q}^{n\times m}$.
Obviously, $\rlccode{n,m}$ is SCC-good (\citeGoodLinearCode) but in
some sense trivial, since its distribution has support
$\field{q}^{n \times m}$, the set of all linear encoders
$f:\field{q}^n\to\field{q}^m$.
In this subsection we provide further examples of SCC-good random
linear encoders with support size much smaller than
$|\field{q}^{n\times m}|=q^{mn}$.
These are derived from so-called maximum-rank-distance (MRD) codes,
and have the stronger property
\begin{equation}\label{eq:SuperGoodBeforeSymmLinearCodes}
\pr\{F(\seq{x}) = \seq{y}\} = q^{-m}
\quad \forall \seq{x} \in \field{q}^n \setminus \{0^n\},
 \seq{y} \in \field{q}^m
\end{equation}
which one might call ``SCC-good before symmetrization''.

We first give a brief review of MRD codes.
Let $n,m,k$ be positive integers with $k\leq\min\{n,m\}$.
Write $n'=\max\{n,m\}$, $m'=\min\{n,m\}$ (so that $(n',m')$ equals
$(n,m)$ or $(m,n)$).
An \emph{$(n,m,k)$ maximum-rank-distance (MRD) code} over $\field{q}$
is a set $\mathcal{C}$ of $q^{kn'}$ matrices in
$\field{q}^{n\times m}$ with minimum rank distance
$\rdist(\mathcal{C})
 \eqdef\min_{\mat{A},\mat{B}\in\mathcal{C}:\mat{A}\neq\mat{B}}
 \rank{\mat{A}-\mat{B}}
 =m'-k+1$.
\omitted{%
MRD codes are optimal for the rank distance on $\field{q}^{n\times m}$
in the same way as maximum distance separable (MDS) codes are optimal
for the Hamming distance, as the following well-known argument shows:
Let $\mathcal{C}\subseteq\field{q}^{n\times m}$ have minimum rank
distance $d$, $1\leq d\leq m'$.
We may view $\mathcal{C}$ as a code of length $m$ over $\field{q}^n$
(the columns of $\mat{A}\in\field{q}^{n\times m}$ being the
``entries'' of a codeword).
As such $\mathcal{C}$ has Hamming distance $\geq d$ (since a matrix of
rank $d$ must have at least $d$ nonzero columns).
Hence $|\mathcal{C}|\leq q^{n(m-d+1)}$ by the Singleton bound.
By transposing we also have $|\mathcal{C}|\leq q^{m(n-d+1)}$, so that
\begin{equation}\label{eq:MRDsingleton}
|\mathcal{C}|
\leq\min\left\{q^{n(m-d+1)},q^{m(n-d+1)}\right\}
=q^{n'(m'-d+1)}.
\end{equation}
The codes for which the Singleton-like bound \eqref{eq:MRDsingleton}
is sharp are exactly the $(n,m,k)$ MRD codes with $k=m'-d+1$.%
}

MRD codes have maximum cardinality among all $n\times m$ matrix codes
over $\field{q}$ with fixed minimum rank distance.
They were introduced in \cite{JSCC:Delsarte197800} under the name
``Singleton system'' and investigated further in
\cite{JSCC:Gabidulin198501,JSCC:Roth199103}.  As shown in
\cite{JSCC:Delsarte197800,JSCC:Gabidulin198501,JSCC:Roth199103},
linear $(n,m,k)$ MRD codes over $\field{q}$ (i.e.\ MRD codes that are
$\field{q}$-subspaces of $\field{q}^{n\times m}$) exist for all
$n,m,k$ with $1\leq k\leq m'$.
The standard construction uses $q$-analogues of Reed-Solomon codes,
which are defined as follows:
Assuming $n\geq m$ for a moment, let $x_1,\dots,x_m\in\field{q^n}$ be
linearly independent over $\field{q}$.
For $1\leq k\leq m$ let $C_k$ be the linear $[m,k]$ code over
$\field{q^n}$ having generator matrix
\begin{equation}\label{eq:qRS}
\mat{G}_k=
\begin{pmatrix}
 x_1&x_2&\hdots&x_m\\
 x_1^q&x_2^q&\hdots&x_m^q\\
 \vdots&\vdots&&\vdots\\
 x_1^{q^{k-1}}&x_2^{q^{k-1}}&\hdots&x_m^{q^{k-1}}
\end{pmatrix}.
\end{equation}
Replace each codeword $\seq{c}=(c_1,\dots,c_m)\in C_k$ by the $n\times
m$ matrix $\mat{C}$ having as columns the coordinate vectors of $c_j$
with respect to some fixed basis of $\field{q^n}$ over $\field{q}$.
The set $\mathcal{C}_k$ of all $q^{nk}$
matrices $\mat{C}$ obtained in this way forms a linear
$(n,m,k)$ MRD code over $\field{q}$.  The restriction $n\geq m$ is not
essential, since transposing each matrix of an $(n,m,k)$ MRD code
yields an $(m,n,k)$ MRD code (due to the fact that
$\mat{A}\mapsto\mat{A}^T$ preserves the rank distance).  We shall
follow \cite{JSCC:Silva200809} and call the codes
$\mathcal{C}_k\subseteq\field{q}^{n \times m}$ ($1\leq k\leq m\leq
n$), as well as their transposes, \emph{Gabidulin codes}.

For our construction of SCC-good random linear encoders we need the
following property of Gabidulin codes.

\begin{proposition}\label{pr:MRDLSCC}
Suppose $\mathcal{C}\subseteq\field{q}^{n\times m}$ is a Gabidulin
code.
Then $\{\seq{x}\mat{A}:\mat{A}\in\mathcal{C}\}=\field{q}^m$ for every
$\seq{x} \in \field{q}^n \setminus \{0^n\}$, and similarly
$\{\mat{A}\seq{y}^T:\mat{A}\in\mathcal{C}\}=\field{q}^n$ for every
$\seq{y} \in \field{q}^m \setminus \{0^m\}$.
\end{proposition}

\begin{theorem}\label{th:MRDLSCC}
Let $\mathcal{C}\subseteq\field{q}^{n\times m}$ be a Gabidulin code
and $F:\field{q}^n\to\field{q}^m$ a random linear encoder uniformly
distributed over $\mathcal{C}$.
Then $F$ is SCC-good (before symmetrization).
\end{theorem}

Theorem~\ref{th:MRDLSCC} provides us with SCC-good random linear
encoders $F:\field{q}^n\to\field{q}^m$ of support size as small as
$q^{n'}=q^{\max\{n,m\}}$, realized by an $(n,m,1)$ MRD code, which has
minimum rank distance $d = m' = \min\{n,m\}$ (the full-rank case).%
\footnote{It was shown in \cite{JSCC:Yang201105} that
 Theorem~\ref{th:MRDLSCC} actually holds for all MRD codes.
 Furthermore, it was proved that the minimum support size of an
 SCC-good-before-symmetrization random linear encoder is exactly
 $q^{\max\{n,m\}}$, and that a random linear encoder of support size
 $q^{\max\{n,m\}}$ is SCC-good before symmetrization if and only if it
 is uniformly distributed over an $(n,m,1)$ MRD code.}

\omitted{%
As an application of Theorem~\ref{th:MRDLSCC}, the next example gives
us some insights into the general form of a random linear encoder that
is SCC-good before symmetrization.

\begin{example}\label{ex:Gabidulin}
Let $m=n=q=2$ and $\field{4} = \field{2}[\alpha]$ with
$\alpha^2 = \alpha+1$.
It is clear that $\{1, \alpha\}$ is a basis of $\field{4}$.
By definition of Gabidulin codes, we may consider the generator matrix
$\mat{G}_1 = (1\;\; \alpha)$.
Then the Gabidulin code $\mathcal{C}_1$ is
\[
\left\{(0\;\; 0), (1\;\; \alpha), (\alpha\;\; 1+\alpha),
 (1+\alpha\;\; 1)\right\}
\]
or
\[
\left\{
\begin{pmatrix} 0 &0\\ 0 &0\end{pmatrix},
\begin{pmatrix} 1 &0\\ 0 &1\end{pmatrix},
\begin{pmatrix} 0 &1\\ 1 &1\end{pmatrix},
\begin{pmatrix} 1 &1\\ 1 &0\end{pmatrix}\right\}.
\]
Note that $\mathcal{C}_1$ is isomorphic to $\field{4}$.
By Theorem~\ref{th:MRDLSCC}, the random linear encoder uniformly
distributed over $\mathcal{C}_1$, which we also denote by
$\mathcal{C}_1$, is SCC-good.
Interchanging the columns of each matrix in $\mathcal{C}_1$ gives
another SCC-good random linear encoder
\[
\mathcal{C}_2 = \left\{
\begin{pmatrix} 0 &0\\ 0 &0\end{pmatrix},
\begin{pmatrix} 0 &1\\ 1 &0\end{pmatrix},
\begin{pmatrix} 1 &0\\ 1 &1\end{pmatrix},
\begin{pmatrix} 1 &1\\ 0 &1\end{pmatrix}\right\}.
\]
The cosets of $\mathcal{C}_1$, such as
$\mathcal{C}_3 = \mathcal{C}_1
 + (\begin{smallmatrix} 0 & 1 \\ 0 & 0 \end{smallmatrix})$,
are also SCC-good.
Note that $\mathcal{C}_3$ contains only one matrix of full rank.
Finally, let $\mathcal{C} \eqdef \mathcal{C}_I$, where $I$ is an
arbitrary random variable taking values in $\mathcal{I}_3$.
It is clear that $\mathcal{C}$ is SCC-good.
\end{example}
}

The proofs of this subsection are omitted since
Proposition~\ref{pr:MRDLSCC} and Theorem~\ref{th:MRDLSCC} are special
cases of \cite[Lemma~2.4 and Theorem~2.5]{JSCC:Yang201105},
respectively.

\subsection{Constructing Good LSCEs Based on Good LSEs or LCEs}
\label{subsec:ConstructingGoodLSCCGPA}

Condition \eqref{eq:SuperGoodLinearCodes2} is a very strong
condition, which in fact reflects the property of the alphabet.
Combining this with the injectivity property of mappings, we say
that an abelian group $\mathcal{X}$ is \emph{super good} if there
exists a sequence $\{F_n\}_{n=1}^\infty$ of SCC-good random linear
encoders $F_n: \mathcal{X}^n \to \mathcal{X}^n$ such that
\begin{equation}\label{eq:InjectiveProbabilityCondition}
\inf_{n \in \pintegers} \pr\{|\ker F_n| = 1\} > 0.
\end{equation}

We shall now prove that a finite abelian group is super good if
and only if it is \emph{elementary abelian}, i.e.\ isomorphic to
$\integers_p^s$ for some prime $p$ and some integer $s\geq 0$.
The nontrivial elementary abelian $p$-groups are exactly the additive
groups of the finite fields $\field{q}$ with $q=p^r$.
Hence choosing $\field{q}$ as alphabets incurs no essential loss of
generality, and we shall later on keep the assumption that the alphabet
is $\field{q}$.%
\footnote{A similar conclusion is valid in classical coding theory.
 It was shown in \cite{JSCC:Forney199211} that group codes over
 general groups can be no better than linear codes over finite fields,
 in terms of Hamming distance.}

The following fact \cite{JSCC:Cooper200010} shows that all elementary
abelian groups are super good.
\begin{equation}
\pr\left\{|\ker\rlccode{n,n}| = 1\right\}
= \prod_{i=1}^n (1-q^{-i})
> K_q,\label{eq:RankOfRLC}
\end{equation}
where $K_q \eqdef \prod_{i=1}^\infty (1-q^{-i}) > 1-q^{-1}-q^{-2}$
by Euler's pentagonal number theorem \cite{JSCC:Andrews198311}.

The next theorem shows that, conversely, a super good finite abelian
group is necessarily elementary abelian.

\begin{theorem}\label{th:elementary_abelian}
For nontrivial abelian groups $\mathcal{X}$ and $\mathcal{Y}$, if
there exists an SCC-good random linear encoder
$F: \mathcal{X}^n\to\mathcal{Y}^m$, then $\mathcal{X}$ and
$\mathcal{Y}$ are elementary abelian $p$-groups for the same prime $p$
(provided that $m,n\geq 1$).
\end{theorem}

Next let us investigate the relation between conditions
\eqref{eq:SuperGoodLinearCodes2} and
\eqref{eq:InjectiveProbabilityCondition} for elementary abelian
groups.

\begin{theorem}\label{th:KernelOfSuperGoodLinearCodes}
Let $\mathcal{X}$ be an elementary abelian group of order $q=p^r$.
Then for every SCC-good random linear encoder
$F:\mathcal{X}^n \to \mathcal{X}^n $ the following bound holds:
\begin{equation*}
\pr\{|\ker F| = 1\} \geq \frac{p-2+q^{-n}}{p-1}.
\end{equation*}
\end{theorem}

An immediate consequence follows.

\begin{corollary}\label{co:KernelOfSuperGoodLinearCodes}
Let $\mathcal{X}$ be an elementary abelian $p$-group for some prime
$p>2$.
Then there exists a positive constant $c(|\mathcal{X}|)$ (which may be
taken as $c(|\mathcal{X}|)= 1-\frac{1}{p-1}$) such that
$\pr\{|\ker F|=1\}\geq c(|\mathcal{X}|)$ for every SCC-good random
linear encoder $F: \mathcal{X}^n \to \mathcal{X}^n$.
\end{corollary}

However, the conclusion of
Corollary~\ref{co:KernelOfSuperGoodLinearCodes} does not hold for
$p=2$.

\begin{proposition}\label{pr:KernelOfSuperGoodLinearCodes}
If $\mathcal{X}$ is an elementary abelian $2$-group, there exists a
sequence $\{F_n\}_{n=1}^\infty$ of SCC-good random linear encoders
$F_n:\mathcal{X}^n\to\mathcal{X}^n$ such that
$\lim_{n\to\infty} \pr\{|\ker F_n|=1\}=0$.
\end{proposition}

With the preparations above, let us investigate the problem of
constructing $\delta$-asymptotically good LSCEs which are
SC-equivalent (resp., CC-equivalent) to given $\delta$-asymptotically
good LSEs (resp., LCEs).
The next theorem provides a method for constructing SC-equivalent or
CC-equivalent linear encoders.

\begin{theorem}\label{th:ConstructionOfEquivalentCodes}
Let $F: \field{q}^n \to \field{q}^m$ be a random linear encoder.
The linear encoder $G_1 \eqdef \rlccode{m,m} \circ F$ satisfies
\begin{equation}
\pr\{\ker G_1 = \ker F\} > K_q \label{eq:SCEquivalentProbability}
\end{equation}
and
\begin{equation}
\avS_{G_1}(Q|P)
\le \spec_{\field{q}^m}(Q) + 1\{Q = P_{0^{m}}\} \avS_F(P_{0^{m}}|P)
\quad
\forall P \in \mathcal{P}_n^*, Q \in \mathcal{P}_m.
\label{eq:SCEquivalentCodeSpectrum}
\end{equation}
The linear encoder $G_2 \eqdef F \circ \rlccode{n,n}$ satisfies
\begin{equation}
\pr\{G_2(\field{q}^n) = F(\field{q}^n)\} > K_q
\label{eq:CCEquivalentProbability}
\end{equation}
and
\begin{equation}
\avS_{G_2}(Q|P) = \avS_{F(\field{q}^n)}(Q)
\quad\forall P \in \mathcal{P}_n^*, Q \in \mathcal{P}_m.
\label{eq:CCEquivalentCodeSpectrum}
\end{equation}
\end{theorem}

Based on Theorem~\ref{th:ConstructionOfEquivalentCodes}, we have the
following answer to the problem:

\begin{theorem}\label{th:ConstructionOfGoodEquivalentCodes}
Let $\seqx{f} = \{f_k\}_{k=1}^\infty$ be a sequence of linear encoders
$f_k: \field{q}^{n_k} \to \field{q}^{m_k}$.
If $\seqx{f}$ is a sequence of $\delta$-asymptotically good LSEs such
that $R_c(\seqx{f}) = \ln q$, then there exists a sequence
$\{g_{1,k}\}_{k=1}^\infty$ of $\delta$-asymptotically good LSCEs
$g_{1,k}: \field{q}^{n_k} \to \field{q}^{m_k}$ such that $g_{1,k}$ is
SC-equivalent to $f_k$ for each $k \in \pintegers$.
Analogously, if $\seqx{f}$ is a sequence of $\delta$-asymptotically
good injective LCEs, then there exists a sequence
$\{g_{2,k}\}_{k=1}^\infty$ of $\delta/R(\seqx{f})$-asymptotically good
injective LSCEs $g_{2,k}: \field{q}^{n_k} \to \field{q}^{m_k}$ such
that $g_{2,k}$ is CC-equivalent to $f_k$ for each $k \in \pintegers$.
\end{theorem}

\omitted{%
Theorem~\ref{th:ConstructionOfGoodEquivalentCodes} is a fundamental
result, which not only claims the existence of SC-equivalent (or
CC-equivalent) $\delta$-asymptotically good LSCEs but also paves the
way for constructing such good LSCEs by concatenating rate-$1$ linear
encoders.
Since rate-$1$ linear codes (e.g., the ``accumulate'' code) are
frequently used to construct good LCEs (e.g.,
\cite{JSCC:Divsalar199809,JSCC:Pfister200306,JSCC:Abbasfar200704}), we
believe that finding good rate-$1$ LSCEs is an issue deserving further
consideration.
}

The proof of Corollary~\ref{co:KernelOfSuperGoodLinearCodes} are
omitted, while the proofs of the other results are given in
Appendix~\ref{subsec:ProofOfConstructingGoodLSCCGPA}.

\subsection{A General Scheme for Constructing Good LSCEs}
\label{subsec:ConstructingGoodLSCCGPB}

Theorems~\ref{th:MRDLSCC} and
\ref{th:ConstructionOfGoodEquivalentCodes} do give possible ways for
constructing asymptotically good LSCEs.
However, such constructions are somewhat difficult to implement in
practice, because the generator matrices of $\rlccode{n,n}$ and random
linear encoders derived from Gabidulin codes are not sparse.
Thus, our next question is how to construct $\delta$-asymptotically
good LSCEs based on sparse matrices so that known iterative decoding
procedures have low complexity.
For such purposes, in this subsection, we shall present a general
scheme for constructing $\delta$-asymptotically good LSCEs.

Let $\seqx{F} = \{F_k\}_{k=1}^\infty$ be a sequence of random linear
encoders $F_k: \field{q}^{n_k} \to \field{q}^{m_k}$.
We say that $\seqx{F}$ is $\delta$-asymptotically SCC-good relative to
the sequence $\{A_k\}_{k=1}^\infty$ of subsets
$A_k \subseteq \mathcal{P}_{n_k}^*$ if
\[
\limsup_{k \to \infty} \max_{P \in A_k, Q \in \mathcal{P}_{m_k}}
 \frac{1}{n_k} \ln \alpha_{F_k}(P, Q)
\le \delta.
\]
Clearly, this is a generalization of $\delta$-asymptotically good
LSCEs, and may be regarded as an approximate version of
$\delta$-asymptotically good LSCEs when $A_k$ is a proper subset of
$\mathcal{P}_{n_k}^*$.
The next theorem shows that $\delta$-asymptotically good LSCEs based
on these linear encoders may be constructed by serial concatenations.

\begin{theorem}\label{th:ConstructionOfGoodLSCC}
Let $\{G_k\}_{k=1}^\infty$ be a sequence of random linear encoders
$G_k: \field{q}^{m_k} \to \field{q}^{l_k}$ that is
$\delta$-asymptotically SCC-good relative to the sequence
$\{A_k\}_{k=1}^\infty$ of sets $A_k \subseteq \mathcal{P}_{m_k}^*$.
If there is a sequence $\seqx{F} = \{F_k\}_{k=1}^\infty$ of random
linear encoders $F_k: \field{q}^{n_k} \to \field{q}^{m_k}$ such that
\begin{equation}\label{eq:ConstructionOfGoodLSCC1}
F_k(\field{q}^{n_k} \backslash \{0^{n_k}\})
\subseteq \bigcup_{P \in A_k} \mathcal{T}_{P}^{m_k},
\end{equation}
then
\[
\limsup_{k \to \infty} \rho(G_k \circ \Sigma_{m_k} \circ F_k)
\le \frac{\delta}{\underline{R}(\seqx{F})}.
\]
\end{theorem}

\begin{remark}\label{re:ConstructionOfGoodLSCC}
Take, for example,
\begin{equation}\label{eq:ApproximateGoodSet}
A_k = \{P \in \mathcal{P}_{m_k} : P(0) \in [0, 1-\gamma]\}
\end{equation}
for some $\gamma\in(0,1)$.
Then Theorem~\ref{th:ConstructionOfGoodLSCC} shows that we can
construct asymptotically good LSCEs using a serial concatenation
scheme, where the inner encoder is approximately
$\delta$-asymptotically SCC-good and the outer code has large minimum
distance.
As we know (see
\cite{JSCC:Gallager196300, JSCC:Bennatan200403, JSCC:Como2008,
 JSCC:Yang201111}, etc.),
there exist good LDPC codes over finite fields such that
\eqref{eq:ConstructionOfGoodLSCC1} is met for an appropriate $\gamma$,
so the problem is reduced to finding a sequence of linear encoders
that is $\delta$-asymptotically good relative to a sequence of sets
such as \eqref{eq:ApproximateGoodSet}.
In the next section, we shall find such candidates in a family of
encoders called LDGM encoders.
\end{remark}

Since an injective linear encoder $F_k$ always satisfies condition
\eqref{eq:ConstructionOfGoodLSCC1} with $A_k = \mathcal{P}_{n_k}^*$,
we immediately obtain the following corollary from
Theorem~\ref{th:ConstructionOfGoodLSCC}.

\begin{corollary}\label{co:ConstructionOfGoodLSCCbyInjectiveMapping}
Let $\{G_k\}_{k=1}^\infty$ be a sequence of $\delta$-asymptotically
good random LSCEs $G_k: \field{q}^{m_k} \to \field{q}^{l_k}$ and
$\seqx{F} = \{F_k\}_{k=1}^\infty$ a sequence of injective random
linear encoders $F_k: \field{q}^{n_k} \to \field{q}^{m_k}$.
Then
\[
\limsup_{k \to \infty} \rho(G_k \circ \Sigma_{m_k} \circ F_k)
\le \frac{\delta}{\underline{R}(\seqx{F})}.
\]
\end{corollary}

In the same vein, we have:

\begin{proposition}
\label{pr:ConstructionOfGoodLSCCbySurjectiveMapping}
Let $\{F_k\}_{k=1}^\infty$ be a sequence of $\delta$-asymptotically
good random LSCEs $F_k: \field{q}^{n_k} \to \field{q}^{m_k}$ and
$\{G_k\}_{k=1}^\infty$ a sequence of surjective random linear
encoders $G_k: \field{q}^{m_k} \to \field{q}^{l_k}$.
Then
\[
\limsup_{k \to \infty} \rho(G_k \circ \Sigma_{m_k} \circ F_k)
\le \delta.
\]
\end{proposition}

\omitted{%
The above two results tell us that any linear encoder, if serially
concatenated with an outer injective linear encoder or an inner
surjective linear encoder, will not have worse performance in terms of
condition \eqref{eq:DefinitionOfAsympGoodLSCC}.
Maybe, if we are lucky, some linear encoders with better joint spectra
can be constructed in this way.
Note that a nonsingular rate-$1$ linear encoder is both injective and
surjective, so adding rate-$1$ linear encoders into a serial
concatenation scheme is never a bad idea for constructing good LSCEs.
But certainly, the addition of rate-$1$ encoders may have a negative
impact on the decoding performance.
}

The proof of
Corollary~\ref{co:ConstructionOfGoodLSCCbyInjectiveMapping} is
omitted, while the proofs of the other results are given in
Appendix~\ref{subsec:ProofOfConstructingGoodLSCCGPB}.

\section{An Explicit Construction Based on Sparse Matrices}
\label{sec:ExplicitConstruction}

In light of Remark~\ref{re:ConstructionOfGoodLSCC}, we proceed to
investigate the joint spectra of regular low-density generator matrix
(LDGM) encoders.

First we define three basic linear encoders:
A \emph{single symbol repetition encoder}
$\repcode{c}: \field{q} \to \field{q}^c$ is given by
$x \mapsto (x, \ldots, x)$, where $c \in \pintegers$.
A \emph{single symbol check encoder}
$\chkcode{d}: \field{q}^d \to \field{q}$ is given by
$\seq{x} \mapsto \sum_{i=1}^d x_i$, where $d \in \pintegers$.
The sum can be abbreviated as $\seq{x}_{\oplus}$, and in the sequel,
we shall use this kind of abbreviation to denote the sum of all
components of a sequence.
A \emph{random multiplier encoder}
$\rmcode{}: \field{q} \to \field{q}$ is given by $x \mapsto Cx$ where
$C$ is uniformly distributed over $\fieldstar{q}$.

Let $c$, $n$, and $d$ be positive integers such that $d$ divides $cn$.
A random regular LDGM encoder
$\ldcode{c,d,n}: \mathbb{F}_q^{n} \to \mathbb{F}_q^{cn/d}$ is defined
by
\begin{equation}\label{eq:DefinitionOfLDGMCode}
\ldcode{c,d,n}
\eqdef \chkcode{d,cn/d} \circ \rmcode{cn} \circ \Sigma_{cn}
 \circ \repcode{c,n},
\end{equation}
where $\repcode{c,n} \eqdef \bigpprod_{i=1}^n \repcode{c}$,
$\chkcode{d,n} \eqdef \bigpprod_{i=1}^{n} \chkcode{d}$,
$\rmcode{n} \eqdef \bigpprod_{i=1}^n \rmcode{}$.

To calculate the joint spectrum of $\ldcode{c,d,n}$, we first need to
calculate the joint spectra of its constituent encoders.
We note that definition \eqref{eq:DefinitionOfLDGMCode} can be
rewritten as
\begin{IEEEeqnarray*}{rCl}
\ldcode{c,d,n}
&\eqvar{d} &\rchkcode{d,cn/d} \circ \Sigma_{cn} \circ \repcode{c,n}
 \IEEEnonumber\\
&\eqvar{d} &\chkcode{d,cn/d} \circ \Sigma_{cn} \circ \rrepcode{c,n},
\end{IEEEeqnarray*}
where the symbol $\eqvar{d}$ means that the random elements at both
sides have the same probability distribution, and
$\rrepcode{c,n} \eqdef \rmcode{cn} \circ \repcode{c,n}
 \circ \rmcode{n}$,
$\rchkcode{d,n} \eqdef \rmcode{n} \circ \chkcode{d,n}
 \circ \rmcode{dn}$.
Thus it suffices to calculate the joint spectra of $\repcode{c,n}$ and
$\rchkcode{d,n}$.

\begin{proposition}\label{pr:SpectrumOfREPCode}
\begin{equation}\label{eq:SpectrumOfREPCodeA1}
\gf_{\repcode{c}}(\seq{u}, \seq{v})
= \frac{1}{q} \sum_{a \in \field{q}} u_a v_a^c,
\end{equation}
\begin{equation}
\avG_{\rrepcode{c}}(\seq{u}, \seq{v})
= \frac{1}{q} u_{0} v_{0}^c
 + \frac{1}{q} (\seq{u}_{\oplus} - u_0)
 \left( \frac{\seq{v}_{\oplus} - v_0}{q-1} \right)^c,
 \label{eq:SpectrumOfREPCodeA2}
\end{equation}
\begin{equation}\label{eq:SpectrumOfREPCodeB}
\gf_{\repcode{c,n}}(\seq{u}, \seq{v})
= \sum_{P\in\mathcal{P}_n} \spec_{\field{q}^n}(P)
 \seq{u}^{nP} \seq{v}^{ncP},
\end{equation}
\begin{equation}\label{eq:SpectrumOfREPCodeC}
\spec_{\repcode{c,n}}(P,Q) = \spec_{\mathbb{F}_q^n}(P) 1\{P=Q\},
\end{equation}
\begin{equation}\label{eq:SpectrumOfREPCodeD}
\spec_{\repcode{c,n}}(Q|P) = 1\{Q=P\}.
\end{equation}
\end{proposition}

\begin{proposition}\label{pr:SpectrumOfCHKCode}
\begin{equation}
\avG_{\rchkcode{d}}(\seq{u}, \seq{v})
= \frac{1}{q^{d+1}} \biggl[ (\seq{u}_\oplus)^d \seq{v}_\oplus
 + \biggl( \frac{qu_0 - \seq{u}_\oplus}{q-1} \biggr)^d
 (qv_0 - \seq{v}_\oplus) \biggr],
 \label{eq:SpectrumOfCHKCodeA}
\end{equation}
\begin{equation}\label{eq:SpectrumOfCHKCodeB}
\avS_{\rchkcode{d,n}}(P, Q)
= [\seq{u}^{dnP}]\left(g_{d,n}^{(1)}(\seq{u}, Q)\right),
\end{equation}
\begin{equation}
\avS_{\rchkcode{d,n}}(P, Q)
\le g_{d,n}^{(2)}(O,P,Q)
\quad
\forall O \in \mathcal{P}_{dn}\;\mbox{with $P\ll O$},
\label{eq:SpectrumOfCHKCodeC}
\end{equation}
\begin{equation}\label{eq:SpectrumOfCHKCodeD}
\frac{1}{n} \ln \alpha_{\rchkcode{d,n}}(P, Q)
\le \delta_{d}(P(0),Q(0)) + d\Delta_{dn}(P),
\end{equation}
where $[\seq{u}^{\seq{n}}](f)$ denotes the coefficient of monomial
$\seq{u}^{\seq{n}}$ in the polynomial $f$,
\begin{IEEEeqnarray}{rCl}
g_{d,n}^{(1)}(\seq{u}, Q)
&\eqdef &\frac{{n \choose nQ}}{q^{n(d+1)}}
 \left[ (\seq{u}_\oplus)^d + (q - 1) \left(
 \frac{qu_0 - \seq{u}_\oplus}{q-1} \right)^d \right]^{nQ(0)}
 \IEEEnonumber\\
& &\breakop{\times} \left[ (\seq{u}_\oplus)^d
 - \left( \frac{qu_0 - \seq{u}_\oplus}{q-1} \right)^d
 \right]^{n(1-Q(0))},\label{eq:DefinitionOfCHKCodeGF1}
\end{IEEEeqnarray}
\begin{IEEEeqnarray}{rCl}
g_{d,n}^{(2)}(O, P, Q)
&\eqdef &\frac{{n \choose nQ}}{q^{n(d+1)} O^{dnP}}
 \left[ 1 + (q - 1) \left( \frac{qO(0) - 1}{q-1} \right)^d
 \right]^{nQ(0)} \IEEEnonumber\\
& &\breakop{\times} \left[ 1 - \left( \frac{qO(0) - 1}{q-1} \right)^d
 \right]^{n(1-Q(0))},\label{eq:DefinitionOfCHKCodeGF2}
\end{IEEEeqnarray}
\begin{IEEEeqnarray}{rCl}
\delta_{d}(x, y)
&\eqdef &\inf_{0 < \hat{x} < 1} \delta_{d}(x, \hat{x}, y),
 \label{eq:DefinitionOfDelta} \\
\delta_{d}(x, \hat{x}, y)
&\eqdef &dD(x \| \hat{x}) + J_{d}(\hat{x}, y),
 \label{eq:DefinitionOfDelta2} \\
D(x\|y) &\eqdef &D((x,1-x)\|(y,1-y)),\IEEEnonumber\\
J_{d}(x, y)
&\eqdef &y \ln\left[1+(q - 1)\left(\frac{qx - 1}{q-1}\right)^d\right]
 \IEEEnonumber\\
& &\breakop{+} (1 - y) \ln
 \left[1-\left(\frac{qx - 1}{q-1}\right)^d\right],
 \label{eq:DefinitionOfJ}\\
\Delta_{n}(P)
&\eqdef &H(P) - \frac{1}{n} \ln {n \choose nP}.
 \label{eq:DefinitionOfDelta3}
\end{IEEEeqnarray}
\end{proposition}

In the formulas above, we have used the notion of spectrum generating
function, a tool to be introduced in
Section~\ref{sec:NewMethodsAndResultsOfCodeSpectrum}.
From Propositions~\ref{pr:SpectrumOfREPCode} and
\ref{pr:SpectrumOfCHKCode} we obtain a tight upper bound on the
joint-spectrum performance of $\ldcode{c,d,n}$.%
\footnote{The tightness of the bound follows from
 \cite[Theorem~1]{JSCC:Burshtein200406}.}

\begin{theorem}\label{th:SpectrumOfLDGMCode}
\begin{equation}\label{eq:SpectrumOfLDGMCodeA}
\avS_{\ldcode{c,d,n}}(Q|P) = \avS_{\rchkcode{d,cn/d}}(Q|P)
\end{equation}
and
\begin{equation}\label{eq:SpectrumOfLDGMCodeB}
\frac{1}{n} \ln \alpha_{\ldcode{c,d,n}}(P,Q)
\le \frac{c}{d}\delta_{d}(P(0),Q(0)) + c\Delta_{cn}(P)
\end{equation}
where $\delta_{d}(x,y)$ is defined by \eqref{eq:DefinitionOfDelta} and
$\Delta_{n}(P)$ is defined by \eqref{eq:DefinitionOfDelta3}.
\end{theorem}

Our next theorem, which is based on the results above, shows that for
any $\delta > 0$, regular LDGM encoders with $d$ large enough are
approximately $\delta$-asymptotically SCC-good.

\begin{theorem}\label{th:ApproximateSCCGoodLDGMCode}
Let $\mathcal{N}$ consist of all positive integers $n$ such that $d$
divides $cn$.
Let $\{\ldcode{c,d,n}\}_{n\in \mathcal{N}}$ be a sequence of random
regular LDGM encoders whose coding rate is $r_0 = d/c$.
Let $\{A_n\}_{n\in\mathcal{N}}$ be a sequence of sets
$A_n \subseteq \mathcal{P}_{n}^*$.
Then
\[
\limsup_{n\in\mathcal{N}:n\to\infty}
 \max_{P \in A_n, Q \in \mathcal{P}_{cn/d}}
 \frac{1}{n} \ln \alpha_{\ldcode{c,d,n}}(P, Q)
\le \rho_0,
\]
where
\[
\rho_0
\eqdef \limsup_{n\in\mathcal{N}:n\to\infty}
 \max_{P \in A_n} \frac{1}{r_0}
 \ln\left[1+(q-1)\left| \frac{qP(0)-1}{q-1} \right|^d\right].
\]
If
\[
A_n
=\left\{P\in\mathcal{P}_{n}: P(0)\in
 \left[\frac{1}{q}-\gamma_1,\frac{1}{q}+\gamma_2\right]\right\}
\]
where $\gamma_1 \in (0,1/q] \backslash \{\frac{1}{2}\}$ and
$\gamma_2 \in (0,1-1/q)$, then
\[
\rho_0
= \frac{1}{r_0}
 \ln\left[1+(q-1)\left(\frac{q\gamma}{q-1}\right)^d\right]
\]
where $\gamma \eqdef \max\{\gamma_1,\gamma_2\}$.
For any $\delta>0$, define
\[
d_0(\gamma,\delta)
\eqdef\ceil{\frac{\ln[(e^{r_0 \delta}-1)/(q-1)]}{\ln[q\gamma/(q-1)]}}.
\]
Then we have $\rho_0 \le \delta$ for all $d \ge d_0(\gamma,\delta)$.
\end{theorem}

Theorem~\ref{th:ApproximateSCCGoodLDGMCode} together with
Theorem~\ref{th:ConstructionOfGoodLSCC} and
Remark~\ref{re:ConstructionOfGoodLSCC} shows that, for any
$\delta > 0$, we can construct $\delta$-asymptotically good LSCEs by
serially concatenating an inner LDGM encoder and an outer encoder of
a linear code with large minimum distance.
In particular, we may use encoders of LDPC codes as outer encoders.
Furthermore, Proposition~\ref{pr:DiagonalArgument} shows that we can
find a sequence of asymptotically good LSCEs in a family of sequences
of $\delta_i$-asymptotically good LSCEs, where $\delta_i$ is
decreasing in $i$ and converges to zero as $i \to \infty$.
An analogous construction was proposed by Hsu in his Ph.D.
dissertation \cite{JSCC:Hsu200600}, but his purpose was only to find
good channel codes and only a rate-1 LDGM encoder was employed as an
inner encoder in his construction.
A similar construction was proposed by Wainwright and Martinian
\cite{JSCC:Wainwright200903}, who proved that such a construction is
optimal for channel coding or lossy source coding with side
information.

\omitted{%
The next example shows how to determine the parameters of the inner
LDGM encoder when designing such encoders.

\begin{example}\label{ex:GoodLDGM}
Let $\delta=0.05$ and $f_n: \field{2}^n \to \field{2}^{5n}$ an
injective linear encoder.
Suppose that the normalized weight of all nonzero codewords of
$f_n(\field{2}^n)$ ranges from $0.05$ to $0.95$.
We shall design a linear encoder $H_n: \field{2}^n \to \field{2}^{2n}$
that is $\delta$-asymptotically SCC-good.
Let $H_n = G_n \circ \Sigma_{5n} \circ f_n$ where $G_n$ is a random
regular LDGM encoder over $\field{2}$.
It is clear that the coding rate of $G_n$ must be $\frac{5}{2}$.
Using Theorem~\ref{th:ApproximateSCCGoodLDGMCode} with $q=2$,
$r_0=\frac{5}{2}$, and $\gamma_1=\gamma_2=0.45$, we have
$d_0(\gamma, \delta R(f_n))=35$.
Then we may choose LDGM encoders with $c=14$ and $d=35$, so that
$\rho_0 \le 0.01$, and therefore
$\limsup_{n \to \infty} \rho(H_n) \le \rho_0 / R(f_n) = 0.05$.
\end{example}

\begin{remark}\label{re:LDGMImplementation}
  Although Theorem~\ref{th:ApproximateSCCGoodLDGMCode}, as well as
  Example~\ref{ex:GoodLDGM}, shows that LDGM encoders with $d$ large
  enough are good candidates for inner encoders, this conclusion may
  not be true for nonideal decoding algorithms.
  It is well known that the minimum distance of a typical regular LDPC
  code increases with the column weight of its defining parity-check
  matrix (e.g., \cite{JSCC:Yang201111}).
  However, in practice, its performance under the belief propagation
  (BP) decoding algorithm decreases as the column weight increases
  (e.g., \cite{JSCC:Sharon200608}).
  For this reason researchers tend to employ irregular LDPC codes in
  order to achieve better performance.
  To some extent, the design of irregular LDPC codes is
  a compromise between ideal decoding performance and iterative
  decoding convergence.
  Similarly, we cannot expect a boost in performance by simply
  increasing the parameter $d$ of a regular LDGM encoder, and we may
  also need to consider irregular LDGM encoders, i.e., a
  class of sparse generator matrices whose row or column weights are
  not uniform.  In this direction some significant work has been done
  in practice.  One example is the ``LDGM code'' in
  \cite{JSCC:Garcia200306}, which is defined as a serial concatenation
  of two systematic LDGM codes, i.e., two irregular LDGM encoders.
  These two codes share the same systematic bits, so that the BP
  decoding algorithm converges easily.  It has been shown by
  simulation that this kind of
  encoder is good for lossless JSCC \cite{JSCC:Zhong200505}.
\end{remark}
}

So far, we have presented two families of good linear encoders, one
based on Gabidulin codes (Section~\ref{subsec:MRDLSCC}) and the other
based on LDGM encoders.
A comparison between these two families seems necessary.
At first glance, it seems that the family based
on Gabidulin codes is better than the family based on
LDGM codes, since the former is SCC-good while the latter is
only asymptotically SCC-good.  However, this is not the whole
truth, because there is no single linear encoder that is SCC-good.
The proposition that follows makes this fact precise.
Consequently, in terms of code spectrum, the two
families have almost the same performance.
On the other hand,
in terms of decoding complexities, the family based on LDGM encoders
is more competitive, since LDGM encoders as well as LDPC codes are
characterized by sparse matrices so that low-complexity iterative
decoding algorithms can be employed.

\begin{proposition}\label{pr:SingleCodeSpectrum}
For any linear encoder $f: \mathcal{X}^n \to \mathcal{Y}^m$ with
$|\mathcal{X}|\ge 2$,
\begin{IEEEeqnarray*}{rCl}
\max_{\substack{P\in\mathcal{P}_n^*(\mathcal{X})\\
 Q\in\mathcal{P}_m(\mathcal{Y})}} \alpha_f(P,Q)
&\ge &\frac{|\mathcal{Y}|^m}
 {\max_{Q \in \mathcal{P}_m(\mathcal{Y})} {m \choose mQ}}\\
&= &\order\left(m^{\frac{|\mathcal{Y}|-1}{2}}\right).
\end{IEEEeqnarray*}%
\end{proposition}

The proofs of this section are given in
Appendix~\ref{subsec:ProofOfExplicitConstruction}.

\section{Advanced Toolbox of the Code-Spectrum Approach}
\label{sec:NewMethodsAndResultsOfCodeSpectrum}

In this section, we introduce some advanced tools required to prove
the results in
Sections~\ref{sec:ConceptsOfGoodLinearCodes}--\ref{sec:ExplicitConstruction}.
In particular, we shall establish tools for serial and parallel
concatenations of linear encoders and the MacWilliams identities on
the duals of linear encoders.
These results are not new in nature, but serve our purpose of
providing a concise mathematical treatment and completing the
code-spectrum approach.
Their proofs are left to the reader as exercises, or can be found in
\cite{JSCC:Yang200909}.

\subsection{Spectra with Coordinate Partitions}
\label{subsec:SpectraWithPartition}

In this subsection, we introduce a generalization of spectra,
viz., spectra of sets with coordinate partitions.

Let $A$ be a subset of $\field{q}^n$, with coordinate set
$\mathcal{I}_n$.
Given a partition $\mathcal{U}$ of $\mathcal{I}_n$, we define the
\emph{$\mathcal{U}$-type} $P^{\mathcal{U}}_{\seq{x}}$ of
$\seq{x} \in \field{q}^n$ as
\[
P^{\mathcal{U}}_{\seq{x}}=(P^{U}_{\seq{x}})_{U \in \mathcal{U}}
\eqdef (P_{x_{U}})_{U \in \mathcal{U}}.
\]
By $\mathcal{P}_{\mathcal{U}}$ we mean the set of all
$\mathcal{U}$-types of vectors in $\field{q}^n$, so that
$
\mathcal{P}_{\mathcal{U}}
= \prod_{U \in \mathcal{U}} \mathcal{P}_{|U|}.
$
A $\mathcal{U}$-type in $\mathcal{P}_{\mathcal{U}}$ is written in the
form $P^{\mathcal{U}} \eqdef (P^{U})_{U \in \mathcal{U}}$.
For a $\mathcal{U}$-type $P^{\mathcal{U}}$, the set of vectors of
$\mathcal{U}$-type $P^{\mathcal{U}}$ in $\field{q}^n$ is denoted by
$\mathcal{T}_{P^{\mathcal{U}}}$.
In the sequel, when given $P^{\mathcal{U}}_{\seq{x}}$ or
$P^{\mathcal{U}}$, we shall slightly abuse the notations
$P^{\mathcal{V}}_{\seq{x}}$ and $P^{\mathcal{V}}$ for any subset
$\mathcal{V}$ of $\mathcal{U}$ to represent part of their components.

Based on the $\mathcal{U}$-type, we define the
\emph{$\mathcal{U}$-spectrum} $\spec_{\field{q}^{\mathcal{U}}}(A)$ of
a nonempty set $A\subseteq\field{q}^n$ as the empirical distribution
of $\mathcal{U}$-types of sequences in $A$, i.e.,
\[
\spec_{\field{q}^\mathcal{U}}(A)(P^{\mathcal{U}})
\eqdef
\frac{|\{ \seq{x} \in A : P^{\mathcal{U}}_{\seq{x}} =
 P^{\mathcal{U}} \}|}{|A|}.
\]
The $\mathcal{U}$-spectrum is in fact a variant of the joint spectrum.%
\footnote{When $\mathcal{U}$ is the disjoint union of $\mathcal{V}$
 and $\mathcal{W}$, we sometimes write
 $\spec_{\field{q}^{\mathcal{V}}\field{q}^{\mathcal{W}}}(A)$,
 which looks more like an ordinary joint spectrum and is used to
 distinguish between coordinates.}
When $\mathcal{U} = \{ \mathcal{I}_n \}$, it reduces to the ordinary
spectrum.
Otherwise it provides more information than the ordinary spectrum, and
in the extreme case $\mathcal{U} = \{ \{1\},\{2\},\dots,\{n\} \}$ the
spectrum $\spec_{\field{q}^\mathcal{U}}(A)$ determines $A$ uniquely.
Since $P^{\mathcal{U}}$ clearly refers to an element of
$\mathcal{P}_\mathcal{U}$, we sometimes write
$\spec(A)(P^\mathcal{U})$ or further $\spec_A(P^\mathcal{U})$ in place
of $\spec_{\field{q}^\mathcal{U}}(A)(\mathcal{P}^\mathcal{U})$ for
convenience.

\omitted{%
\begin{example}\label{ex:scp}
Consider the linear code of Example~\ref{ex:HammingCode1} with
coordinate partition $\{\{1,2,3\}, \{4,5,6,7\}\}$.
Its spectrum, written as a two-dimensional array, is
\begin{equation*}
\renewcommand{\arraystretch}{1.2}
  \begin{array}{c|ccccc}
    w_1\backslash w_2&0&1&2&3&4\\\hline
    0&\frac{1}{16}&&&\frac{1}{16}\\
    1&&&\frac{3}{16}&\frac{3}{16}\\
    2&&\frac{3}{16}&\frac{3}{16}&\\
    3&&\frac{1}{16}&&&\frac{1}{16}
  \end{array}
\end{equation*}
where $w_1$ and $w_2$ correspond to $\{1,2,3\}$ and $\{4,5,6,7\}$,
respectively.
\end{example}
}

Analogous to ordinary spectra, we further define the marginal and
conditional spectra with respect to a proper subset $\mathcal{V}$ of
$\mathcal{U}$, denoting them by
$\spec_{\field{q}^{\mathcal{V}}}(A)(P^{\mathcal{V}})$ and
$\spec_{\field{q}^{\mathcal{U}\setminus\mathcal{V}}
 |\field{q}^{\mathcal{V}}}(A)(P^{\mathcal{U}\setminus\mathcal{V}}
 |P^{\mathcal{V}})$,
respectively.
We also define the $(\mathcal{U},\mathcal{V})$-spectrum of a map
$f: \field{q}^n \to \field{q}^m$ as
$\spec_{\field{q}^{\mathcal{U}}\field{q}^{\mathcal{V}}}(\mathrm{rl}(f))$
where $\mathcal{U}$ and $\mathcal{V}$ are partitions of
$\mathcal{I}_n$ and $\mathcal{I}_m$, respectively.

For ease of notation, when we explicitly write
$A \subseteq \prod_{i=1}^s \field{q}^{n_i}$ with
$\sum_{i=1}^s n_i = n$, we tacitly assume that the default coordinate
partition is
\begin{IEEEeqnarray*}{rCl}
\mathcal{U}_0 &\eqdef &\{U_1, U_2, \ldots, U_s\}\\
&= &\left\{ \{1,\ldots,n_1\}, \ldots,
 \left\{\sum_{i=1}^{s-1} n_i + 1, \ldots, n\right\} \right\}.
\end{IEEEeqnarray*}
Thus the default spectrum of $A$ is
$\spec_{\field{q}^{\mathcal{U}_0}}(A)$ and is denoted by
$\spec_{\field{q}^s}(A)$ (or $\spec(A)$ for $s=1$).

To further explore the properties of $\mathcal{U}$-spectra, we first
take a closer look at the $\mathcal{U}$-type.
Recall that any two vectors $\seq{x}, \seq{x}' \in \field{q}^n$ have
the same type if and only if $\sigma(\seq{x}) = \seq{x}'$ for some
$\sigma \in \symmetricgroup{n}$.
Since the type is a special case of the $\mathcal{U}$-type, it is
natural to ask which permutations in $\symmetricgroup{n}$ preserve the
$\mathcal{U}$-type.
A moment's thought shows that the $\mathcal{U}$-type is preserved by
any permutation that maps each member of $\mathcal{U}$ onto itself.
We denote the set of all such permutations by
$\symmetricgroup{\mathcal{U}}$, which forms a subgroup of
$\symmetricgroup{n}$ isomorphic to
$\prod_{U\in\mathcal{U}} \symmetricgroup{|U|}$.
Now considering a random permutation uniformly distributed over
$\symmetricgroup{\mathcal{U}}$, we obtain a generalization of
$\Sigma_n$, which is denoted by $\Sigma_{\mathcal{U}}$ and is called a
\emph{uniform random permutation with respect to $\mathcal{U}$}.
We are now ready to state a fundamental result about
$\mathcal{U}$-spectra.

\begin{proposition}\label{pr:GeneralSpectrumPropertyOfSets}
Let $\mathcal{U}$ be a partition of $\mathcal{I}_n$.
For any $\seq{x} \in \field{q}^n$ and any random nonempty set
$A \subseteq \field{q}^n$,
\begin{equation}\label{eq:GeneralSpectrumPropertyOfSets1}
\av\left[\frac{1\{\seq{x} \in \Sigma_{\mathcal{U}}(A)\}}{|A|}\right]
= q^{-n} \alpha_A(P^{\mathcal{U}}_{\seq{x}})
\end{equation}
where
\begin{equation}\label{eq:Definition2OfAlpha}
\alpha_A(P^{\mathcal{U}}) \eqdef
\frac{\avS_A(P^{\mathcal{U}})}{\spec_{\field{q}^n}(P^{\mathcal{U}})}.
\end{equation}
Moreover, for any proper subset $\mathcal{V}$ of $\mathcal{U}$ and any
$Q^{\mathcal{U}} \in \mathcal{P}_{\mathcal{U}}$ with
$Q^{\mathcal{V}} = P^\mathcal{V}_{\seq{x}}$, we have
\begin{equation}\label{eq:GeneralSpectrumPropertyOfSets2}
\av\left[\frac{|B \cap \Sigma_{\mathcal{U}}(A)|}{|A|}\right]
=
\frac{\avS_A(Q^{\mathcal{U}})}
 {\prod_{V\in\mathcal{V}} q^{|V|}\spec_{\field{q}^{|V|}}(Q^V)},
\end{equation}
where
$B \eqdef \{\seq{y} \in \mathcal{T}_{Q^{\mathcal{U}}}: y_V=x_V
 \mbox{ for all $V \in \mathcal{V}$}\}$.
\end{proposition}

\omitted{%
\begin{remark}
Identity \eqref{eq:GeneralSpectrumPropertyOfSets1} can be
rewritten as
\[
\pr\{\seq{x} \in \Sigma_{\mathcal{U}}(A)\}
= q^{-n}|A| \alpha_A(P^{\mathcal{U}}_{\seq{x}})
\]
whenever $|A|$ is a constant.
For a random function $F: \field{q}^n \to \field{q}^m$, if we let
$A = \mathrm{rl}(F) \subseteq \field{q}^{n}\times \field{q}^m$ (which
implies the default partition $\{\mathcal{I}_{n}, \mathcal{I}_{m}'\}$%
\footnote{There is a collision between coordinate sets when we
 consider the pair $(\seq{x}, F(\seq{x}))$ as a vector of
 $\field{q}^{m+n}$.
 The trick is to rename the output coordinate set as
 $\mathcal{I}_m' = \{1', 2', \ldots, m'\}$, so that the whole
 coordinate set is $\mathcal{I}_n \cup \mathcal{I}_m'$.}%
) and note that $|A| = q^n$, then \citeSpectrumProperty{} (for
$\mathcal{X}=\mathcal{Y}=\field{q}$) follows as a special case.
In general, for a random function
$F: \prod_{i=1}^s \field{q}^{n_i} \to \prod_{i=1}^t \field{q}^{m_i}$,
we may consider the default coordinate partitions $\mathcal{U}_0$ and
$\mathcal{V}_0$ of its domain and range, respectively, and define
$\tilde{F} \eqdef \Sigma_{\mathcal{V}_0} \circ F \circ
 \Sigma_{\mathcal{U}_0}$.
Then Proposition~\ref{pr:GeneralSpectrumPropertyOfSets} yields a
generalization of \citeSpectrumProperty, that is,
\[
\pr\{\tilde{F}(\seq{x}) = \seq{y}\}
= q^{-m} \alpha_F(P^{\mathcal{U}_0}_{\seq{x}},
 P^{\mathcal{V}_0}_{\seq{y}}),
\]
where $m \eqdef \sum_{i=1}^t m_i$ and
$\alpha_F(P^{\mathcal{U}_0},Q^{\mathcal{V}_0})
 \eqdef \alpha_{\mathrm{rl}(F)}(P^{\mathcal{U}_0},Q^{\mathcal{V}_0})$.
\end{remark}
}

\subsection{Encoders and Conditional Probability Distributions}
\label{subsec:ConditionalProbability}

In this subsection, we show that any encoder may be regarded as a
conditional probability distribution.
Such a viewpoint is very helpful in calculating the spectrum of a
complex encoder composed of many simple encoders.

\begin{proposition}\label{pr:SpectrumPropertyX1OfFunctions}
For any random function
$F: \prod_{i=1}^s \field{q}^{n_i} \to \prod_{i=1}^t \field{q}^{m_i}$,
\begin{equation}
\pr\left\{\rtilde{F}(\seq{x})\in\mathcal{T}_{Q^{\mathcal{V}_0}}
 \right\}
= \avS_F(Q^{\mathcal{V}_0} | P^{\mathcal{U}_0}_{\seq{x}})
\label{eq:SpectrumPropertyX1OfFunctions}
\end{equation}
for all $\seq{x} \in \prod_{i=1}^s \field{q}^{n_i}$ and
$Q^{\mathcal{V}_0} \in \mathcal{P}_{\mathcal{V}_0}$, where
$\rtilde{F} \eqdef F \circ \Sigma_{\mathcal{U}_0}$, and
$\mathcal{U}_0$ and $\mathcal{V}_0$ are the default coordinate
partitions.
\end{proposition}

\omitted{%
\begin{remark}
Identity \eqref{eq:SpectrumPropertyX1OfFunctions} can also be
rewritten as
\[
\pr\{\rtilde{F}(\seq{x}) \in \mathcal{T}_{Q^{\mathcal{V}_0}}
 \mid \seq{x} \in \mathcal{T}_{P^{\mathcal{U}_0}} \}
= \avS_F(Q^{\mathcal{V}_0} | P^{\mathcal{U}_0}),
\]
which clearly indicates that the average forward conditional spectrum
$\avS_F(Q^{\mathcal{V}_0}|P^{\mathcal{U}_0})$ may be regarded as the
transition probability from $P^{\mathcal{U}_0}$ to
$Q^{\mathcal{V}_0}$ under $\rtilde{F}$.
This fundamental observation implies that coding modules like
$\rtilde{F}$ or $\tilde{F}$ (instead of $F$) should be regarded as
basic units in a coding system, and that the serial concatenation of
such units may behave like the serial concatenation of conditional
probability distributions.
The following proposition proves this speculation.
\end{remark}
}

\begin{proposition}\label{pr:SpectrumOfSeriallyConcatenatedFunctions}
For any two random functions $F: \field{q}^n \to \field{q}^m$ and
$G: \field{q}^m \to \field{q}^l$,
\[
\avS_{G \circ \Sigma_m \circ F}(Q|O)
= \sum_{P \in \mathcal{P}_m} \avS_F(P|O) \avS_G(Q|P),
\]
where $O \in \mathcal{P}_n$ and $Q \in \mathcal{P}_l$.
\end{proposition}

\subsection{Spectrum Generating Functions}
\label{subsec:SpectrumGeneratingFunctions}

In Section~\ref{subsec:ConditionalProbability}, we introduced a
method for calculating the spectra of serial concatenations of
encoders.
In this subsection, we proceed to investigate another important
combination of encoders, viz.\ parallel concatenations.
To cope with problems involving concatenations (cartesian products) of
sequences, we shall introduce the approach of spectrum generating
functions.

At first, we need some additional terminology for partitions to
simplify the treatment of spectrum generating functions.
Associated with any partition $\mathcal{U}$ of a set $S$ is the
mapping $\pi_{\mathcal{U}}\colon S\to\mathcal{U}$ that maps $s\in S$
to the member of $\mathcal{U}$ containing $s$.
A partition $\mathcal{V}$ of $S$ is a refinement of $\mathcal{U}$ if
and only if there is a (unique) mapping
$\psi: \mathcal{V} \to \mathcal{U}$ such that
$\pi_{\mathcal{U}} = \psi \circ \pi_{\mathcal{V}}$.

Now let $A$ be a nonempty subset of $\field{q}^n$ and $\mathcal{U}$ a
partition of $\mathcal{I}_n$.
The \emph{$\mathcal{U}$-spectrum generating function} of $A$ is a
polynomial in $q|\mathcal{U}|$ indeterminates, whose coefficients form
the $\mathcal{U}$-spectrum of $A$.
As an element of
$\complexnumbers[u_{U,a}; U\in\mathcal{U}, a\in\field{q}]$ (which
denotes the ring of polynomials in the indeterminates $u_{U,a}$ and
with coefficients in $\complexnumbers$), it can be defined as
\begin{IEEEeqnarray*}{rCl}
\gf_{\field{q}^{\mathcal{U}}}(A)(\seq{u}_{\mathcal{U}})
&\eqdef &\frac{1}{|A|} \sum_{\seq{x} \in A} \prod_{i=1}^n
 u_{\pi_\mathcal{U}(i), x_i} \\
&= &\sum_{P^{\mathcal{U}} \in \mathcal{P}_{\mathcal{U}}}
 \left(\spec_A(P^{\mathcal{U}})
 \prod_{U \in \mathcal{U}} \prod_{a \in \field{q}}
 u_{U,a}^{|U|P^U(a)}\right),
\end{IEEEeqnarray*}
where
$\seq{u}_{\mathcal{U}}
 \eqdef (\seq{u}_U)_{U \in \mathcal{U}}
 = (u_{U,a})_{U\in\mathcal{U}, a\in\field{q}}$.
For convenience, we sometimes write $\gf_{\field{q}^{\mathcal{U}}}(A)$
or $\gf(A)(\seq{u}_{\mathcal{U}})$, or further
$\gf_A(\seq{u}_\mathcal{U})$ (since $\mathcal{U}$ conveys all
necessary information), and write $\seq{u}^{\seq{v}}$ in place of
$\prod_{i\in\mathcal{I}} u_i^{v_i}$ for any sequences $\seq{u}$,
$\seq{v}$ with the same coordinate set $\mathcal{I}$ (whenever the
product makes sense).
Thus the product $\prod_{a \in \field{q}} u_{U,a}^{|U|P^U(a)}$ is
rewritten as $\seq{u}_{U}^{|U|P^U}$.
As is done for $\mathcal{U}$-spectra, we write $\gf_{\field{q}^s}(A)$
(or $\gf(A)$ for $s=1$) when $\mathcal{U}$ is the default coordinate
partition.

\omitted{%
\begin{example}\label{ex:sgf}
Let us compute the spectrum generating function of the linear encoder
in Example~\ref{ex:HammingCode1}.
Its input coordinate set is $\{1, 2, 3, 4\}$ and its output coordinate
set is $\{1, 2, 3, 4, 5, 6, 7\}$, so its spectrum generating function
with respect to the default coordinate partition is
\begin{equation}
\frac{1}{16} \big[ u_0^4v_0^7 + u_0^3u_1(3v_0^4v_1^3 + v_0^3v_1^4)
 + u_0^2u_1^2(2v_0^4v_1^3 + 3v_0^3v_1^4 + v_1^7)
 + u_0u_1^3(v_0^4v_1^3 + 3v_0^3v_1^4) + u_1^4v_0^4v_1^3 \big],
\label{eq:sgf.a}
\end{equation}
where $u_i$ and $v_i$ ($i = 0$, $1$) correspond to the symbol $i$ in
the input and output alphabets, respectively.
If we replace the default output partition by
\[
\{A, B\} = \{\{1,2,3\}, \{4,5,6,7\}\}
\quad(\text{cf.\ Example~\ref{ex:scp}}),
\]
then we obtain
\begin{multline}
\frac{1}{16} \Big\{ u_0^4a_0^3b_0^4 + u_0^3u_1\big[a_0^2a_1b_0^2b_1^2
 + a_0a_1^2(2b_0^3b_1 + b_0^2b_1^2)\big]\\
{}+ u_0^2u_1^2\big[(a_0^3 + a_0^2a_1)b_0b_1^3
 + a_0a_1^2(b_0^3b_1+b_0^2b_1^2) + a_1^3(b_0^3b_1+b_1^4)\big]\\
{}+ u_0u_1^3\big[a_0^2a_1(b_0^2b_1^2 + 2b_0b_1^3)
 + a_0a_1^2b_0^2b_1^2)\big] + u_1^4a_0^2a_1b_0^2b_1^2 \Big\},
\label{eq:sgf.b}
\end{multline}
where $a_i$ and $b_i$ correspond to the partial coordinate sets $A$
and $B$, respectively.
\end{example}
}

The relation between spectrum generating functions for different
partitions is well described by a special substitution homomorphism,
which we shall define now.
Let $\psi$ be a map of $\mathcal{U}$ onto $\mathcal{V}$.
It induces a mapping from
$\complexnumbers[u_{U,a}; U \in \mathcal{U}, a \in \field{q}]$ to
$\complexnumbers[v_{V,a}; V \in \mathcal{V}, a \in \field{q}]$ given
by
\[
f((u_{U,a})_{U\in\mathcal{U}, a\in\field{q}})
\mapsto f((v_{\psi(U),a})_{U\in\mathcal{U}, a\in\field{q}}),
\]
which is a substitution homomorphism by
\cite[Corollary~5.6]{JSCC:Hungerford197400}.
Intuitively $\psi$ does nothing but substitutes each indeterminate
$u_{U,a}$ with $v_{\psi(U),a}$.

\begin{proposition}
\label{pr:SubstitutionPrincipleOfGeneratingFunction}
Suppose $\mathcal{U}$ and $\mathcal{V}$ are two partitions of
$\mathcal{I}_n$.
If $\mathcal{U}$ is a refinement of $\mathcal{V}$ and
$\psi:\mathcal{U} \to \mathcal{V}$ is the map such that
$\pi_\mathcal{V} = \psi \circ \pi_\mathcal{U}$, then $\psi$ maps
$\gf_{\field{q}^{\mathcal{U}}}(A)$ to
$\gf_{\field{q}^{\mathcal{V}}}(A)$ for $A \subseteq \field{q}^n$.
\end{proposition}

\omitted{%
\begin{example}
Let us apply
Proposition~\ref{pr:SubstitutionPrincipleOfGeneratingFunction} to
Example~\ref{ex:sgf}.
It is easy to see that the map from the partition $\{A,B\}$ to the
default output partition $\{\mathcal{I}_7\}$ is the constant map
$\psi(x) = \mathcal{I}_7$, and therefore \eqref{eq:sgf.a} follows from
\eqref{eq:sgf.b} with substitutions $a_i,b_i \mapsto v_i$.
\end{example}
}

Another notable fact is
multiplicativity of the spectrum generating function
with respect to the cartesian product of sets.

\begin{proposition}\label{pr:GeneratingFunctionOfProductOfSets}
For any sets $A_i \subseteq \field{q}^{n_i}$ where $1\le i\le s$,
\[
\gf_{\prod_{i=1}^s A_i}(\seq{u}_{\mathcal{I}_s})
= \prod_{i=1}^s \gf_{A_i}(\seq{u}_i).
\]
\end{proposition}

From Propositions~\ref{pr:SubstitutionPrincipleOfGeneratingFunction}
and \ref{pr:GeneratingFunctionOfProductOfSets}, three corollaries
follow.

\begin{corollary}\label{co:GeneratingFunctionOfSetProduct}
For any sets $A_1 \subseteq \field{q}^{n_1}$ and
$A_2 \subseteq \field{q}^{n_2}$,
\[
\gf_{A_1\times A_2}(\seq{u})
= \gf_{A_1}(\seq{u}) \cdot \gf_{A_2}(\seq{u}).
\]
\end{corollary}

\begin{corollary}\label{co:GeneratingFunctionOfSpace}
\[
\gf_{\field{q}^n}(\seq{u})
= [\gf_{\field{q}}(\seq{u})]^n
= \left( \frac{\sum_{a \in \field{q}} u_a}{q} \right)^n.
\]
\end{corollary}

\begin{corollary}
\label{co:GeneratingFunctionOfFunctionParallelProduct}
For any two maps $f_1: \field{q}^{n_1} \to \field{q}^{m_1}$ and
$f_2: \field{q}^{n_2} \to \field{q}^{m_2}$,
\[
\gf_{f_1 \pprod f_2}(\seq{u}, \seq{v})
= \gf_{f_1}(\seq{u}, \seq{v}) \cdot \gf_{f_2}(\seq{u}, \seq{v}),
\]
where $f_1 \pprod f_2$ is understood as a map from
$\field{q}^{n_1+n_2}$ to $\field{q}^{m_1+m_2}$.
\end{corollary}

Note that Corollary~\ref{co:GeneratingFunctionOfSpace} is an easy
consequence of Corollary~\ref{co:GeneratingFunctionOfSetProduct}, and
that Corollary~\ref{co:GeneratingFunctionOfFunctionParallelProduct} is
the desired tool for computing the
spectra of parallel concatenations of linear encoders.

When $A \subseteq \mathcal{X}^n$ is random, its associated spectrum
generating function is also random.
To analyze a random polynomial, we consider its expectation.
For a random polynomial $F: \Omega \to \complexnumbers[u_a; a \in A]$
with finite image, its expectation can be defined as
\[
\av[F] \eqdef \sum_{f \in F(\Omega)} \pr\{F = f\} f
\]
by using the $\complexnumbers$-algebra structure of the polynomial
ring.%
\footnote{Since the collection of random polynomials with finite image
 is enough for our purpose, we shall not discuss the expectation of a
 general random polynomial.}
\omitted{%
Analogous to ordinary expectations, expectations of random polynomials
have the following properties:
For any random polynomials $F_1$ and $F_2$ over
$\complexnumbers[u_a; a \in A]$,
\[
\av[F_1+F_2] = \av[F_1] + \av[F_2].
\]
If $F_1$ and $F_2$ are independent, then
\[
\av[F_1 F_2] = \av[F_1] \av[F_2],
\]
which also implies $\av[f_1F_2] = f_1\av[F_2]$.
Using these properties, we also have
$\av[F] = \sum_{\seq{n} \in \nnintegers^A}
 \av[[\seq{u}^{\seq{n}}](F)] \cdot \seq{u}^\seq{n}$.
}%
The next proposition states an important property of expectations of
random spectrum generating functions.
To simplify the notation, we shall write, e.g.,
$\avG_{\field{q}^{\mathcal{U}}}(A)$ in place of
$\av[\gf_{\field{q}^{\mathcal{U}}}(A)]$.

\begin{proposition}\label{pr:RenameGeneratingFunction}
Let $\mathcal{U}$ be a partition of $\mathcal{I}_n$ and
$\{F_U: \field{q} \to \field{q}\}_{U\in \mathcal{U}}$ be a collection
of random bijective mappings, which induces a substitution
homomorphism
$\overline{F}: \complexnumbers[u_{U,a}; U \in \mathcal{U},
 a \in \field{q}] \to \complexnumbers[u_{U,a}; U \in \mathcal{V},
 a \in \field{q}]$
given by
\[
f((u_{U,a})_{U\in\mathcal{U}, a\in\field{q}})
\mapsto f((\av[u_{U,F_U(a)}])_{U\in\mathcal{U}, a\in\field{q}})
\]
and a random map $F: \field{q}^n \to \field{q}^n$ given by
\[
\seq{x} \mapsto F^{(1)}(x_1)
F^{(2)}(x_2) \cdots F^{(n)}(x_n),
\]
where $F^{(i)}$ is an independent copy of $F_{\pi_\mathcal{U}(i)}$.
Then
$\avG_{\field{q}^{\mathcal{U}}}(F(A))
 = \overline{F}(\avG_{\field{q}^{\mathcal{U}}}(A))$
for any random nonempty set $A \subseteq \field{q}^n$.
\end{proposition}

\subsection{MacWilliams Identities}\label{subsec:NewResultsV}

One of the most famous results in coding theory are the MacWilliams
identities \cite{JSCC:MacWilliams196300}, relating the weight
enumerator of a linear code to that of its dual code.
In this subsection, we shall introduce the MacWilliams identities in
the framework of the code-spectrum approach.
This may be regarded as a combination of the results in
\cite{JSCC:Wood199903,JSCC:Simonis199502,JSCC:Honold200100}.

The \emph{dual} $A^\perp$ of a linear code $A \subseteq \field{q}^n$
is the orthogonal set
$\{\seq{x} \in \field{q}^n : \seq{x}\transpose{\seq{z}} = 0
 \mbox{ for all $\seq{z}\in A$}\}$.
Clearly, $A^{\perp}$ is a subspace of $\field{q}^n$ as well.
The next theorem shows the relation between $A^{\perp}$ and $A$ in
terms of spectrum generating functions.

\begin{theorem}\label{th:MacWilliamsIdentitiesF}
Let $A$ be a subspace of $\field{q}^n$ and $\mathcal{U}$ a partition
of $\mathcal{I}_n$.
Then
\[
\gf_{A^{\perp}}(\seq{u}_{\mathcal{U}})
= \frac{1}{|A^{\perp}|}
 \gf_A((\seq{u}_U \mat{M})_{U \in \mathcal{U}}),
\]
where $\mat{M}$ is the $q \times q$ matrix (indexed by the elements of
$\field{q}$) defined by
\begin{equation}\label{eq:MacWilliamsMatrixF}
\matentry{M}_{a_1, a_2} = \chi(a_1 a_2)
\qquad \forall a_1, a_2 \in \field{q},
\end{equation}
using the ``generating'' character
$\chi(x) \eqdef e^{2\pi i\, \mathrm{Tr}(x)/p}$
with
$\mathrm{Tr}(x) \eqdef x+x^p+\cdots+x^{p^{r-1}}$.
\end{theorem}

\begin{remark}
Note that $x\mapsto\mathrm{Tr}(x)$, the absolute trace of $\field{q}$,
is $\field{p}$-linear, and hence $\chi(x)$ is a homomorphism from the
additive group of $\field{q}$ to the multiplicative group
$\unitsof{\complexnumbers}$ (a so-called additive character of
$\field{q}$).
It is easy to see that $\sum_{x \in \field{q}}\chi(ax) = 0$ for
$a\ne 0$, so that $q^{-\frac{1}{2}}\mat{M}$ is a symmetric unitary
matrix.
In particular,
\begin{equation}\label{eq:MacWilliamsMatrixPropertyF}
\sum_{x \in \field{q}} \matentry{M}_{a,x}
= \sum_{x \in \field{q}} \matentry{M}_{x,a}
= q1\{a = 0\}.
\end{equation}
\end{remark}

One important application of Theorem~\ref{th:MacWilliamsIdentitiesF}
is calculating the spectrum of a linear encoder
$\seq{x} = \seq{y} \transpose{\mat{A}}$ when the spectrum of
$\seq{y} = \seq{x}\mat{A}$ is known.
The next theorem gives the details.

\begin{theorem}\label{th:MacWilliamsIdentitiesJF}
Let $\mat{A}$ be an $n \times m$ matrix over $\field{q}$.
Define the linear encoders $f: \field{q}^n \to \field{q}^m$ and
$g: \field{q}^m \to \field{q}^n$ by $f(\seq{x})\eqdef\seq{x}\mat{A}$
and $g(\seq{y})\eqdef\seq{y}\transpose{\mat{A}}$, respectively.
Let $\mathcal{U}$ be a partition of $\mathcal{I}_n$ and $\mathcal{V}$
a partition of $\mathcal{I}_m$.
Then
\[
\gf_{\field{q}^\mathcal{V}\field{q}^{\mathcal{U}}}(-g)
 (\seq{v}_{\mathcal{V}}, \seq{u}_{\mathcal{U}})
= \frac{1}{q^m}
 \gf_{\field{q}^\mathcal{U}\field{q}^{\mathcal{V}}}(f)
 ((\seq{u}_U \mat{M})_{U \in \mathcal{U}},
 (\seq{v}_V \mat{M})_{V \in \mathcal{V}}),
\]
where $\mat{M}$ is defined by \eqref{eq:MacWilliamsMatrixF}.
\end{theorem}

\section{Conclusion}\label{sec:Conclusion}

In this paper, we present some general principles and schemes for
constructing linear encoders with good joint spectra:

\begin{itemize}
\item
In Section~\ref{subsec:MRDLSCC}, we provide a family of SCC-good
random linear encoders derived from Gabidulin MRD codes.

\item
In Section~\ref{subsec:ConstructingGoodLSCCGPA}, it is proved in
Theorem~\ref{th:ConstructionOfGoodEquivalentCodes} that we can
construct $\delta$-asymptotically good LSCEs which are SC-equivalent
(resp., CC-equivalent) to given $\delta$-asymptotically good LSEs
(resp., LCEs).

\item
In Section~\ref{subsec:ConstructingGoodLSCCGPB}, we propose in
Theorem~\ref{th:ConstructionOfGoodLSCC} a general serial concatenation
scheme for constructing good LSCEs.

\item
In Section~\ref{sec:ExplicitConstruction}, the joint spectrum of
a regular LDGM encoder is analyzed.
By means of Theorem~\ref{th:ApproximateSCCGoodLDGMCode}, we show that
regular LDGM encoders with appropriate parameters are approximately
$\delta$-asymptotically SCC-good.
Based on this analysis, we finally present a serial concatenation
scheme with one encoder of an LDPC code as outer encoder and one LDGM
encoder as inner encoder, and prove it to be asymptotically SCC-good.
\end{itemize}

In addition, we define in Section~\ref{sec:ConceptsOfGoodLinearCodes}
three code-spectrum criteria for good linear encoders, so that all
important coding issues are subsumed under one single research
problem: \emph{Constructing linear encoders with good spectra}.
Through investigating the relations among these criteria, we find
that a good joint spectrum is the most important feature of a linear
encoder.

The main ideas of this paper formed during the period from 2007
to 2008.  Since then, there have been many advances in coding theory,
two of them deserving particular attention.
One is \emph{spatial coupling} \cite{JSCC:Kudekar201102}, a
fundamental mechanism that helps
increase the BP threshold of a new ensemble of codes to the MAP
threshold of its underlying ensemble.
In fact, this technique has already been used for a long time in the
design of LDPC convolutional codes \cite{JSCC:Felstrom199909}, and its
excellent iterative decoding performance is well known, e.g. from
\cite{JSCC:Lentmaier201010}.
Clearly, combining this technique with the LDGM-based scheme (in
Sec.~\ref{sec:ExplicitConstruction}) seems a promising way for
designing good coding schemes in practice.
For example, we may serially concatenate an outer encoder of a
quasi-cyclic LDPC code (e.g., \cite{JSCC:Huang201205}) with an inner
spatially-coupled regular LDGM encoder.
The other advance are \emph{polar codes}
\cite{JSCC:Arikan200907}, which constitute the first known
code construction that approaches capacity within a gap
$\epsilon>0$ with delay and complexity both depending polynomially on
$1/\epsilon$ \newChange{\cite{JSCC:Hassani201205,JSCC:Guruswami201304}}.
However, the minimum distance of a polar code is only a sublinear
function of the block length.
It is unknown if there exists a fundamental trade-off among
minimum distance, decoding complexity, gap to capacity, etc.
Regardless of whether such a law exists, it is
valuable in practice to think of the ``sublinear'' counterpart of
linear encoders with good joint spectra, that is, we may allow
$\lim_{k\to\infty} \min_{\seq{x}\ne 0^{n_k}}
 (H(P_{\seq{x}})R(f_k) + H(P_{f_k(\seq{x})})) = 0$
(cf.\ \cite{JSCC:Yang200909} and \cite[Theorem~4.1]{JSCC:Yang200904}).

\appendices

\section{Proofs of Results in Section
 \ref{sec:ConceptsOfGoodLinearCodes}}
\label{subsec:ProofOfConceptsOfGoodLinearCodes}

\begin{proofof}{Proposition \ref{pr:RelationBetweenLSCandLCC}}
Let $\mat{G}_k$ be an $l_k\times n_k$ generator matrix that yields
$\ker F_k$.
Thus $G_k(\seq{x}) \eqdef \seq{x} \mat{G}_k$ is the linear encoder
desired.
\end{proofof}

\begin{proofof}{Proposition \ref{pr:RelationBetweenLCCandLSC}}
Let $\mat{H}_k$ be an $l_k\times m_k$ parity-check matrix that yields
$F_k(\field{q}^{n_k})$.
Thus $G_k(\seq{x}) \eqdef \seq{x} \transpose{\mat{H}_k}$ is the linear encoder
desired.
\end{proofof}

\begin{proofof}{Proposition \ref{pr:RelationBetweenLSCCandLSCLCC}}
It follows from \eqref{eq:DefinitionOfAsympGoodLSCC} that
\begin{equation}\label{eq:ProofEq1OfRelationBetweenLSCCandLSCLCC}
\max_{P \in \mathcal{P}_{n_k}^*, Q \in \mathcal{P}_{m_k}}
 \frac{1}{n_k} \ln \alpha_{F_k}(P, Q)
\le \delta + \epsilon
\end{equation}
for any $\epsilon > 0$ and sufficiently large $k$.
Then it follows that
\begin{IEEEeqnarray*}{rCl}
\max_{P \in \mathcal{P}_{n_k}^*} \frac{1}{n_k}\ln\alpha_{\ker F_k}(P)
&\eqvar{(a)} &\max_{\seq{x}: \seq{x} \ne 0^{n_k}} \frac{1}{n_k}
 \ln \left( \av\left[
 \frac{q^{n_k} 1\{\seq{x} \in \Sigma_{n_k}(\ker F_k)\}}{|\ker F_k|}
 \right] \right)\\
&\le &\max_{\seq{x}: \seq{x} \ne 0^{n_k}} \frac{1}{n_k}
 \ln (q^{m_k} \pr\{\rtilde{F}_k(\seq{x}) = 0^{m_k}\})\\
&= &\max_{\seq{x}: \seq{x} \ne 0^{n_k}} \frac{1}{n_k} \ln (q^{m_k}
 \pr\{\tilde{F}_k(\seq{x}) = 0^{m_k}\})\\
&\eqvar{(b)} &\max_{P \in \mathcal{P}_{n_k}^*} \frac{1}{n_k}
 \ln \alpha_{F_k}(P, P_{0^{m_k}})\\
&\levar{(c)} &\delta + \epsilon
\end{IEEEeqnarray*}
for sufficiently large $k$, where (a) follows from
Proposition~\ref{pr:GeneralSpectrumPropertyOfSets}, (b) from
\citeSpectrumProperty, and (c) follows from
\eqref{eq:ProofEq1OfRelationBetweenLSCCandLSCLCC}.
Since $\epsilon$ is arbitrary, we conclude that $\seqx{F}$ is
$\delta$-asymptotically SC-good.

Also by \eqref{eq:ProofEq1OfRelationBetweenLSCCandLSCLCC}, we have
\begin{equation}
\avS_{F_k}(P, Q)
\le e^{n_k(\delta + \epsilon)}
 \spec_{\field{q}^{n_k} \times \field{q}^{m_k}}(P, Q)
\quad
\forall P \in \mathcal{P}_{n_k}^*, Q \in \mathcal{P}_{m_k}
\label{eq:ProofEq2OfRelationBetweenLSCCandLSCLCC}
\end{equation}
for sufficiently large $k$.
Then for any $Q \in \mathcal{P}_{m_k}^*$, it follows that
\begin{IEEEeqnarray*}{rCl}
\avS_{F_k(\field{q}^{n_k})}(Q)
&\eqvar{(a)} &\avS_{F_k}(Q)\\
&= &\sum_{P \in \mathcal{P}_{n_k}} \avS_{F_k}(P,Q)\\
&\eqvar{(b)} &\sum_{P \in \mathcal{P}_{n_k}^*} \avS_{F_k}(P,Q)\\
&\levar{(c)} &\sum_{P \in \mathcal{P}_{n_k}^*}
 e^{n_k(\delta+\epsilon)}
 \spec_{\field{q}^{n_k}\times\field{q}^{m_k}}(P,Q)\\
&\le &e^{n_k(\delta+\epsilon)} \spec_{\field{q}^{m_k}}(Q)
\end{IEEEeqnarray*}
for sufficiently large $k$, where (a) follows from the linear property
of $F_k$, (b) from $Q \ne P_{0^{m_k}}$, and (c) follows from
\eqref{eq:ProofEq2OfRelationBetweenLSCCandLSCLCC}.
Hence we have
\begin{IEEEeqnarray*}{rCl}
\limsup_{k \to \infty} \max_{Q \in \mathcal{P}_{m_k}^*} \frac{1}{m_k}
 \ln \alpha_{F_k(\field{q}^{n_k})}(Q)
&= &\limsup_{k \to \infty} \max_{Q \in \mathcal{P}_{m_k}^*}
 \frac{1}{m_k} \ln
 \frac{\avS_{F_k(\field{q}^{n_k})}(Q)}{\spec_{\field{q}^{m_k}}(Q)}\\
&\le &\limsup_{k \to \infty} \frac{1}{m_k}
 \ln e^{n_k(\delta + \epsilon)}\\
&= &(\delta + \epsilon) \overline{R}(\seqx{F}).
\end{IEEEeqnarray*}
Because $\epsilon$ is arbitrary, $\seqx{F}$ is
$\delta\overline{R}(\seqx{F})$-asymptotically CC-good.
\end{proofof}

\begin{proofof}{Proposition~\ref{pr:DiagonalArgument}}
Note that
\[
\limsup_{k \to \infty} \rho(G_{k,k})
\le \limsup_{k \to \infty} \rho(G_{i,k})
\le \delta_i \quad \forall i\in\pintegers
\]
and hence
$\limsup_{k \to \infty} \rho(G_{k,k})
 \le \inf_{i \in \pintegers} \delta_i = \delta$.
\end{proofof}

\section{Proofs of Results in Section
 \ref{subsec:ConstructingGoodLSCCGPA}}
\label{subsec:ProofOfConstructingGoodLSCCGPA}

\begin{proofof}{Theorem \ref{th:elementary_abelian}}
Condition \eqref{eq:SuperGoodLinearCodes2} implies that for every pair
$\seq{x}\in\mathcal{X}^n\setminus\{0^n\}$, $\seq{y}\in\mathcal{Y}^m$
there exists at least one linear encoder
$f:\mathcal{X}^n\to\mathcal{Y}^m$ satisfying $f(\seq{x})=\seq{y}$.%
\footnote{The existence of such an $f$ depends only on the types of
 $\seq{x}$ and $\seq{y}$, since $f$ is linear iff
 $\sigma\circ f\circ\pi$ is linear for any permutations
 $\pi \in \symmetricgroup{n}$, $\sigma \in \symmetricgroup{m}$.}
Since $m,n\geq 1$, this in turn implies, for every pair
$x\in\mathcal{X}\setminus\{0\}$, $y\in\mathcal{Y}$, the existence of
at least one group homomorphism $h:\mathcal{X}\to\mathcal{Y}$
satisfying $h(x)=y$.

If $p$ is a prime dividing $|\mathcal{X}|$, there exists
$x\in\mathcal{X}$ of order $p$.
Then, by the condition above, every $y\in\mathcal{Y}$ must have order
$1$ or $p$, so $\mathcal{Y}\cong\integers_p^s$ for some $s$.
If $|\mathcal{X}|$ had a prime divisor $q\neq p$, then
$\mathcal{Y}\cong\integers_p^s\cong\integers_q^t$ and so
$|\mathcal{Y}|=1$, a contradiction.
Thus, $\mathcal{X}$ and $\mathcal{Y}$ must be $p$-groups for the same
prime $p$, and $\mathcal{Y}$ must be elementary abelian.
Finally, if $\mathcal{X}$ contained an element $x$ of order $p^2$, we
would have $px\neq 0$ but $h(px)=ph(x)=0$ for any group homomorphism
$h:\mathcal{X}\to\mathcal{Y}$.
This implies again $|\mathcal{Y}|=1$ and concludes the proof.
\end{proofof}

To prove Theorem~\ref{th:KernelOfSuperGoodLinearCodes}, we need the
following lemma.

\begin{lemma}\label{le:KernelOfSuperGoodLinearCodes1}
Let $F: \mathcal{X}^n \to \mathcal{Y}^m$ be a random linear encoder.
If $F$ is SCC-good, then
\begin{equation}
\av[|\ker F|] = 1 + |\mathcal{Y}|^{-m} (|\mathcal{X}|^n - 1).
\end{equation}
\end{lemma}

\begin{proof}
\begin{IEEEeqnarray*}{rCl}
\av[|\ker F|]
&= &\av[|\ker \tilde{F}|] \\
&= &\av\left[\sum_{\seq{x} \in \mathcal{X}^n}
 1\{\tilde{F}(\seq{x}) = 0^m\}\right]\\
&= &1 + \sum_{\seq{x} \in \mathcal{X}^n \backslash \{0^n\}}
 \pr\{\tilde{F}(\seq{x}) = 0^m\}\\
&\eqvar{(a)} &1 + |\mathcal{Y}|^{-m} (|\mathcal{X}|^n - 1),
\end{IEEEeqnarray*}
where (a) follows from \eqref{eq:SuperGoodLinearCodes2}.
\end{proof}

\begin{proofof}{Theorem \ref{th:KernelOfSuperGoodLinearCodes}}
By Lagrange's theorem, $|\ker F|$ can take only values in
$\{1,p,\dots,p^{rn}\}$.
Hence, using Lemma~\ref{le:KernelOfSuperGoodLinearCodes1}, we obtain
\begin{IEEEeqnarray*}{rCl}
1+\frac{q^n-1}{q^n}
&= &\av[|\ker F|]\\
&\geq &\pr\{|\ker F|=1\}+p\cdot\left(1-\pr\{|\ker F|=1\}\right).
\end{IEEEeqnarray*}
Solving for $\pr\{|\ker F|=1\}$ gives the stated inequality.
\end{proofof}

\begin{proofof}{Proposition~\ref{pr:KernelOfSuperGoodLinearCodes}}
Suppose $\mathcal{X}\cong\integers_2^s$, so that
$\mathcal{X}^n\cong\integers_2^{ns}\cong(\field{2}^{ns},+)$ for all
$n\in\pintegers$.
Let $F_n:\field{2}^{ns}\to\field{2}^{ns}$ be the random linear encoder
derived from a binary $(ns,ns,2)$ Gabidulin MRD code $\mathcal{C}$ in
accordance with Theorem~\ref{th:MRDLSCC}.
By definition, the code $\mathcal{C}$ consists of $2^{2ns}$ matrices
$\mat{A}\in\field{2}^{ns\times ns}$ with
$\rank{\mat{A}}\in\{0,ns-1,ns\}$, and by the rank distribution of MRD
codes (\cite[Theorem~5.6]{JSCC:Delsarte197800} or
\cite[Theorem~5]{JSCC:Gabidulin198501}), there are $2(2^{ns}-1)$
matrices of rank $ns$ in $\mathcal{C}$, so that
\[
\lim_{n\to\infty}\pr\{|\ker F_n|=1\}
=\lim_{n\to\infty}\frac{2(2^{ns}-1)}{2^{2ns}}
=0.
\]
Since $F_n$ is SCC-good, this proves the proposition.
\end{proofof}

\begin{proofof}{Theorem \ref{th:ConstructionOfEquivalentCodes}}
Inequalities \eqref{eq:SCEquivalentProbability} and
\eqref{eq:CCEquivalentProbability} follow immediately from
\eqref{eq:RankOfRLC}, so our task is to evaluate the average
conditional spectra of $G_1$ and $G_2$.

For any $P \in \mathcal{P}_n^*$ and $Q \in \mathcal{P}_m$,
\begin{IEEEeqnarray*}{rCl}
\avS_{G_1}(Q|P)
&\eqvar{(a)} &\sum_{O \in \mathcal{P}_m} \avS_F(O|P)
 \avS_{\rlccode{m,m}}(Q|O)\\
&= &\sum_{O \in \mathcal{P}_m^*} \avS_F(O|P)
 \avS_{\rlccode{m,m}}(Q|O)\\
& &\breakop{+} \avS_F(P_{0^m}|P) \avS_{\rlccode{m,m}}(Q|P_{0^m})\\
&\eqvar{(b)} &\spec_{\field{q}^m}(Q) \sum_{O\in\mathcal{P}_m^*}
 \avS_F(O|P)\\
& &\breakop{+} 1\{Q = P_{0^m}\} \avS_F(P_{0^m}|P)\\
&\le &\spec_{\field{q}^m}(Q) + 1\{Q = P_{0^m}\} \avS_F(P_{0^m}|P)
\end{IEEEeqnarray*}
where (a) follows from
Proposition~\ref{pr:SpectrumOfSeriallyConcatenatedFunctions} and (b)
follows from \eqref{eq:SuperGoodLinearCodes1}.
This concludes \eqref{eq:SCEquivalentCodeSpectrum}.

Analogously, for any $P \in \mathcal{P}_n^*$ and
$Q \in \mathcal{P}_m$,
\begin{IEEEeqnarray*}{rCl}
\avS_{G_2}(Q|P)
&\eqvar{(a)} &\sum_{O \in \mathcal{P}_n} \avS_{\rlccode{n,n}}(O|P)
 \avS_F(Q|O)\\
&\eqvar{(b)} &\sum_{O \in \mathcal{P}_n} \spec_{\field{q}^n}(O)
 \avS_F(Q|O)\\
&= &\sum_{O \in \mathcal{P}_n} \avS_F(O,Q)\\
&= &\avS_F(Q)\\
&= &\avS_{F(\field{q}^n)}(Q)
\end{IEEEeqnarray*}
where (a) follows from
Proposition~\ref{pr:SpectrumOfSeriallyConcatenatedFunctions} and (b)
follows from \eqref{eq:SuperGoodLinearCodes1}.
This concludes \eqref{eq:CCEquivalentCodeSpectrum} and hence completes
the proof.
\end{proofof}

\begin{proofof}{Theorem \ref{th:ConstructionOfGoodEquivalentCodes}}
For the first statement, recall that the source transmission rate of
$\delta$-asymptotically good LSEs must converge.
Then for any $\epsilon>0$, since $\seqx{f}$ is $\delta$-asymptotically
SC-good and $R_c(\seqx{f}) = \ln q$, we have
\begin{equation}\label{eq:ProofEq1OfConstructionOfGoodEquivalentCodes}
\alpha_{\ker f_k}(P)
\le e^{n_k(\delta + \epsilon)}
\qquad \forall P \in \mathcal{P}_{n_k}^*
\end{equation}
and
\begin{equation}\label{eq:ProofEq2OfConstructionOfGoodEquivalentCodes}
|f_k(\field{q}^{n_k})|
\ge q^{m_k} e^{-m_k \epsilon}
\ge q^{m_k} e^{-2n_k\epsilon/R(\seqx{f})}
\end{equation}
for sufficiently large $k$.
Define $G_{1,k} \eqdef \rlccode{m_k,m_k} \circ f_k$.
It follows from Theorem~\ref{th:ConstructionOfEquivalentCodes} that
\begin{equation}\label{eq:ProofEq3OfConstructionOfGoodEquivalentCodes}
\pr\{\ker G_{1,k} = \ker f_k\} > K_q
\end{equation}
and
\[
\avS_{G_{1,k}}(Q|P)
\le \spec_{\field{q}^{m_k}}(Q)
 + 1\{Q = P_{0^{m_k}}\} \spec_{f_k}(P_{0^{m_k}}|P)
\quad
\forall P \in \mathcal{P}_{n_k}^*, Q \in \mathcal{P}_{m_k}.
\]
Hence for any $P \in \mathcal{P}_{n_k}^*$ and
$Q \in \mathcal{P}_{m_k}$,
\begin{IEEEeqnarray*}{rCl}
\alpha_{G_{1,k}}(P,Q)
&= &\frac{\avS_{G_{1,k}}(Q|P)}{\spec_{\field{q}^{m_k}}(Q)}\\
&\le &1 + 1\{Q = P_{0^{m_k}}\} \alpha_{f_k}(P, P_{0^{m_k}})\\
&\eqvar{(a)} &1 + 1\{Q = P_{0^{m_k}}\} q^{m_k}
 \pr\{\rtilde{f}_k(\seq{x}) = 0^{m_k}\}\\
&= &1 + 1\{Q = P_{0^{m_k}}\} q^{m_k}
 \pr\{\seq{x} \in \Sigma_{n_k}(\ker f_k)\}\\
&\eqvar{(b)} &1 + 1\{Q = P_{0^{m_k}}\} q^{m_k}
 \frac{|\ker f_k|}{q^{n_k}} \alpha_{\ker f_k}(P)\\
&= &1 + 1\{Q = P_{0^{m_k}}\} \frac{q^{m_k}}{|f_k(\field{q}^{n_k})|}
 \alpha_{\ker f_k}(P)\\
&\levar{(c)} &1 + 1\{Q = P_{0^{m_k}}\} e^{2n_k\epsilon/R(\seqx{f})}
 \alpha_{\ker f_k}(P)\\
&\levar{(d)} &e^{n_k(\delta + 2\epsilon + 2\epsilon/R(\seqx{f}))}
\end{IEEEeqnarray*}
for sufficiently large $k$, where (a) follows from
\citeSpectrumProperty{} and $\seq{x}$ is a vector of type $P$, (b)
from Proposition~\ref{pr:GeneralSpectrumPropertyOfSets}, (c) from
\eqref{eq:ProofEq2OfConstructionOfGoodEquivalentCodes}, and (d)
follows from \eqref{eq:ProofEq1OfConstructionOfGoodEquivalentCodes}.
Thus for sufficiently large $k$,
\[
\rho(G_{1,k}) \le \delta + 2\epsilon + \frac{2\epsilon}{R(\seqx{f})}.
\]
Define the random linear encoder $G_{1,k}'$ as $G_{1,k}$ given the
event $A_k \eqdef \{\ker G_{1,k} = \ker f_k\}$.
Then it follows that for sufficiently large $k$,
\begin{IEEEeqnarray*}{rCl}
\rho(G_{1,k}')
&\le &\rho(G_{1,k}) - \frac{1}{n_k} \ln \pr(A_k)\\
&\le &\delta + 2\epsilon
 + \frac{2\epsilon}{R(\seqx{f})} - \frac{1}{n_k} \ln \pr(A_k)\\
&\levar{(a)} &\delta + 3\epsilon + \frac{2\epsilon}{R(\seqx{f})},
\end{IEEEeqnarray*}
where (a) follows from
\eqref{eq:ProofEq3OfConstructionOfGoodEquivalentCodes}.
Since $\epsilon$ is arbitrary, $\{G_{1,k}'\}_{k=1}^\infty$ is a
sequence of $\delta$-asymptotically good LSCEs such that
$\ker G_{1,k}' = \ker f_k$.
By \cite[Proposition 4.1]{JSCC:Yang200904}, we conclude that there
exists a sequence $\{g_{1,k}\}_{k=1}^\infty$ of
$\delta$-asymptotically good LSCEs
$g_{1,k}: \field{q}^{n_k} \to \field{q}^{m_k}$ such that $g_{1,k}$ is
SC-equivalent to $f_k$ for each $k \in \pintegers$.

The proof of the second statement is analogous.
Let $\epsilon > 0$ be given.
Since $\seqx{f}$ is $\delta$-asymptotically CC-good,
\begin{equation}\label{eq:ProofEq4OfConstructionOfGoodEquivalentCodes}
\alpha_{f_k(\field{q}^{n_k})}(Q) \le e^{m_k(\delta + \epsilon)}
\qquad \forall Q \in \mathcal{P}_{m_k}^*
\end{equation}
for sufficiently large $k$.
Define $G_{2,k} \eqdef f_k \circ \rlccode{n_k,n_k}$.
Then it follows from Theorem~\ref{th:ConstructionOfEquivalentCodes}
that
\begin{equation}\label{eq:ProofEq5OfConstructionOfGoodEquivalentCodes}
\pr\{G_{2,k}(\field{q}^{n_k}) = f_k(\field{q}^{n_k})\} > K_q
\end{equation}
and
\[
\avS_{G_{2,k}}(Q|P) = \spec_{f_k(\field{q}^{n_k})}(Q)
\quad
\forall P \in \mathcal{P}_{n_k}^*, Q \in \mathcal{P}_{m_k}.
\]
Hence for any $P \in \mathcal{P}_{n_k}^*$ and
$Q \in \mathcal{P}_{m_k}$,
\begin{IEEEeqnarray*}{rCl}
\alpha_{G_{2,k}}(P,Q)
&= &\frac{\avS_{G_{2,k}}(Q|P)}{\spec_{\field{q}^{m_k}}(Q)}\\
&= &\frac{\spec_{f_k(\field{q}^{n_k})}(Q)}
 {\spec_{\field{q}^{m_k}}(Q)}\\
&= &\alpha_{f_k(\field{q}^{n_k})}(Q)
\end{IEEEeqnarray*}
for sufficiently large $k$.
Define the random linear encoder $G_{2,k}'$ as $G_{2,k}$ given the
event
$B_k \eqdef \{G_{2,k}(\field{q}^{n_k}) = f_k(\field{q}^{n_k})\}$.
Then it follows that for any $P \in \mathcal{P}_{n_k}^*$ and
$Q \in \mathcal{P}_{m_k}^*$,
\[
\alpha_{G_{2,k}'}(P,Q)
\le \frac{\alpha_{f(\field{q}^{n_k})}(Q)}{\pr(B_k)}
\levar{(a)} e^{m_k(\delta + 2\epsilon)}
\]
for sufficiently large $k$, where (a) follows from
\eqref{eq:ProofEq4OfConstructionOfGoodEquivalentCodes} and
\eqref{eq:ProofEq5OfConstructionOfGoodEquivalentCodes}.
Since $f_k$ is injective, we have
\[
\alpha_{G_{2,k}'}(P,P_{0^{m_k}}) = 0
\qquad \forall P \in \mathcal{P}_{n_k}^*
\]
and $R(f_k) = R_c(f_k)/\ln q$ converges as $k \to \infty$.
Therefore,
\[
\rho(G_{2,k}') \le \frac{\delta + 2\epsilon}{R(\seqx{f})} + \epsilon
\]
for sufficiently large $k$.
Since $\epsilon$ is arbitrary, $\{G_{2,k}'\}_{k=1}^\infty$ is a
sequence of $\delta/R(\seqx{f})$-asymptotically good LSCEs such that
$G_{2,k}'(\field{q}^{n_k}) = f_k(\field{q}^{n_k})$.
By \cite[Proposition 4.1]{JSCC:Yang200904}, we conclude that there
exists a sequence $\{g_{2,k}\}_{k=1}^\infty$ of
$\delta/R(\seqx{f})$-asymptotically good LSCEs
$g_{2,k}: \mathcal{X}^{n_k} \to \mathcal{Y}^{m_k}$ such that $g_{2,k}$
is CC-equivalent to $f_k$ for each $k \in \pintegers$.
\end{proofof}

\section{Proofs of Results in Section
 \ref{subsec:ConstructingGoodLSCCGPB}}
\label{subsec:ProofOfConstructingGoodLSCCGPB}

\begin{proofof}{Theorem \ref{th:ConstructionOfGoodLSCC}}
For any $\epsilon > 0$, since $G_k$ is $\delta$-asymptotically
SCC-good relative to $A_k$, we have
\begin{equation}
\avS_{G_k}(Q|P) \le e^{m_k(\delta+\epsilon)}\spec_{\field{q}^{l_k}}(Q)
\quad \forall P \in A_k, Q \in P_{l_k}
\label{eq:ConstructionOfGoodLSCCEq1}
\end{equation}
for sufficiently large $k$.
Then for all $O \in \mathcal{P}_{n_k}^*$ and $Q \in P_{l_k}$,
\begin{IEEEeqnarray*}{rCl}
\avS_{G_k \circ \Sigma_{m_k} \circ F_k}(Q|O)
&\eqvar{(a)} &\sum_{P \in \mathcal{P}_{m_k}} \avS_{F_k}(P|O)
 \avS_{G_k}(Q|P)\\
&\eqvar{(b)} &\sum_{P \in A_k} \avS_{F_k}(P|O) \avS_{G_k}(Q|P)\\
&\levar{(c)} &\sum_{P \in A_k} e^{m_k(\delta + \epsilon)}
 \spec_{\field{q}^{l_k}}(Q) \avS_{F_k}(P|O)\\
&\le &e^{m_k(\delta + \epsilon)} \spec_{\field{q}^{l_k}}(Q)
\end{IEEEeqnarray*}
for sufficiently large $k$, where (a) follows from
Proposition~\ref{pr:SpectrumOfSeriallyConcatenatedFunctions}, (b) from
condition \eqref{eq:ConstructionOfGoodLSCC1}, and (c) follows from
\eqref{eq:ConstructionOfGoodLSCCEq1}.
Therefore, for sufficiently large $k$,
\[
\rho(G_k \circ \Sigma_{m_k} \circ F_k)
\le \frac{\delta + \epsilon}{\underline{R}(\seqx{F})} + \epsilon.
\]
Since $\epsilon$ is arbitrary, this establishes the theorem.
\end{proofof}

\begin{proofof}{Proposition
 \ref{pr:ConstructionOfGoodLSCCbySurjectiveMapping}}
For any $\epsilon > 0$, since $F_k$ is $\delta$-asymptotically
SCC-good, we have
\begin{equation}
\avS_{F_k}(P|O) \le e^{n_k(\delta+\epsilon)}\spec_{\field{q}^{m_k}}(P)
\quad
\forall O \in P_{n_k}^*, P \in P_{m_k}
\label{eq:ConstructionOfGoodLSCCbySurjectiveMappingEq1}
\end{equation}
for sufficiently large $k$.
Then for all $O \in \mathcal{P}_{n_k}^*$ and $Q \in P_{l_k}$,
\begin{IEEEeqnarray*}{rCl}
\avS_{G_k \circ \Sigma_{m_k} \circ F_k}(Q|O)
&\eqvar{(a)} &\sum_{P \in \mathcal{P}_{m_k}} \avS_{F_k}(P|O)
 \avS_{G_k}(Q|P)\\
&\levar{(b)} &\sum_{P \in \mathcal{P}_{m_k}}
 e^{n_k(\delta+\epsilon)} \spec_{\field{q}^{m_k}}(P)\avS_{G_k}(Q|P)\\
&= &e^{n_k(\delta + \epsilon)} \sum_{P \in \mathcal{P}_{m_k}}
 \avS_{G_k}(P,Q)\\
&= &e^{n_k(\delta + \epsilon)} \avS_{G_k}(Q)\\
&= &e^{n_k(\delta + \epsilon)} \avS_{G_k(\field{q}^{m_k})}(Q)\\
&\eqvar{(c)} &e^{n_k(\delta + \epsilon)} \spec_{\field{q}^{l_k}}(Q)
\end{IEEEeqnarray*}
for sufficiently large $k$, where (a) follows from
Proposition~\ref{pr:SpectrumOfSeriallyConcatenatedFunctions}, (b) from
\eqref{eq:ConstructionOfGoodLSCCbySurjectiveMappingEq1}, and (c)
follows from the surjectivity of $G_k$.
Therefore,
$\rho(G_k \circ \Sigma_{m_k} \circ F_k) \le \delta + \epsilon$
for sufficiently large $k$.
This concludes the proof, because $\epsilon$ is arbitrary.
\end{proofof}

\section{Proofs of Results in Section \ref{sec:ExplicitConstruction}}
\label{subsec:ProofOfExplicitConstruction}

\begin{proofof}{Proposition \ref{pr:SpectrumOfREPCode}}
The identity \eqref{eq:SpectrumOfREPCodeA1} holds clearly.
This together with Proposition~\ref{pr:RenameGeneratingFunction}
gives \eqref{eq:SpectrumOfREPCodeA2}.

From \eqref{eq:SpectrumOfREPCodeA1} and
Corollary~\ref{co:GeneratingFunctionOfFunctionParallelProduct}, it
further follows that
\begin{IEEEeqnarray*}{rCl}
\gf_{\repcode{c,n}}(\seq{u},\seq{v})
&= &\left( \frac{1}{q} \sum_{a \in \field{q}} u_a v_a^c \right)^n\\
&= &\sum_{P \in \mathcal{P}_n} \spec_{\field{q}^n}(P) \seq{u}^{nP}
 \seq{v}^{ncP}.
\end{IEEEeqnarray*}
This proves \eqref{eq:SpectrumOfREPCodeB}, and then identities
\eqref{eq:SpectrumOfREPCodeC} and \eqref{eq:SpectrumOfREPCodeD}
follows.
\end{proofof}

\begin{proofof}{Proposition \ref{pr:SpectrumOfCHKCode}}
Note that the generator matrix of $\rchkcode{d}$ is the transpose of
the generator matrix of $\rrepcode{d}$.
Then by Theorem~\ref{th:MacWilliamsIdentitiesJF}, it follows that
\begin{IEEEeqnarray*}{rCl}
\avG_{-\rchkcode{d}}(\seq{u},\seq{v})
&= &\frac{1}{q^d} \avG_{\rrepcode{d}}(\hat{\seq{v}},\hat{\seq{u}})\\
&\eqvar{(a)} &\frac{1}{q^{d+1}} \left[ \hat{u}_0^d \hat{v}_0
 + \left( \frac{\hat{\seq{u}}_\oplus - \hat{u}_0}{q-1} \right)^d
 (\hat{\seq{v}}_\oplus - \hat{v}_0) \right]\\
&\eqvar{(b)} &\frac{1}{q^{d+1}} \left[ (\seq{u}_\oplus)^d
 \seq{v}_\oplus + \left( \frac{qu_0 - \seq{u}_\oplus}{q-1} \right)^d
 (qv_0 - \seq{v}_\oplus) \right]
\end{IEEEeqnarray*}
where $\hat{\seq{u}} = \seq{u} \mat{M}$ and
$\hat{\seq{v}} = \seq{v} \mat{M}$, (a) follows from
Proposition~\ref{pr:SpectrumOfREPCode}, and (b) follows from property
\eqref{eq:MacWilliamsMatrixPropertyF}.
This together with Proposition~\ref{pr:RenameGeneratingFunction}
concludes \eqref{eq:SpectrumOfCHKCodeA}.

By Corollary~\ref{co:GeneratingFunctionOfFunctionParallelProduct}, we
further have
\begin{IEEEeqnarray*}{rCl}
\avG_{\rchkcode{d,n}}(\seq{u},\seq{v})
&= &\frac{1}{q^{n(d+1)}} \left[ (\seq{u}_\oplus)^d \seq{v}_\oplus
 + \left( \frac{qu_0 - \seq{u}_\oplus}{q-1} \right)^d
 (qv_0 - \seq{v}_\oplus) \right]^n\\
&= &\frac{1}{q^{n(d+1)}} \Biggl\{ \left[ (\seq{u}_\oplus)^d
 + (q-1)\left( \frac{qu_0 - \seq{u}_\oplus}{q-1} \right)^d
 \right] v_0\\
& &\breakop{+} \left[ (\seq{u}_\oplus)^d
 - \left( \frac{qu_0 - \seq{u}_\oplus}{q-1} \right)^d \right]
 (\seq{v}_\oplus - v_0) \Biggr\}^n\\
&= &\frac{1}{q^{n(d+1)}} \sum_{Q \in \mathcal{P}_{n}}
 \Biggl\{ {n \choose nQ} \seq{v}^{nQ}\\
& &\breakop{\times} \left[ (\seq{u}_\oplus)^d
 + (q-1)\left( \frac{qu_0 - \seq{u}_\oplus}{q-1} \right)^d
 \right]^{nQ(0)}\\
& &\breakop{\times} \left[ (\seq{u}_\oplus)^d
 - \left( \frac{qu_0 - \seq{u}_\oplus}{q-1} \right)^d
 \right]^{n(1-Q(0))} \Biggr\}.
\end{IEEEeqnarray*}
Hence,
\begin{IEEEeqnarray*}{rCl}
\avS_{\rchkcode{d,n}}(P,Q)
&= &[\seq{u}^{dnP}\seq{v}^{nQ}] \left(
 \avG_{\rchkcode{d,n}}(\seq{u},\seq{v}) \right)\\
&= &[\seq{u}^{dnP}] \left( g_{d,n}^{(1)}(\seq{u},Q) \right)
\end{IEEEeqnarray*}
where $g_{d,n}^{(1)}(\seq{u}, Q)$ is defined by
\eqref{eq:DefinitionOfCHKCodeGF1}.
This proves \eqref{eq:SpectrumOfCHKCodeB}.

Since $g_{d,n}^{(1)}(\seq{u}, Q)$ is a polynomial with nonnegative
coefficients, $[\seq{u}^{dnP}](g_{d,n}^{(1)}(\seq{u},Q))$ can be
bounded above by
\[
[\seq{u}^{dnP}] \left( g_{d,n}^{(1)}(\seq{u}, Q) \right)
\le \frac{g_{d,n}^{(1)}(O, Q)}{O^{dnP}}
= g_{d,n}^{(2)}(O, P, Q)
\]
where $O$ is an arbitrary type in $\mathcal{P}_{dn}$ such that
$P\ll O$, and $g_{d,n}^{(2)}(O, P, Q)$ is defined by
\eqref{eq:DefinitionOfCHKCodeGF2}.
This gives \eqref{eq:SpectrumOfCHKCodeC}.

Finally, let us estimate $\alpha_{\rchkcode{d,n}}(P,Q)$.
\begin{IEEEeqnarray*}{rCl}
\alpha_{\rchkcode{d,n}}(P,Q)
&\le &\frac{g_{d,n}^{(2)}(O,P,Q)}
 {\spec_{\field{q}^{dn} \times \field{q}^n}(P,Q)}\\
&= &\frac{1}{{dn \choose dnP} O^{dnP}} \left[ 1 + (q - 1)
 \left( \frac{qO(0) - 1}{q-1} \right)^d \right]^{nQ(0)}\\
& &\breakop{\times} \left[ 1 - \left( \frac{qO(0) - 1}{q-1} \right)^d
 \right]^{n(1-Q(0))}\\
&= &\frac{e^{dnH(P)} P^{dnP}}{{dn \choose dnP} O^{dnP}} \left[
 1 + (q - 1) \left( \frac{qO(0) - 1}{q-1} \right)^d \right]^{nQ(0)}\\
& &\breakop{\times} \left[ 1 - \left( \frac{qO(0) - 1}{q-1} \right)^d
 \right]^{n(1-Q(0))}\\
&= &e^{dn\Delta_{dn}(P)} e^{dn D(P\|O)}\\
& &\breakop{\times} \left[ 1 + (q - 1) \left( \frac{qO(0) - 1}{q-1}
 \right)^d \right]^{nQ(0)}\\
& &\breakop{\times} \left[ 1 - \left( \frac{qO(0) - 1}{q-1} \right)^d
 \right]^{n(1-Q(0))}.
\end{IEEEeqnarray*}
Note that $D(P\|O) \ge D(P(0)\|O(0))$ with equality if and only if
$P(a)(1-O(0)) = O(a)(1-P(0))$ for all $a \ne 0$, and thus we obtain a
minimized upper bound \eqref{eq:SpectrumOfCHKCodeD}.
\end{proofof}

\begin{proofof}{Theorem \ref{th:SpectrumOfLDGMCode}}
By the definition of $\ldcode{c,d,n}$ and
Proposition~\ref{pr:SpectrumOfSeriallyConcatenatedFunctions}, it
follows that
\begin{IEEEeqnarray*}{rCl}
\avS_{\ldcode{c,d,n}}(Q|P)
&= &\sum_{O \in \mathcal{P}_{cn}} \avS_{\repcode{c,n}}(O|P)
 \avS_{\rchkcode{d,cn/d}}(Q|O)\\
&\eqvar{(a)} &\sum_{O \in \mathcal{P}_{cn}} 1\{O = P\}
 \avS_{\rchkcode{d,cn/d}}(Q|O)\\
&= &\avS_{\rchkcode{d,cn/d}}(Q|P),
\end{IEEEeqnarray*}
where (a) follows from \eqref{eq:SpectrumOfREPCodeD}.
This proves \eqref{eq:SpectrumOfLDGMCodeA}.

Furthermore, we have
\begin{IEEEeqnarray*}{rCl}
\frac{1}{n} \ln \alpha_{\ldcode{c,d,n}}(P,Q)
&= &\frac{1}{n} \ln \frac{\avS_{\ldcode{c,d,n}}(Q|P)}
 {\spec_{\field{q}^{cn/d}}(Q)}\\
&\eqvar{(a)} &\frac{1}{n} \ln \frac{\avS_{\rchkcode{d,cn/d}}(Q|P)}
 {\spec_{\field{q}^{cn/d}}(Q)}\\
&= &\frac{1}{n} \ln \alpha_{\rchkcode{d,cn/d}}(P,Q)\\
&\levar{(b)} &\frac{c}{d} \delta_{d}(P(0), Q(0)) + c \Delta_{cn}(P),
\end{IEEEeqnarray*}
where (a) follows from \eqref{eq:SpectrumOfLDGMCodeA} and (b) follows
from \eqref{eq:SpectrumOfCHKCodeD}.
This concludes \eqref{eq:SpectrumOfLDGMCodeB} and hence completes the
proof.
\end{proofof}

To prove Theorem~\ref{th:ApproximateSCCGoodLDGMCode}, we need the
following lemma.

\begin{lemma}\label{le:UpperBoundOfDelta}
For all $x, y \in [0, 1]$,
\[
\delta_{d}(x,y) \le J_{d}(x, y)
\le \ln \left[ 1 + (qy-1) \left( \frac{qx-1}{q-1} \right)^d \right]
\]
where $\delta_d(x,y)$ and $J_d(x,y)$ are defined by
\eqref{eq:DefinitionOfDelta} and \eqref{eq:DefinitionOfJ},
respectively.
\end{lemma}

\begin{proof}
When $x \in (0, 1)$ and $y \in [0, 1]$, the first inequality clearly
holds by taking $\hat{x} = x$ in \eqref{eq:DefinitionOfDelta}.
If however $x = 0$, then
\[
\lim_{\hat{x} \to 0} \delta_{d}(0, \hat{x}, y)
= \lim_{\hat{x} \to 0} \left( d\ln
 \frac{1}{1-\hat{x}} + J_{d}(\hat{x},y) \right)
= J_{d}(0,y).
\]
Hence $\delta_{d}(0, y) \le J_{d}(0,y)$.
A similar argument also applies to the case of $x = 1$.
The second inequality follows from Jensen's inequality.
\end{proof}

\begin{proofof}{Theorem \ref{th:ApproximateSCCGoodLDGMCode}}
By Theorem~\ref{th:SpectrumOfLDGMCode},
Lemma~\ref{le:UpperBoundOfDelta}, and the condition $r_0=d/c$, it
follows that
\begin{IEEEeqnarray*}{rCl}
\frac{1}{n} \ln \alpha_{\ldcode{c,d,n}}(P,Q)
&\le &\frac{1}{r_0} \ln\left[1+(qQ(0)-1) \left( \frac{qP(0)-1}{q-1}
 \right)^d\right] + c \Delta_{cn}(P)\\
&\levar{(a)} &\frac{1}{r_0}
 \ln\left[1+(q-1)\left|\frac{qP(0)-1}{q-1}\right|^d\right]
 + \frac{q\ln(cn+1)}{n}
\end{IEEEeqnarray*}
where (a) follows from the strict increasing property of $\ln x$,
$Q(0) \in [0,1]$, and the inequality
\[
{n \choose nP} \ge \frac{1}{(n+1)^q} e^{nH(P)}
\quad \mbox{ (see \cite[Lemma~2.3]{JSCC:Csiszar198100})}.
\]
Note that $\lim_{n \to \infty} q\ln(cn+1)/n = 0$.
All conclusions of the theorem follow immediately.
\end{proofof}

\begin{proofof}{Proposition \ref{pr:SingleCodeSpectrum}}
\begin{IEEEeqnarray*}{rCl}
\max_{\substack{P \in \mathcal{P}_n^*(\mathcal{X}),
 \\ Q \in \mathcal{P}_m(\mathcal{Y})}} \alpha_f(P,Q)
&\ge &\max_{Q \in \mathcal{P}_m(\mathcal{Y})}
 \frac{\spec_f(Q|P_{a^n})}
 {\spec_{\mathcal{Y}^m}(Q)}\\
&\ge &\frac{\max_{Q \in \mathcal{P}_m(\mathcal{Y})}\spec_f(Q|P_{a^n})}
 {\max_{Q \in \mathcal{P}_m(\mathcal{Y})} \spec_{\mathcal{Y}^m}(Q)}\\
&= &\frac{|\mathcal{Y}|^m}{\max_{Q \in \mathcal{P}_m(\mathcal{Y})}
 {m \choose mQ}}\\
&\eqvar{(a)} &\frac{|\mathcal{Y}|^m}
 {\order\left(m^{-\frac{|\mathcal{Y}|-1}{2}} |\mathcal{Y}|^m\right)}\\
&= &\order\left(m^{\frac{|\mathcal{Y}|-1}{2}}\right)
\end{IEEEeqnarray*}
where $a\in\mathcal{X}\setminus\{0\}$ and (a) follows from Stirling's
approximation.
\end{proofof}

\omitted{%
\InputIfFileExists{specext}{}{}
\InputIfFileExists{gleext}{}{}
\InputIfFileExists{omitted}{}{}
\InputIfFileExists{whyspec}{}{}
}

\section*{Acknowledgement}

The authors would like to thank the anonymous reviewers and the
handling editors, M.~Blaum and R.~Fischer, for their very helpful
comments which greatly improved the quality of this paper.

\bibliographystyle{IEEEtran}
\bibliography{IEEEabrv,legs}

\end{document}

%% file: specext.tex

\section{Omitted Material of
Section~\ref{sec:BasicsOfCodeSpectrumApproach}}

\begin{example}\label{ex:HammingCode1}
We start with the binary $[7,4,3]$ Hamming code, the smallest
non-trivial perfect code. 
Here $\mathcal{X}=\mathcal{Y}=\field{2}$, $n=4$, $m=7$.
The encoder $f\colon\field{2}^4\to\field{2}^7$ can be taken as
$f(\seq{x})=\seq{x}\mat{G}_1$ with
\begin{equation}
  \label{eq:Ham(3)}
  \mat{G}_1=
  \begin{pmatrix}
    1&1&0&1&0&0&0\\
    0&1&1&0&1&0&0\\
    0&0&1&1&0&1&0\\
    1&1&0&0&1&0&1
  \end{pmatrix}.
\end{equation}
(We have replaced the last row of the more common cyclic generating
matrix by the complement of the third row.)
This particular choice of $\mat{G}_1$ ensures that $(1111111)$ encodes
a message of weight $2$.
The input-output weight distribution of $f$, counting the number of
message-codeword pairs $\bigl(\seq{x},f(\seq{x})\bigr)$ having fixed
weight pair $(w_1,w_2)\in\{0,1,2,3,4\}\times\{0,1,2,3,4,5,6,7\}$, is
given by the following array (with zero entries omitted):
\begin{equation*}
  \begin{array}{c|cccccccc}
    w_1\backslash w_2&0&1&2&3&4&5&6&7\\\hline
    0&1\\
    1&&&&3&1\\
    2&&&&2&3&&&1\\
    3&&&&1&3\\
    4&&&&1
  \end{array}
\end{equation*}
Up to the normalizing factor $\frac{1}{16}$ this is also the spectrum
of $f$ (since we are in the binary case).
The function $\alpha_f(P,Q)$ or, equivalently, $\alpha_f(w_1,w_2)$
is obtained by dividing each entry of this array by the corresponding
number $\binom{4}{w_1}\binom{7}{w_2}$ (total number of pairs
$(\seq{x},\seq{y})\in\field{2}^4\times\field{2}^7$ having weight pair
$(w_1,w_2)$) and scaling by $2^{11}/2^4=128$.
The numbers $\alpha_f(w_1,w_2)$ are shown in the following table:
\begin{equation*}
\renewcommand{\arraystretch}{1.2}
  \begin{array}{c|cccccccc}
    w_1\backslash w_2&0&1&2&3&4&5&6&7\\\hline
    0&128\\
    1&&&&\frac{96}{35}&\frac{32}{35}\\
    2&&&&\frac{128}{105}&\frac{192}{105}&&&\frac{64}{3}\\
    3&&&&\frac{32}{35}&\frac{96}{35}\\
    4&&&&\frac{128}{35}
  \end{array}
\end{equation*}
The encoder $f$ has been chosen in such a way that it minimizes the
maximum of $\alpha_f(w_1,w_2)$, taken over all $(w_1,w_2)$ with
$w_1\neq 0$.
The corresponding maximum is $\alpha_f(2,7)=\frac{64}{3}$.
\end{example}

As we shall see later, the Hamming code considered in
Example~\ref{ex:HammingCode1} is not a good channel code in the sense
of \eqref{eq:DefinitionOfAsympGoodLCC}, because it contains the
all-one codeword.
This also implies that its associated encoder $f$ cannot have a small
maximum of $\alpha_f$ (over all $(w_1,w_2)$ with $w_1\neq 0$).
Indeed, the optimal encoder $f$ in Example~\ref{ex:HammingCode1} has
$\max_{w_1\neq 0,w_2}\alpha_f(w_1,w_2)=\frac{64}{3}$, which is far
from the lower bound
\[
\frac{2^n}{\max_{0\le k \le n} {n \choose k}}
= \frac{2^7}{{7 \choose 3}}
= \frac{128}{35}
\qquad \mbox{(see Proposition~\ref{pr:SingleCodeSpectrum})}
\]
for binary linear $[7,4]$ codes.
However, this lower bound can be achieved by choosing a different
code, as our next example shows.

\begin{example}\label{ex:[7,4]2}
We extend the binary $[7,3,4]$ simplex code (even-weight subcode of
the Hamming code) by a word of weight $1$ to a linear $[7,4,1]$ code
$C$.
The weight distribution of $C$ is then $A_0=A_1=1$, $A_2=0$, $A_3=4$,
$A_4=7$, $A_5=3$, $A_6=A_7=0$.
The encoder $f\colon\field{2}^4\to\field{2}^7$ is chosen in such a way
that the four codewords of small and large weight (weights $1$ and
$5$) encode words of weight $2$.
This can be done, since these four codewords are linearly dependent.
For example, we can choose $f(\seq{x})=\seq{x}\mat{G}_2$ with
\begin{equation}
\mat{G}_2=
\begin{pmatrix}
  1&1&1&1&0&0&0\\
  1&0&0&0&0&1&1\\
  0&0&1&1&1&1&0\\
  1&1&1&0&0&0&0
\end{pmatrix}.\label{eq:[7,4]2}
\end{equation}
The input-output weight distribution of $f$ is
\begin{equation*}
  \begin{array}{c|ccccccccc}
    w_1\backslash w_2&0&1&2&3&4&5&6&7\\\hline
    0&1\\
    1&&&&2&2\\
    2&&1&&&2&3\\
    3&&&&2&2\\
    4&&&&&1
  \end{array}
\end{equation*}
and the maximum of $\alpha_f$ over all $(w_1,w_2)$ with $w_1\neq 0$ is
\[
\alpha_f(4,4)
= \frac{1/2^4}{{4 \choose 4}{7 \choose 4}/2^{11}}
= \frac{128}{35},
\]
meeting the lower bound as asserted.
Further values close to the lower bound are
$\alpha_f(2,1)=\alpha_f(2,5)=\frac{64}{21}$.
\end{example}

From the perspective of traditional coding theory, it is absurd to
state that a linear code of minimum distance one is better than a
perfect linear code of minimum distance three (and otherwise the same
parameters).  This is mainly because
the code length of our examples is too short.  In fact, as length goes
to infinity, any linear code that has a
linear encoder achieving the lower
bound of Proposition~\ref{pr:SingleCodeSpectrum} (or up to an
exponentially negligible factor) satisfies the asymptotic
Gilbert-Varshamov (GV) bound (see \cite[Remark~4.1]{JSCC:Yang200904}
and Theorem~\ref{th:exGVB}).  Moreover, as proven in
\cite{JSCC:Yang200904}, these linear encoders are universal for
all sources and channels, although the decoder may be
dependent on the source and channel.
We shall dig into this issue in Section~\ref{sec:gleext}, where we
show that encoder \eqref{eq:[7,4]2} is in fact better than encoder
\eqref{eq:Ham(3)} in some sense.

%% file: gleext.tex

\section{Omitted Material of
Section~\ref{sec:ConceptsOfGoodLinearCodes}}
\label{sec:gleext}

For better understanding of the definitions of good linear
encoders, let us review the original requirements of good linear
encoders for lossless source coding, channel coding, and lossless
JSCC, respectively.%
\footnote{This review is merely based on the ideas and results in
 previous papers.
 For technical reasons, the requirements we give here are not the same
 as those in the literature.}

\emph{Lossless source coding} \cite{JSCC:Yang200503}:
A sequence $\seqx{F}$ of random
linear encoders with the asymptotic source transmission rate
$R_s(\seqx{F})$ is said to be $\delta$-asymptotically good for
lossless source coding if for any $\epsilon > 0$ there exists a
sequence of events $A_k \in \mathcal{A}$ such that for sufficiently
large $k$,
\begin{equation}
\pr(A_k)\ge 1-\epsilon,\label{eq:OriginalDefinitionX1OfAsympGoodLSC}
\end{equation}
\begin{equation}
|R_s(F_k) - R_s(\seqx{F})|
\le \epsilon \qquad \forall \omega \in A_k,
\label{eq:OriginalDefinitionX2OfAsympGoodLSC}
\end{equation}
\begin{equation}
\max_{\seq{x}, \hat{\seq{x}}: \seq{x} \ne \hat{\seq{x}}}
 \frac{1}{n_k} \ln \pr\{F_k(\seq{x}) = F_k(\hat{\seq{x}}) | A_k \}
\le -R_s(\seqx{F}) + \delta + \epsilon.
\label{eq:OriginalDefinitionOfAsympGoodLSC}
\end{equation}
The use of event $A_k$ is to exclude some encoders with unwanted rates
or some bad encoders that may have a major impact on the average
performance.
In coding theory such a technique is called ``expurgating code
ensembles''.
Since $F_k$ is linear,
condition~\eqref{eq:OriginalDefinitionOfAsympGoodLSC} is equivalent to
\begin{equation}
\max_{\seq{x}: \seq{x} \ne \seq{0}}
 \frac{1}{n_k} \ln \pr\{\seq{x} \in \ker F_k | A_k \}
\le -R_s(\seqx{F}) + \delta + \epsilon.
\label{eq:OriginalDefinition2OfAsympGoodLSC}
\end{equation}

\emph{Channel coding}
\cite{JSCC:Gallager196800,JSCC:Shulman199909,JSCC:Bennatan200403}:
A sequence $\seqx{F}$ of random linear encoders with the asymptotic
channel transmission rate $R_c(\seqx{F})$ is said to be
$\delta$-asymptotically good for channel coding if for any
$\epsilon > 0$ there exists a sequence of events $A_k \in \mathcal{A}$
such that for sufficiently large $k$,
\begin{equation}
\pr(A_k)\ge 1-\epsilon,\label{eq:OriginalDefinitionX1OfAsympGoodLCC}
\end{equation}
\begin{equation}
|R_c(F_k) - R_c(\seqx{F})|
\le \epsilon \qquad \forall \omega \in A_k,
\label{eq:OriginalDefinitionX2OfAsympGoodLCC}
\end{equation}
\begin{equation}\label{eq:OriginalDefinitionOfAsympGoodLCC}
\max_{\seq{y}, \hat{\seq{y}}: \seq{y} \ne \hat{\seq{y}}}
 \frac{1}{m_k} \ln \left( \frac{\pr\{\seq{y} \in \mathcal{C}_{F_k},
  \hat{\seq{y}} \in \mathcal{C}_{F_k} | A_k\}}
 {\pr\{\seq{y} \in \mathcal{C}_{F_k} | A_k\}
  \pr\{\hat{\seq{y}} \in \mathcal{C}_{F_k} | A_k\}} \right)
\le \delta + \epsilon,
\end{equation}
where $\mathcal{C}_{F_k} \eqdef F_k(\field{q}^{n_k}) + \bar{Y}^{m_k}$,
and $\bar{Y}^{m_k}$ is a uniform random vector on $\field{q}^{m_k}$.
Clearly, for any $f_k \in F_k(\Omega)$,
$\pr\{\seq{y} \in \mathcal{C}_{f_k}\}
 = f_k(\field{q}^{n_k})/|q^{m_k}|$,
so it follows from \eqref{eq:OriginalDefinitionX2OfAsympGoodLCC} that
\[
\left| \frac{1}{m_k} \ln \pr\{\seq{y} \in \mathcal{C}_{F_k} | A_k\}
 + \ln q - R_c(\seqx{F}) \right|
\le \epsilon
\]
for all $\seq{y} \in \field{q}^{m_n}$.
Because $F_k$ is linear, we also have
\begin{IEEEeqnarray*}{rCl}
\pr\{\seq{y}\in\mathcal{C}_{F_k}, \hat{\seq{y}}\in\mathcal{C}_{F_k}
 | A_k\}
&= &\pr\{\seq{y} \in \mathcal{C}_{F_k},
 \hat{\seq{y}} - \seq{y} \in F_k(\field{q}^{n_k}) | A_k\} \\
&= &\sum_{f_k} \pr\{F_k = f_k | A_k\}
 \av[1\{\seq{y} \in \mathcal{C}_{f_k}, \hat{\seq{y}} - \seq{y}
 \in f_k(\field{q}^{n_k})\}] \\
&= &\sum_{f_k} \pr\{F_k = f_k | A_k\}
 \pr\{\seq{y} \in \mathcal{C}_{f_k}\}
 1\{\hat{\seq{y}} - \seq{y} \in f_k(\field{q}^{n_k})\}.
\end{IEEEeqnarray*}
This together with \eqref{eq:OriginalDefinitionX2OfAsympGoodLCC} gives
\[
\biggl| \frac{1}{m_k} \ln \frac{\pr\{\seq{y} \in \mathcal{C}_{F_k},
 \hat{\seq{y}} \in \mathcal{C}_{F_k} | A_k\}}
 {\pr\{\hat{\seq{y}}-\seq{y} \in F_k(\field{q}^{n_k}) | A_k\}} + \ln q
 - R_c(\seqx{F}) \biggr|
\le \epsilon,
\]
so that condition \eqref{eq:OriginalDefinitionOfAsympGoodLCC} can be
rewritten as
\begin{equation}
\max_{\seq{y}: \seq{y} \ne \seq{0}} \frac{1}{m_k}
 \ln \pr\{\seq{y} \in F_k(\field{q}^{n_k}) | A_k\}
\le R_c(\seqx{F}) - \ln q + \delta + \epsilon.
\label{eq:OriginalDefinition2OfAsympGoodLCC}
\end{equation}

\emph{Lossless JSCC} \cite{JSCC:Yang200904}:
A sequence $\seqx{F}$ of random
linear encoders is said to be $\delta$-asymptotically good for
lossless JSCC if
\begin{equation}
\limsup_{k \to \infty}
 \max_{\scriptstyle \seq{x}, \hat{\seq{x}}: \seq{x} \ne \hat{\seq{x}}
  \atop \scriptstyle \seq{y}, \hat{\seq{y}}}
 \frac{1}{n_k} \ln \left(
 \frac{\pr\{\mathcal{F}_{F_k}(\seq{x}) = \seq{y},
  \mathcal{F}_{F_k}(\hat{\seq{x}}) = \hat{\seq{y}}\}}
  {\pr\{\mathcal{F}_{F_k}(\seq{x}) = \seq{y}\}
   \pr\{\mathcal{F}_{F_k}(\hat{\seq{x}}) = \hat{\seq{y}}\}} \right)
\le \delta,
\label{eq:OriginalDefinitionOfAsympGoodLSCC}
\end{equation}
where $\mathcal{F}_{F_k}(\seq{x}) \eqdef F_k(\seq{x}) + \bar{Y}^{m_k}$
and $\bar{Y}^{m_k}$ is a uniform random vector on $\field{q}^{m_k}$.
By the arguments in the proof of
\cite[Proposition~2.6]{JSCC:Yang200904}, we have the following
alternative condition:
\begin{equation}\label{eq:OriginalDefinition2OfAsympGoodLSCC}
\limsup_{k \to \infty}
 \max_{\scriptstyle \seq{x}: \seq{x} \ne \seq{0} \atop \scriptstyle
  \seq{y}}
 \frac{1}{n_k} \ln \left( q^{m_k} \pr\{F_k(\seq{x})=\seq{y}\} \right)
\le \delta.
\end{equation}

The requirements above are fundamental, but
are not easy and convenient for use.  The next three propositions show
that spectra of linear encoders can serve as alternative criteria for
good linear encoders, and that the uniform random permutation is a
useful tool for constructing good linear encoders.

\begin{proposition}\label{pr:KernelSpectrumCondition}
Let $\seqx{F}$ be a sequence of random linear encoders with the
asymptotic source transmission rate $R_s(\seqx{F})$.
If $\seqx{F}$ satisfies the kernel-spectrum
condition~\eqref{eq:DefinitionOfAsympGoodLSC},
then the sequence of random linear encoders
$\rtilde{F_k} = F_k \circ \Sigma_{n_k}$ is $\delta$-asymptotically
good for lossless source coding.
\end{proposition}
\begin{proof}
For any $\epsilon > 0$, define the sequence of events
\[
A_k
= \left\{ \omega \in \Omega:
 \left| R_s(\rtilde{F}_k) - R_s(\seqx{F}) \right|
 \le \frac{\epsilon}{3} \right\}.
\]
It is clear that $\lim_{k\to\infty} \pr(A_k) = 1$, so that conditions
\eqref{eq:OriginalDefinitionX1OfAsympGoodLSC} and
\eqref{eq:OriginalDefinitionX2OfAsympGoodLSC} hold.
Furthermore, we have
\begin{IEEEeqnarray*}{rCl}
\max_{\seq{x}: \seq{x} \ne 0^{n_k}} \frac{1}{n_k}
 \ln \pr\{\seq{x} \in \ker \rtilde{F}_k | A_k\}
&= &\max_{\seq{x}: \seq{x} \ne 0^{n_k}} \frac{1}{n_k} \ln
 \av\left[\left.\frac{q^{n_k} 1\{\seq{x} \in \ker\rtilde{F}_k\}}
 {|\rtilde{F}_k(\field{q}^{n_k})| |\ker \rtilde{F}_k|}
 \right| A_k \right]\\
&\le &\max_{\seq{x}: \seq{x} \ne 0^{n_k}} \frac{1}{n_k} \ln
 \av\left[\left.\frac{q^{n_k} 1\{\seq{x} \in \Sigma_{n_k}(\ker F_k)\}}
 {e^{n_k(R_s(\seqx{F}) - \frac{\epsilon}{3})} |\ker F_k|}
 \right| A_k \right]\\
&\le &-R_s(\seqx{F}) + \frac{\epsilon}{3}\\
& &\breakop{+} \max_{\seq{x}: \seq{x} \ne 0^{n_k}} \frac{1}{n_k}
 \ln \left( \frac{q^{n_k}}{\pr(A_k)} \av\left[\frac{1\{\seq{x}
 \in \Sigma_{n_k}(\ker F_k)\}}{|\ker F_k|} \right] \right)\\
&\levar{(a)} &-R_s(\seqx{F}) + \frac{2\epsilon}{3}
 + \max_{P \in \mathcal{P}_{n_k}^*} \frac{1}{n_k}
 \ln \alpha_{\ker F_k}(P)\\
&\levar{(b)} &-R_s(\seqx{F}) + \delta + \epsilon
\end{IEEEeqnarray*}
for sufficiently large $k$, where (a) follows from
Proposition~\ref{pr:GeneralSpectrumPropertyOfSets} and
$\lim_{k\to\infty} \pr(A_k) = 1$, and (b) follows from
\eqref{eq:DefinitionOfAsympGoodLSC}.
This concludes \eqref{eq:OriginalDefinition2OfAsympGoodLSC} and hence
proves the proposition.
\end{proof}

\begin{proposition}\label{pr:ImageSpectrumCondition}
Let $\seqx{F}$ be a sequence of random linear encoders with the
asymptotic channel transmission rate $R_c(\seqx{F})$.
If $\seqx{F}$ satisfies the image-spectrum
condition~\eqref{eq:DefinitionOfAsympGoodLCC},
then the sequence of random linear encoders $\Sigma_{m_k} \circ F_k$
is $\delta$-asymptotically good for channel coding.
\end{proposition}

\begin{proof}
Use argument similar to that of
Proposition~\ref{pr:KernelSpectrumCondition}.
\end{proof}

\begin{proposition}\label{pr:JointSpectrumCondition}
Let $\seqx{F}$ be a sequence of random linear encoders.
If $\seqx{F}$ satisfies the joint-spectrum
condition~\eqref{eq:DefinitionOfAsympGoodLSCC},
then the sequence of random linear encoders
$\tilde{F}_k = \Sigma_{m_k} \circ F_k \circ \Sigma_{n_k}$ is
$\delta$-asymptotically good for lossless JSCC.
\end{proposition}

\begin{proof}
Apply \citeSpectrumProperty.
\end{proof}

\begin{remark}
It should be noted that conditions
\eqref{eq:DefinitionOfAsympGoodLSC}--\eqref{eq:DefinitionOfAsympGoodLSCC}
are only sufficient but not necessary.
In other words, there may exist other good (random) linear encoders.
Note, for example, that condition \eqref{eq:DefinitionOfAsympGoodLCC}
requires that the average spectrum of $F_k(\field{q}^{n_k})$ should be
uniformly close to the spectrum of $\field{q}^{m_k}$.
This requirement is obviously very strict.
For instance, all linear codes containing the all-one vector are
excluded by this condition with $\delta=0$.
To take into account such cases, we would need to use some
sophisticated bounding techniques (see e.g.,
\cite{JSCC:Bennatan200403}), which are still not mature or even
infeasible (Example~\ref{ex:comp2}) for more complicated cases.
So for simplicity of analysis, we choose conditions
\eqref{eq:DefinitionOfAsympGoodLSC}--\eqref{eq:DefinitionOfAsympGoodLSCC}.
At least to some extent, this choice is reasonable.
This is because $\rlccode{n_k, m_k}$ is asymptotically SCC-good (resp.
SC- and CC-good), and hence it can be shown by Markov's inequality
that most linear encoders are asymptotically SCC-good (resp. SC- and
CC-good).
\end{remark}

\begin{remark}\label{re:Condition.A}
Conditions
\eqref{eq:DefinitionOfAsympGoodLSC}--\eqref{eq:DefinitionOfAsympGoodLSCC}
all apply to sequences of random encoders.
For readers not familiar with probabilistic analysis we provide some
further explanation:
First, a deterministic encoder is a special random encoder, so
conditions
\eqref{eq:DefinitionOfAsympGoodLSC}--\eqref{eq:DefinitionOfAsympGoodLSCC}
also apply to a sequence of deterministic encoders.
Second, for example, if a sequence of random linear encoders is
$\delta$-asymptotically SCC-good, then there exists a sequence of
sample encoders that is $\delta$-asymptotically SCC-good
(\cite[Proposition~4.1]{JSCC:Yang200904}).
The proof of this fact relies on Markov's inequality and the fact that
the size of the set over which the maximum (in
\eqref{eq:DefinitionOfAsympGoodLSCC}) is taken is a polynomial
function of $m_k$ and $n_k$.
In fact, by the same argument, we can obtain a stronger result:
\[
\lim_{k\to\infty}
\pr\left\{\max_{P \in \mathcal{P}_{n_k}^*, Q \in \mathcal{P}_{m_k}}
 \frac{1}{n_k} \ln \frac{\spec_{F_k}(P,Q)}
 {\spec_{\field{q}^{n_k}\times\field{q}^{m_k}}(P,Q)}
\le \delta+\epsilon\right\}
= 1
\]
for any $\epsilon>0$.
Third, as we shall see, a typical sample encoder of good random
encoders (in the sense of
\eqref{eq:DefinitionOfAsympGoodLSC}--\eqref{eq:DefinitionOfAsympGoodLSCC})
has a fundamental property, which is characterized by the so-called
entropy weight, where ``typical'' means that the set of such encoders
contains most of the probability mass.
Since the proof is again a simple application of Markov's inequality,
we will leave it to the reader as an exercise.
\end{remark}

For better understanding of these three kinds of good linear encoders,
let us take a look from another perspective.
We define the entropy weight $H(\seq{x})$ of a vector
$\seq{x} \in \field{q}^n$ by
$H(\seq{x}) \eqdef \ln{n \choose nP_{\seq{x}}}$, and the normalized
entropy weight as $h(\seq{x})=H(\seq{x})/n$.
Recalling the identity
\[
\lim_{n\to\infty} \frac{1}{n}\ln{n \choose nP} = H(P)
\qquad \mbox{(cf. \cite[Lemma~2.3]{JSCC:Csiszar198100})},
\]
we have $h(\seq{x}) \approx H(P_{\seq{x}})$, so $H(P_{\seq{x}})$ may
be used as a substitute for $h(\seq{x})$ in the asymptotic sense, and
accordingly we call $H(P_{\seq{x}})$ the asymptotic normalized entropy
weight of $\seq{x}$.
Now suppose $\delta = 0$ and $k$ is large enough.
From \eqref{eq:DefinitionOfAsympGoodLSC} and
\cite[Proposition~2.1]{JSCC:Yang200904},
it follows that a typical linear encoder $f_k$ of an asymptotically
SC-good $F_k$ satisfies
\[
\limsup_{k\to\infty} \frac{1}{n_k}
 \ln |\mathcal{T}_P^{n_k} \cap \ker f_k|
\le H(P) - R_s(\seqx{F}),
\]
where $P \in \mathcal{P}_{n_k}^*$.
This implies that $\ker f_k$ does not contain nonzero vectors of
normalized entropy weight less than $R_s(\seqx{F})$.
Similarly, it follows from \eqref{eq:DefinitionOfAsympGoodLCC} that a
typical linear encoder $f_k$ of an asymptotically CC-good $F_k$
satisfies
\[
\limsup_{k\to\infty} \frac{1}{m_k}
 \ln |\mathcal{T}_Q^{m_k} \cap f_k(\field{q}^{n_k})|
\le H(Q) + R_c(\seqx{F}) - \ln q,
\]
where $Q \in \mathcal{P}_{m_k}^*$.
The case of \eqref{eq:DefinitionOfAsympGoodLSCC} is more
complicated (cf. Theorem~\ref{th:exGVB}).
It follows that a typical linear encoder $f_k$ of an asymptotically
SCC-good $F_k$ satisfies
\[
\limsup_{k\to\infty} \frac{1}{n_k} \ln
 |(\mathcal{T}_P^{n_k}\times\mathcal{T}_Q^{m_k})\cap\mathrm{rl}(f_k)|
\le H(P) + \frac{H(Q)-\ln q}{\overline{R}(\seqx{F})},
\]
where $P \in \mathcal{P}_{n_k}^*$ and $Q \in \mathcal{P}_{m_k}$.
This implies that $\mathrm{rl}(f_k)$ does not contain pairs of
low-entropy-weight vectors except $(0^{n_k}, 0^{m_k})$.
In other words, all nonzero input vectors of low entropy weight must
be mapped to high-entropy-weight vectors.
In particular, all small-weight vectors, as well as the all-$x$
vectors for $x \in \fieldstar{q}$, must be mapped to vectors of
close-to-uniform type and hence (Hamming) weight around $m_k(1-1/q)$.

From the perspective of entropy weight, many linear codes with large
minimum distance are not good because they contain vectors of low
entropy weight, for example the all-one vector.
In fact, entropy weight guides us to a subclass of linear codes with
not only large minimum distance but also large minimum entropy
distance \cite{JSCC:Yang201301}.
At this point, it is appropriate to mention a bound that includes the
above three properties.

\begin{theorem}[\cite{JSCC:Yang200904}]\label{th:exGVB}
For any $r>0$, there is a sequence $\seqx{f}$ of linear encoders
$f_k: \field{q}^{n_k} \to \field{q}^{m_k}$ such that $R(\seqx{f})=r$
and
\begin{equation}\label{eq:ExtendedGVB}
\liminf_{k\to\infty} \min_{\seq{x}\ne 0^{n_k}}
 (H(P_{\seq{x}})R(f_k) + H(P_{f_k(\seq{x})}))
\ge \ln q
\end{equation}
where $f_k$ is injective for $r \le 1$ and surjective for $r \ge 1$.
\end{theorem}

Theorem~\ref{th:exGVB} is a special case of
\cite[Theorem~4.1]{JSCC:Yang200904} for $\delta=0$
(but with some simple improvements),
and may be regarded as an extension of the asymptotic GV bound.
For $r>1$, we can take $f_k(\seq{x}) = \seq{0}^{m_k}$ and then get
\[
\liminf_{k\to\infty} \min_{\seq{x}\in\ker f_k\setminus\{0^{n_k}\}}
 H(P_{\seq{x}})
\ge r^{-1} \ln q = R_s(\seqx{f}).
\]
If $r<1$, since $H(P_{\seq{x}}) \le \ln q$, we have
\[
\liminf_{k\to\infty} \min_{\seq{x}\ne 0^{n_k}} H(P_{f_k(\seq{x})})
\ge \ln q - r\ln q
= \ln q - R_c(\seqx{f}).
\]
Although the left-hand side of \eqref{eq:ExtendedGVB} provides a
refinement of traditional minimum Hamming distance, it
still cannot ensure good coding performance.
In fact, condition \eqref{eq:DefinitionOfAsympGoodLSCC} (resp.,
\eqref{eq:DefinitionOfAsympGoodLSC} and
\eqref{eq:DefinitionOfAsympGoodLCC}) requires the joint (resp., kernel
and image) spectrum of the encoder to be close to the average joint
(resp., kernel and image) spectrum of all linear encoders of the same
coding rate (cf. Remark~\ref{re:PDProperty1}).
We may call such encoders random-like encoders.
Since minimum Hamming weight and minimum entropy weight only
focus on one or two specific points of weight distribution or
spectrum, linear encoders designed under these criteria cannot be
universally good (see Example~\ref{ex:comp2}).

So far, we have extensively discussed
criteria of good linear encoders in an abstract manner.
A comparison between the linear encoders \eqref{eq:Ham(3)} and
\eqref{eq:[7,4]2} in Examples~\ref{ex:HammingCode1} and
\ref{ex:[7,4]2} will help the reader understand why we care
about joint spectra and why about the whole shape of the spectrum.
The following example shows that the joint spectrum has such a great
impact on the performance of lossless JSCC, that a perfect linear code
of minimum distance three may perform worse than a linear code of
minimum distance one (and otherwise the same parameters) if the
generator matrix is not carefully chosen.

\begin{example}\label{ex:comp1}
  Consider a zero-one binary independent and identically distributed
  (IID) source with probability $p_1$ of symbol $1$ and a binary
  symmetric channel (BSC) with crossover probability $p_2$.
  Further consider a coding scheme that transmits four source symbols
  by utilizing the channel seven times.
  The scheme is based on Fig.~\ref{fig:Scheme1} with the quatization
  module removed, where the linear encoder used is
  either $\mat{G}_1$ or $\mat{G}_2$ defined by
  \eqref{eq:Ham(3)} and \eqref{eq:[7,4]2}, respectively.  Because the
  source and channel are both IID and the channel is an additive noise
  channel over $\field{2}$, the two random interleavers and the random
  vector module in Fig.~\ref{fig:Scheme1} can all be omitted.  Since
  the code length is short, we can easily compute the
  exact decoding error probabilities under maximum a posteriori (MAP)
  decoding.  For further comparison, we also include results for the
  linear encoders
\begin{equation}
  \label{eq:Ham(3)+}
  \mat{G}_3=
  \begin{pmatrix}
    1&1&0&1&0&0&0\\
    0&1&1&0&1&0&0\\
    0&0&1&1&0&1&0\\
    1&1&1&1&1&1&1
  \end{pmatrix}
\end{equation}
and
\begin{equation}
\mat{G}_4=
\begin{pmatrix}
  0&0&0&1&0&0&0\\
  1&0&0&0&0&1&1\\
  0&0&1&1&1&1&0\\
  1&1&1&0&0&0&0
\end{pmatrix}\label{eq:[7,4]2+},
\end{equation}
which yield the same linear code as $\mat{G}_1$ and $\mat{G}_2$,
respectively.
\begin{figure}[htbp]
\centering
\includegraphics{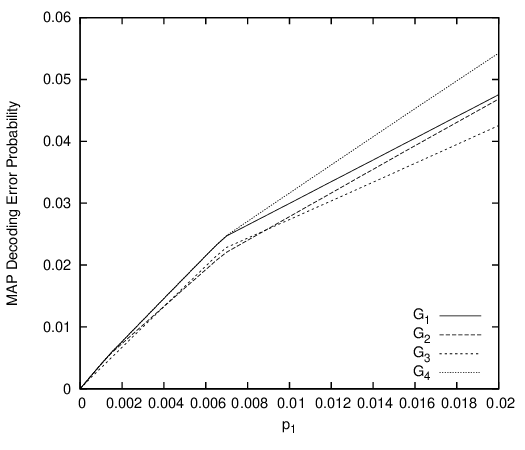}
\caption{MAP decoding error probability versus $p_1$ comparison among
 $\mat{G}_i$ for $p_1\in(0,0.02)$ and $p_2=0.16$.}
\label{fig:comp}
\end{figure}
Fig.~\ref{fig:comp} compares the performance of $\mat{G}_i$
($1\leq i\leq 4$) for $p_1\in(0,0.02)$
and $p_2=0.16$.  For example, the order of their performance at
$p=0.008$, from best to worst, is $\mat{G}_2$, $\mat{G}_3$,
$\mat{G}_1$, and $\mat{G}_4$.  There are two interesting facts
to be learned.
First, $\mat{G}_3$ outperforms $\mat{G}_1$ and $\mat{G}_2$ outperforms
$\mat{G}_4$.  This implies that the choice of generator matrix does
have an impact on JSCC performance.
Second, $\mat{G}_2$ beats $\mat{G}_1$ for all $p_1\in(0,0.02)$.
This surprising result shows that in JSCC, a perfect code of minimum
distance three may perform worse than a code of minimum distance one
if the generator matrix is not chosen properly.
\end{example}

In order to explain this phenomenon, we shall introduce the concept of
pairwise discrimination, which forms the key idea of
lossless JSCC and will now be expressed in an intuitive but less strict
manner.
Recall the concept of a typical set (cf. \cite{JSCC:Cover199100}).
Let $X^n=(X_1, X_2, \ldots, X_n)$ be a random $n$-dimensional vector
over $\field{q}$.
The typical set $A_\epsilon^{(n)}$ of $X^n$ is defined as the set of
all vectors $\seq{x} \in \field{q}^n$ satisfying
\begin{equation}\label{eq:AEP}
e^{-n(h(X^n)+\epsilon)}
\le \pr_{X^n}(\seq{x})
\le e^{-n(h(X^n)-\epsilon)},
\end{equation}
where $h(X^n)\eqdef H(X^n)/n = -n^{-1}\av[\ln\pr_{X^n}(X^n)]$ and
$\pr_{X^n}(\seq{x}) \eqdef \pr\{X^n = \seq{x}\}$.
Usually, we add some conditions to ensure that $h(X^n)$
converges to the so-called entropy rate as $n \to \infty$ and that
$-n^{-1}\ln\pr_{X^n}(X^n)$ converges to the entropy rate almost
surely.
But here, we just borrow the concept
and do not rigorously justify every detail.
Two distinct $n$-dimensional vectors are
considered to be discriminable if at least one of them is not in the
typical set $A_\epsilon^{(n)}$.
In a more intuitive manner, we may think of two distinct vectors
indiscriminable if both of them are elements of the
high-probability set $B_\epsilon^{(n)}\eqdef\{\seq{x}\in\field{q}^n:
\pr_{X^n}(\seq{x})\ge e^{-n(h(X^n)+\epsilon)}\}$.

Roughly speaking, the art of lossless JSCC is to focus on the
most probable source vectors (with a high total probability) and to
choose appropriate channel input vectors for them so that all these
source vectors, combined with any channel output vector in the
high-probability set, are pairwise discriminable.

Keeping this idea in mind, we continue the discussion of
Example~\ref{ex:comp1}.
We note that the most probable source vectors are the zero vector
and all weight-one vectors, whose total probability is
$(1-p_1)^7+7p_1(1-p_1)^6\ge 0.992$ for $p_1\in(0, 0.02)$.
Because the zero vector owns the dominant probability, we only need to
pay attention to the pairs consisting of the zero vector and a vector
of weight one.
Other pairs consisting of two weight-one vectors may be ignored.
Therefore, the performance of an injective linear encoder is mainly
determined by its output for weight-one input.
The output weight distribution of $\mat{G}_i$ for weight-one input is
listed in Table~\ref{tab:comp1}.
\begin{table}[htbp]
\caption{Output weight distributions of $\mat{G}_i$ for weight-one
 input}
\label{tab:comp1}
\centering
\begin{tabular}{ccccccccc}\hline
 &$0$ &$1$ &$2$ &$3$ &$4$ &$5$ &$6$ &$7$ \\\hline
 $\mat{G}_1$ & & & &$3$ &$1$ \\
 $\mat{G}_2$ & & & &$2$ &$2$ \\
 $\mat{G}_3$ & & & &$3$ & & & &$1$ \\
 $\mat{G}_4$ & &$1$ & &$2$ &$1$ \\\hline
\end{tabular}
\end{table}
In order to make the zero vector and all weight-one vectors
discriminable at channel
output, a good strategy is to map these source vectors to channel
input vectors as far from each other as possible in terms of
Hamming distance (since the channel is a BSC).  Therefore, we shall
get a boost in performance if we map weight-one vectors
to vectors of weight as large as possible.  Comparing the output
weight distributions of $\mat{G}_i$ in Table~\ref{tab:comp1},
especially for weights $\ge 4$, it is easy to figure out that
$\mat{G}_2$ is better than $\mat{G}_1$, $\mat{G}_3$ better than
$\mat{G}_1$, and $\mat{G}_4$ worse than $\mat{G}_2$.  The comparison
between $\mat{G}_2$ and $\mat{G}_3$ is slightly more complicated,
because one has two vectors of weight $4$ while the other has one
vector of weight $7$.
This explains why $\mat{G}_2$ and $\mat{G}_3$ have almost the same
performance for small $p$, as shown in Fig.~\ref{fig:comp}.

In Example~\ref{ex:comp1} we successfully explained the
importance of the joint spectrum of a linear encoder, but in
this case choosing the all-one vector as a codeword is not a bad idea,
which seems contrary to facts about entropy weight, that is, the
all-one vector is of zero entropy weight and hence must be avoided.
While the viewpoint of minimum Hamming distance is very appropriate
for coding over a BSC, it is not a good measure
for designing universal linear encoders.  The next example shows that
a binary linear code containing the all-one vector may have very bad
performance for some special channels, even if it is a perfect
code.

\begin{example}\label{ex:comp2}
For any nonzero $\seq{x}_0 \in \field{q}^n$, define an additive noise
channel $J_{\seq{x}_0}: \field{2}^n \to \field{2}^n$ by
$\seq{x} \mapsto \seq{x} + \seq{N}$, where $\seq{N}$ is a random noise
with distribution
\[
\pr\{\seq{N} = \seq{x}\} \eqdef
\begin{cases}
0.5, &\mbox{for $\seq{x} = \seq{0}$} \\
0.5, &\mbox{for $\seq{x} = \seq{x}_0$} \\
0, &\mbox{otherwise}.
\end{cases}
\]
Clearly, the capacity of $J_{\seq{x}_0}$ is $(n-1)\ln 2$, independent
of the choice of $\seq{x}_0$.
Now consider a channel coding scheme based on Fig.~\ref{fig:Scheme1}
with the quantization module removed.
It transmits four bits over a vector channel $J_{1^\ell 0^{7-\ell}}$
($\ell=1$, $2$, \ldots, $7$).
The linear encoder used is $\mat{G}_1$ or $\mat{G}_2$ defined by
\eqref{eq:Ham(3)} and \eqref{eq:[7,4]2}, respectively.
The interleaver before linear encoder and the random vector module can
be omitted because the source is assumed to be uniform and the channel
noise is additive.
It is easy to figure out the decoding error probability.
The trick is to check whether the noise vector $1^\ell 0^{7-\ell}$
hits a codeword of the linear code.
If it misses, the transmitted information can be decoded successfully;
otherwise, the information can only be guessed with error probability
$\frac{1}{2}$.
Owing to the random interleaver after linear encoder, we should
compute the decoding error probability for each possible interleaver
and then compute their average.
Accordingly, for channel $J_{1^\ell 0^{7-\ell}}$, the decoding error
probability of a linear code is $n_\ell/(2{7\choose\ell})$ where
$n_\ell$ is the number of codewords of weight $\ell$.
Table~\ref{tab:comp2} lists the decoding error probability of
$\mat{G}_1$ and $\mat{G}_2$ for $\ell=1$, $2$, \ldots, $7$.
\begin{table}[htbp]
\caption{Decoding error probability of $\mat{G}_i$ for $\ell=1$, $2$,
 \ldots, $7$}\label{tab:comp2}
\centering\renewcommand{\arraystretch}{1.25}
\begin{tabular}{*{8}{c}}\hline
 &$1$ &$2$ &$3$ &$4$ &$5$ &$6$ &$7$ \\\hline
 $\mat{G}_1$ &$0$ &$0$ &$\frac{1}{10}$ &$\frac{1}{10}$ &$0$ &$0$
  &$\frac{1}{2}$ \\
 $\mat{G}_2$ &$\frac{1}{14}$ &$0$ &$\frac{2}{35}$ &$\frac{1}{10}$
  &$\frac{1}{14}$ &$0$ &$0$ \\\hline
\end{tabular}
\end{table}
Note that for $\ell=7$, the performance of $\mat{G}_1$ is very bad.
This is because there is only one vector of weight $7$ and hence the
random interleaver cannot help the codeword avoid being hit by noise.
Note that this issue cannot be resolved by simply increasing the code
length, so any binary linear code containing the all-one vector
performs bad over channel $J_{1^n}$.
\end{example}

The reader may argue that Example~\ref{ex:comp2} is too special and
that random-like encoders perhaps do not work in certain examples.
The fact that follows will show that a random-like encoder defined by
\eqref{eq:DefinitionOfAsympGoodLSCC} is universally good in the
asymptotic sense.
We continue to utilize the concept of pairwise discrimination as well
as typical set in a less strict manner.

\begin{remark}\label{re:PDProperty1}
Consider a pair $(X^n, Y^m)$ of random vectors and its typical set
$A_\epsilon^{(n,m)}$.
Let $f:\field{q}^n \to \field{q}^m$ be a linear encoder whose joint
spectrum is approximately
\[
\spec_{\field{q}^n}(P)\spec_{\field{q}^m}(Q)
= q^{-(m+n)} {n \choose nP}{m \choose mQ}
\quad \mbox{(cf. \eqref{eq:DefinitionOfAsympGoodLSCC})}
\]
for $P\in\mathcal{P}_n^*$ and $Q\in\mathcal{P}_m$.
It follows from \citeSpectrumProperty{} that for any distinct pairs
$(\seq{x}, \seq{y})$ and $(\seq{x}', \seq{y}')$ in
$A_\epsilon^{(n,m)}$,
\[
\pr\{\tilde{f}(\seq{x}')=\seq{y}'|\tilde{f}(\seq{x}) = \seq{y}\}
= \pr\{\tilde{f}(\seq{x}'-\seq{x})=\seq{y}'-\seq{y}
 |\tilde{f}(\seq{x})=\seq{y}\}
= \left\{\begin{array}{ll}
0, &\mbox{for $\seq{x}=\seq{x}'$}\\
q^{-m}, &\mbox{otherwise}
\end{array}\right.
\]
so that
\begin{equation}\label{eq:PDP1.eq1}
\pr\left\{|A_\epsilon^{(n,m)}\cap\mathrm{rl}(\tilde{f})| > 1
 \Big| \tilde{f}(\seq{x}) = \seq{y}\right\}
\le q^{-m} |A_\epsilon^{(n,m)}|.
\end{equation}
Note that $|A_\epsilon^{(m,n)}|$ is bounded above by
$e^{H(X^n, Y^m)+n\epsilon}$ (cf. \eqref{eq:AEP}), so if
$m \ln q > H(X^n,Y^m) + n\epsilon$, the probability of
$\mathrm{rl}(\tilde{f})$ containing other members of
$A_\epsilon^{(m,n)}$ given $(\seq{x}, \tilde{f}(\seq{x}))$ being a
member of $A_\epsilon^{(m,n)}$ is asymptotically negligible.
In other words, with high probability, the pair
$(\seq{x}, \tilde{f}(\seq{x}))$ is pairwise discriminable with every
other pair in $\mathrm{rl}(\tilde{f})$.
Note that \eqref{eq:PDP1.eq1} does not depend on the probability
distribution of $(X^n, Y^m)$, but only on their joint entropy.

What is the use of \eqref{eq:PDP1.eq1}?
Imagine we are transmitting source vector $X^n$ over a channel and
suppose that the current sample drawn from $X^n$ is $\seq{x}$.
If we send $\tilde{f}(\seq{x})$ over the channel, then after receiving
the channel output $\seq{z}$, using knowledge of the source and
channel, we get the a posteriori information about the channel input,
identified with $Y^m=Y^m(\seq{z})$.
Combining it with the a priori knowledge of the source, we obtain a
pair of random vectors, $(X^n, Y^m)$.
In a typical case, the pair $(\seq{x}, \tilde{f}(\seq{x}))$ must be a
member of the typical set $A_\epsilon^{(n,m)}$ of $(X^n, Y^m)$, so if
\begin{equation}\label{eq:PDP1.eq2}
m\ln q > H(X^n, Y^m) + n\epsilon
\end{equation}
for some $\epsilon>0$, we can decode successfully with high
probability by guessing the unique typical pair.
To illustrate this further, let us consider two special cases.

First, we suppose that the channel is noiseless, so that
$H(X^n, Y^m) = H(X^n)$.
Condition \eqref{eq:PDP1.eq2} then becomes
\[
\frac{m}{n}\ln q > \frac{1}{n} H(X^n) + \epsilon,
\]
a familiar condition for the achievable rate of lossless source
coding.

Second, we suppose that the source is uniformly distributed (i.e.,
channel coding), so we can assume that
$H(X^n, Y^m) = n\ln q + H(Y^m)$.
Condition \eqref{eq:PDP1.eq2} then becomes
\[
n\ln q < m\ln q - H(Y^m) - n\epsilon
\]
and further,
\[
R_c(f)
\le \frac{n}{m} \ln q
< \ln q - \frac{1}{m} H(Y^m) - R(f)\epsilon.
\]
Note that $\ln q - \frac{1}{m} H(Y^m)$ has the same form as the
capacity formula of those channels whose capacity is achieved by the
uniform input probability distribution.  In fact, if we add a random
vector module and a quantization module (as depicted in
Fig.~\ref{fig:Scheme1}) to simulate the capacity-achieving input
probability distribution of a given channel, we can eventually obtain
a capacity-achieving coding scheme, but we shall not
delve further into this because we are interested in its
relevance for coding rather than its nature as a problem
of information theory.
\end{remark}

In Remark~\ref{re:PDProperty1} we showed that an asymptotically
SCC-good linear encoder is universally good for JSCC.
By a similar argument, we can also show that other kinds of
random-like encoders, i.e., those defined by
\eqref{eq:DefinitionOfAsympGoodLSC} and
\eqref{eq:DefinitionOfAsympGoodLCC}, are also universally good for
lossless source coding and channel coding, respectively.

%% file: omitted.tex

\section{Proofs of Results in Section
 \ref{subsec:SpectraWithPartition}}
\label{subsec:ProofOfSpectraWithPartition}

\begin{proofof}{Proposition~\ref{pr:GeneralSpectrumPropertyOfSets}}
It is clear that, for any $\seq{x}$ and $\hat{\seq{x}}$ satisfying
$P^{\mathcal{U}}_{\seq{x}} = P^{\mathcal{U}}_{\hat{\seq{x}}}$,
\begin{IEEEeqnarray*}{rCl}
\av\left[ \frac{1\{\seq{x}\in\Sigma_{\mathcal{U}}(A)\}}{|A|} \right]
&= &\av\left[ \frac{1\{\seq{x} \in \Sigma_{\mathcal{U}}(A)\}}
 {|\Sigma_{\mathcal{U}}(A)|} \right]\\
&\eqvar{(a)} &\av\left[ \frac{1\{\hat{\seq{x}} \in
 \Sigma_{\mathcal{U}}(A)\}}{|\Sigma_{\mathcal{U}}(A)|} \right]\\
&= &\av\left[ \frac{1\{\hat{\seq{x}}\in\Sigma_{\mathcal{U}}(A)\}}{|A|}
 \right]
\end{IEEEeqnarray*}
where (a) follows from the fact that the distribution of
$\Sigma_{\mathcal{U}}(A)$ is invariant under any permutation in
$\symmetricgroup{\mathcal{U}}$.
Then it follows that
\begin{IEEEeqnarray*}{rCl}
\av\left[ \frac{1\{\seq{x}\in\Sigma_{\mathcal{U}}(A)\}}{|A|} \right]
&= &\frac{1}{|\mathcal{T}_{P^{\mathcal{U}}_{\seq{x}}}|}
 \sum_{\hat{\seq{x}} \in \mathcal{T}_{P^{\mathcal{U}}_{\seq{x}}}}
 \av\left[ \frac{1\{\hat{\seq{x}}\in\Sigma_{\mathcal{U}}(A)\}}{|A|}
 \right]\\
&= &\frac{1}{\prod_{U\in\mathcal{U}} {|U| \choose |U| P_{x_U}}}
 \av\left[ \frac{\left| A\cap\mathcal{T}_{P^{\mathcal{U}}_{\seq{x}}}
 \right|}{|A|} \right]\\
&\eqvar{(a)} &q^{-n} \alpha_A(P^{\mathcal{U}}_{\seq{x}})
\end{IEEEeqnarray*}
where (a) follows from \cite[Proposition~2.1]{JSCC:Yang200904} and the
definition of $\mathcal{U}$-spectrum.
This proves \eqref{eq:GeneralSpectrumPropertyOfSets1}, and identity
\eqref{eq:GeneralSpectrumPropertyOfSets2} comes from
\begin{IEEEeqnarray*}{rCl}
\av\left[ \frac{\left| B\cap\Sigma_{\mathcal{U}}(A) \right|}{|A|}
 \right]
&= &\sum_{\seq{y} \in B}
 \av\left[ \frac{1\{\seq{y}\in\Sigma_{\mathcal{U}}(A)\}}{|A|} \right]
\end{IEEEeqnarray*}
combined with \eqref{eq:GeneralSpectrumPropertyOfSets1} and
\cite[Proposition~2.1]{JSCC:Yang200904}.
\end{proofof}

\section{Proofs of Results in Section
 \ref{subsec:ConditionalProbability}}
\label{subsec:ProofOfConditionalProbability}

\begin{proofof}{Proposition \ref{pr:SpectrumPropertyX1OfFunctions}}
\begin{IEEEeqnarray*}{rCl}
\pr\{\rtilde{F}(\seq{x}) \in \mathcal{T}_{Q^{\mathcal{V}_0}}\}
&= &\pr\{\tilde{F}(\seq{x}) \in \mathcal{T}_{Q^{\mathcal{V}_0}}\}\\
&= &\av[1\{\tilde{F}(\seq{x})\in\mathcal{T}_{Q^{\mathcal{V}_0}}\}]\\
&= &\av\left[ \sum_{\seq{y} \in \mathcal{T}_{Q^{\mathcal{V}_0}}}
 1\{ (\seq{x}, \seq{y}) \in \mathrm{rl}(\tilde{F}) \} \right]\\
&= &\av\left[ \sum_{\seq{y} \in \mathcal{T}_{Q^{\mathcal{V}_0}}}
 1\{ (\seq{x}, \seq{y}) \in
 \Sigma_{\mathcal{U}_0\cup\mathcal{V}_0}(\mathrm{rl}(F)) \} \right]\\
&= &\av\left[ \left| \left( \seq{x}
 \times \mathcal{T}_{Q^{\mathcal{V}_0}} \right)
 \cap \Sigma_{\mathcal{U}_0 \cup \mathcal{V}_0}(\mathrm{rl}(F))
 \right| \right]\\
&\eqvar{(a)} &\frac{\avS_F(P^{\mathcal{U}_0}_{\seq{x}},
 Q^{\mathcal{V}_0})}
 {\prod_{U\in\mathcal{U}_0} \spec_{\field{q}^{|U|}}(P_{x_{U}})}\\
&= &\avS_F(Q^{\mathcal{V}_0}|P^{\mathcal{U}_0}_{\seq{x}}),
\end{IEEEeqnarray*}
where (a) follows from
Proposition~\ref{pr:GeneralSpectrumPropertyOfSets} and the identity
$|\mathrm{rl}(F)| = \prod_{i=1}^s q^{n_i}$.
\end{proofof}

\begin{proofof}{Proposition
 \ref{pr:SpectrumOfSeriallyConcatenatedFunctions}}
\begin{IEEEeqnarray*}{rCl}
\avS_{G \circ \Sigma_m \circ F}(Q|O)
&\eqvar{(a)} &\pr\{(\rtilde{G} \circ \rtilde{F})(\seq{x})
 \in \mathcal{T}_Q^l\}\\
&= &\sum_{P \in \mathcal{P}_m}
 \pr\{\rtilde{F}(\seq{x}) \in \mathcal{T}_P^m,
 \rtilde{G}(\rtilde{F}(\seq{x})) \in \mathcal{T}_Q^l\}\\
&= &\sum_{P \in \mathcal{P}_m} \Bigl(
 \pr\{\rtilde{F}(\seq{x}) \in \mathcal{T}_P^m\}\\
& &\breakop{\times} \pr\{\rtilde{G}(\rtilde{F}(\seq{x}))
 \in\mathcal{T}_Q^l | \rtilde{F}(\seq{x})\in\mathcal{T}_P^m\} \Bigr)\\
&\eqvar{(b)} &\sum_{P \in \mathcal{P}_m} \avS_F(P|O) \avS_G(Q|P),
\end{IEEEeqnarray*}
where (a) and (b) follow from
Proposition~\ref{pr:SpectrumPropertyX1OfFunctions}, and $\seq{x}$ is
an arbitrary sequence such that $P_{\seq{x}} = O$.
\end{proofof}

\section{Proofs of Results in Section
 \ref{subsec:SpectrumGeneratingFunctions}}
\label{subsec:ProofOfSpectrumGeneratingFunctions}

\begin{proofof}{Proposition
 \ref{pr:SubstitutionPrincipleOfGeneratingFunction}}
Since $\pi_\mathcal{V} = \psi \circ \pi_\mathcal{U}$,
\begin{IEEEeqnarray*}{rCl}
\gf_{\field{q}^{\mathcal{V}}}(A)(\seq{v}_{\mathcal{V}})
&= &\frac{1}{|A|} \sum_{\seq{x} \in A}
 \prod_{i=1}^n v_{\pi_\mathcal{V}(i), x_i}\\
&= &\frac{1}{|A|} \sum_{\seq{x} \in A}
 \prod_{i=1}^n v_{\psi(\pi_\mathcal{U}(i)), x_i}\\
&= &\psi(\gf_{\field{q}^{\mathcal{U}}}(A)(\seq{u}_{\mathcal{U}})).
\end{IEEEeqnarray*}
\end{proofof}

\begin{proofof}{Proposition
 \ref{pr:GeneratingFunctionOfProductOfSets}}
By definition,
\begin{IEEEeqnarray*}{rCl}
\gf_{\prod_{i=1}^s A_i}(\seq{u}_{\mathcal{I}_s})
&= &\frac{1}{|\prod_{i=1}^s A_i|} \sum_{\seq{x} \in \prod_{i=1}^s A_i}
 \prod_{i=1}^s \seq{u}_{U_i}^{|U_i|P_{x_{U_i}}}\\
&= &\frac{1}{\prod_{i=1}^s |A_i|} \prod_{i=1}^s
 \sum_{\seq{x}_i \in A_i} \seq{u}_{U_i}^{|U_i|P_{\seq{x}_i}}\\
&= &\prod_{i=1}^s \gf_{A_i}(\seq{u}_i),
\end{IEEEeqnarray*}
where $\{U_1, \ldots, U_s\}$ is the default coordinate partition.
\end{proofof}

\begin{proofof}{Corollary \ref{co:GeneratingFunctionOfSetProduct}}
\begin{IEEEeqnarray*}{rCl}
\gf_{A_1 \times A_2}(\seq{u})
&\eqvar{(a)} &\gf_{A_1 \times A_2}(\seq{u}, \seq{u}) \\
&\eqvar{(b)} &\gf_{A_1}(\seq{u}) \cdot \gf_{A_2}(\seq{u}),
\end{IEEEeqnarray*}
where (a) follows from
Proposition~\ref{pr:SubstitutionPrincipleOfGeneratingFunction} and (b)
follows from Proposition~\ref{pr:GeneratingFunctionOfProductOfSets}.
\end{proofof}

\begin{proofof}{Corollary
 \ref{co:GeneratingFunctionOfFunctionParallelProduct}}
\begin{IEEEeqnarray*}{rCl}
\gf_{f_1 \pprod f_2}(\seq{u}, \seq{v})
&\eqvar{(a)} &\gf_{f_1\pprod f_2}(\seq{u},\seq{u},\seq{v},\seq{v})\\
&\eqvar{(b)} &\gf_{f_1}(\seq{u}, \seq{v})
 \cdot \gf_{f_2}(\seq{u}, \seq{v}),
\end{IEEEeqnarray*}
where (a) follows from
Proposition~\ref{pr:SubstitutionPrincipleOfGeneratingFunction} with
$\mathrm{rl}(f_1 \pprod f_2) \subseteq \field{q}^{n_1+n_2}
 \times \field{q}^{m_1+m_2}
 = \field{q}^{n_1} \times \field{q}^{n_2} \times \field{q}^{m_1}
 \times \field{q}^{m_2}$,
and (b) follows from
Proposition~\ref{pr:GeneratingFunctionOfProductOfSets} with
$\mathrm{rl}(f_1 \pprod f_2)
 = \mathrm{rl}(f_1) \times \mathrm{rl}(f_2)$.
\end{proofof}

\begin{proofof}{Proposition \ref{pr:RenameGeneratingFunction}}
Let $A' = F(A)$.
Since $F$ is bijective, the generating function
$\gf_{\field{q}^{\mathcal{U}}}(A')$ can be rewritten as
\begin{IEEEeqnarray*}{rCl}
\gf_{\field{q}^{\mathcal{U}}}(A')(\seq{u}_{\mathcal{U}})
&= &\frac{1}{|A'|} \sum_{\seq{x}} 1\{\seq{x} \in A'\}
 \prod_{i=1}^n u_{\pi_\mathcal{U}(i),x_i}\\
&= &\frac{1}{|A|} \sum_{\seq{x}} 1\{\seq{x} \in A\}
 \prod_{i=1}^n u_{\pi_\mathcal{U}(i),F^{(i)}(x_i)}.
\end{IEEEeqnarray*}
Taking expectations on both sides, we obtain
\begin{IEEEeqnarray*}{rCl}
\avG_{\field{q}^{\mathcal{U}}}(A')(\seq{u}_{\mathcal{U}})
&= &\sum_{\seq{x}} \av\left[\frac{1\{\seq{x} \in A\}}{|A|}\right]
 \prod_{i=1}^n \av[u_{\pi_\mathcal{U}(i),F_{\pi_\mathcal{U}(i)}(x_i)}]
 \\
&= &\avG_{\field{q}^{\mathcal{U}}}(A)
 ((\av[u_{U,F_U(a)}])_{U\in\mathcal{U}, a\in\field{q}}),
\end{IEEEeqnarray*}
which is just $\overline{F}(\avG_{\field{q}^{\mathcal{U}}}(A))$.
\end{proofof}

\section{Proofs of Results in Section \ref{subsec:NewResultsV}}
\label{subsec:ProofOfNewResultsV}

To prove Theorem \ref{th:MacWilliamsIdentitiesF}, we need two lemmas.

\begin{lemma}[see e.g., \cite{JSCC:Yang201109}]
\label{le:MacWilliamsIdentities1F}
For a subspace $A$ of $\field{q}^n$,
\begin{equation}\label{eq:MacWilliamsIdentitiesLemma1}
\frac{1}{|A|} \sum_{\seq{x}_1 \in A} \chi(\seq{x}_1 \cdot \seq{x}_2)
= 1\{\seq{x}_2\in A^{\perp}\} \qquad \forall \seq{x}_2\in\field{q}^n.
\end{equation}
\end{lemma}

The reader is referred to \cite[Lemma~A.1]{JSCC:Yang201109} for a
proof.

\begin{lemma}\label{le:MacWilliamsIdentities2F}
Let $\mathcal{U}$ be a partition of $\mathcal{I}_n$.
Then
\[
\sum_{\seq{x}_2 \in \field{q}^n} \chi(\seq{x}_1 \cdot \seq{x}_2)
 \prod_{i=1}^n u_{\pi_\mathcal{U}(i), x_{2,i}}
= \prod_{U \in \mathcal{U}} (\seq{u}_U \mat{M})^{|U|P^{U}_{\seq{x}_1}}
\]
for all $\seq{x}_1 \in \field{q}^n$, where $\mat{M}$ is defined by
\eqref{eq:MacWilliamsMatrixF}.
\end{lemma}

\begin{proof}
\begin{IEEEeqnarray*}{rCl}
\sum_{\seq{x}_2 \in \field{q}^n} \chi(\seq{x}_1 \cdot \seq{x}_2)
 \prod_{i=1}^n u_{\pi_\mathcal{U}(i), x_{2,i}}
&\eqvar{(a)} &\sum_{\seq{x}_2 \in \field{q}^n} \prod_{i=1}^n
 \chi(x_{1,i} x_{2,i}) u_{\pi_\mathcal{U}(i),x_{2,i}}\\
&= &\prod_{U \in \mathcal{U}} \prod_{i \in U} \sum_{a_2 \in \field{q}}
 \chi(x_{1,i} a_2) u_{U,a_2}\\
&= &\prod_{U \in \mathcal{U}} \prod_{a_1 \in \field{q}}
 \left( \sum_{a_2 \in \field{q}} \chi(a_1 a_2) u_{U,a_2}
 \right)^{|U| P^{U}_{\seq{x}_1}(a_1)}\\
&= &\prod_{U \in \mathcal{U}}
 (\seq{u}_U \mat{M})^{|U| P^{U}_{\seq{x}_1}}
\end{IEEEeqnarray*}
where (a) follows from
$\chi(\seq{x}_1 \cdot \seq{x}_2)
 = \prod_{i=1}^n \chi(x_{1,i} x_{2,i})$.
\end{proof}

\begin{proofof}{Theorem \ref{th:MacWilliamsIdentitiesF}}
\begin{IEEEeqnarray*}{rCl}
\gf_{A^{\perp}}(\seq{u}_{\mathcal{U}})
&= &\frac{1}{|A^{\perp}|} \sum_{\seq{x}_2 \in \field{q}^n}
 1\{\seq{x}_2 \in A^{\perp}\} \prod_{i=1}^n
 u_{\pi_\mathcal{U}(i), x_{2,i}}\\
&\eqvar{(a)} &\frac{1}{|A^{\perp}|} \sum_{\seq{x}_2 \in \field{q}^n}
 \frac{1}{|A|} \sum_{\seq{x}_1 \in A} \chi(\seq{x}_1 \cdot \seq{x}_2)
 \prod_{i=1}^n u_{\pi_\mathcal{U}(i), x_{2,i}}\\
&= &\frac{1}{|A||A^{\perp}|} \sum_{\seq{x}_1 \in A}
 \sum_{\seq{x}_2 \in \field{q}^n} \chi(\seq{x}_1 \cdot \seq{x}_2)
 \prod_{i=1}^n u_{\pi_\mathcal{U}(i), x_{2,i}}\\
&\eqvar{(b)} &\frac{1}{|A||A^{\perp}|} \sum_{\seq{x}_1 \in A}
 \prod_{U \in \mathcal{U}}
 (\seq{u}_U \mat{M})^{|U|P^{U}_{\seq{x}_1}}\\
&= &\frac{1}{|A^{\perp}|}
 \gf_A((\seq{u}_U \mat{M})_{U \in \mathcal{U}})
\end{IEEEeqnarray*}
where (a) follows from Lemma~\ref{le:MacWilliamsIdentities1F} and (b)
follows from Lemma~\ref{le:MacWilliamsIdentities2F}.
\end{proofof}

\begin{proofof}{Theorem \ref{th:MacWilliamsIdentitiesJF}}
Define the sets
\[
Z_1
\eqdef \{(\seq{x}\mat{A}, \seq{x}) \in \field{q}^m \times \field{q}^n:
 \seq{x} \in \field{q}^n\}
\]
and
\[
Z_2
\eqdef \{(\seq{y}, -\seq{y}\transpose{\mat{A}})
 \in \field{q}^m \times \field{q}^n: \seq{y} \in \field{q}^m\}.
\]
Clearly, for any $\seq{z}_1 = (\seq{x}\mat{A}, \seq{x}) \in Z_1$ and
$\seq{z}_2 = (\seq{y}, -\seq{y}\transpose{\mat{A}}) \in Z_2$, we have
\begin{IEEEeqnarray*}{rCl}
\seq{z}_1 \cdot \seq{z}_2
&= &(\seq{x}\mat{A}) \cdot \seq{y} + \seq{x} \cdot (-\seq{y}\transpose{\mat{A}}) \\
&= &(\seq{x}\mat{A}) \transpose{\seq{y}} - \seq{x} \transpose{(\seq{y}\transpose{\mat{A}})} \\
&= &\seq{x}\mat{A}\transpose{\seq{y}} - \seq{x}\mat{A}\transpose{\seq{y}} \\
&= &0
\end{IEEEeqnarray*}
which implies $Z_2 \subseteq Z_1^{\perp}$.
Note that $|Z_1||Z_2| = q^{m+n}$.
This, together with the identity $|Z_1||Z_1^{\perp}| = q^{m+n}$, gives
$Z_2 = Z_1^{\perp}$.

Then it follows from Theorem~\ref{th:MacWilliamsIdentitiesF} that
\begin{IEEEeqnarray*}{rCl}
\gf_{\field{q}^{\mathcal{V}}\field{q}^{\mathcal{U}}}(-g)
 (\seq{v}_{\mathcal{V}}, \seq{u}_{\mathcal{U}})
&= &\gf_{\field{q}^{\mathcal{V}}\field{q}^{\mathcal{U}}}(Z_2)
 (\seq{v}_{\mathcal{V}}, \seq{u}_{\mathcal{U}})\\
&= &\gf_{\field{q}^{\mathcal{V}}\field{q}^{\mathcal{U}}}
 (Z_1^{\perp})(\seq{v}_{\mathcal{V}}, \seq{u}_{\mathcal{U}})\\
&= &\frac{1}{|Z_1^{\perp}|}
 \gf_{\field{q}^{\mathcal{V}}\field{q}^{\mathcal{U}}}(Z_1)
 ((\seq{v}_V \mat{M})_{V \in \mathcal{V}},
 (\seq{u}_U \mat{M})_{U \in \mathcal{U}})\\
&= &\frac{1}{q^m}
 \gf_{\field{q}^{\mathcal{U}}\field{q}^{\mathcal{V}}}(f)
 ((\seq{u}_U \mat{M})_{U \in \mathcal{U}},
 (\seq{v}_V \mat{M})_{V \in \mathcal{V}})
\end{IEEEeqnarray*}
as desired.
\end{proofof}

%% file: whyspec.tex

\section{Spectrum or Complete Weight Distribution?}
\label{sec:SpectrumVsCWD}

As discussed in Section~\ref{sec:BasicsOfCodeSpectrumApproach}, now
that spectrum is simply the normalization of complete weight
distribution, why use it at all?
In the nonrandom setting, these two concepts make indeed no
difference.
If random encoders are involved, however, there is a remarkable
difference.

Let $A$ be a random nonempty subset of $\field{q}^n$.
Then its average spectrum is $\avS(A)$ while its average complete
weight distribution can be expressed as $\av[|A|\spec(A)]$.
Note that the equation $\av[|A|\spec(A)]=\av[|A|]\avS(A)$ does not
hold in general (unless $|A|$ and $\spec(A)$ are uncorrelated, e.g.,
$|A|$ is nonrandom), so there is no simple relation between $\avS(A)$
and $\av[|A|\spec(A)]$.
In fact, for certain random sets we can obtain an elegant exact
formula for the average spectrum, but only a less strict approximate
expression for the average complete weight distribution (for example,
with the assumption that the vectors are not necessarily distinct
\cite{JSCC:Blinovsky200912}).
Similarly, there are some cases in which the complete weight
distribution is more appropriate.

\begin{example}
Consider the random linear encoder $\rlccode{m,n}$ over $\field{q}$.
By \citeGoodLinearCode, its average joint spectrum is
\[
\avS_{\rlccode{m,n}}(P,Q)=
\begin{cases}
q^{-m}1\{Q=P_{0^n}\} &P = P_{0^m} \\
q^{-m-n}{m \choose mP}{n \choose nQ} &\mbox{otherwise},
\end{cases}
\]
so the average image spectrum of $\rlccode{m,n}$, or equivalently, the
average spectrum of $C_1 = \rlccode{m,n}(\field{q}^m)$ is 
\[
\avS_{C_1}(Q)=
\begin{cases}
q^{-m} + q^{-n}(1-q^{-m}) &Q = P_{0^n} \\
q^{-n}(1-q^{-m}){n \choose nQ} &\mbox{otherwise}.
\end{cases}
\]
On the other hand, it follows from
Proposition~\ref{pr:GeneralSpectrumPropertyOfSets} that the average
spectrum of $C_2 = \ker\rlccode{n,m}$ is 
\[
\avS_{C_2}(P)
= {n \choose nP} \av\left[\frac{1\{\seq{x}\in C_2\}}{|C_2|}\right],
\]
where $\seq{x}$ is an arbitrary vector of $\mathcal{T}_P^n$.
Since the expectation term on the right-hand side is too complicated,
we cannot proceed without resorting to approximation.
However, computing instead the average complete weight distribution,
we obtain
\begin{IEEEeqnarray*}{rCl}
\av[|C_2|\spec_{C_2}(P)]
&= &{n \choose nP} \av[1\{\seq{x}\in C_2\}] \\
&= &{n \choose nP} \pr\left\{\rlccode{n,m}(\seq{x}) = 0^m\right\} \\
&= &\begin{cases}
1 &P=P_{0^n}\\
q^{-m}{n \choose nP} &\mbox{otherwise},
\end{cases}
\end{IEEEeqnarray*}
which is surprisingly simple compared to the spectrum form.
A similar situation is encountered when computing the average complete
weight distribution of $C_1$.
In fact, the generator matrix of $\rlccode{n,m}$ is the transpose of
the generator matrix of $\rlccode{m,n}$, so $C_1$ and $C_2$, as a pair
of dual codes, satisfy
\[
\av[|C_2| \gf_{C_2}(\seq{u})] = \avG_{C_1}(\seq{u}\mat{M})
\]
and
\[
\av[|C_1| \gf_{C_1}(\seq{u})] = \avG_{C_2}(\seq{u}\mat{M})
\]
according to Theorem~\ref{th:MacWilliamsIdentitiesF}.
These two identities explain why the average spectrum of $C_1$ and the
average complete weight distribution of $C_2$ are easy to compute
while the other two quantities do not have simple expressions. 
\end{example}

It turns out that most of the cases treated in this paper are more
conveniently formulated in the spectrum form.
For this reason we have chosen the code-spectrum approach.
Furthermore, as a side benefit, the law of serial concatenation of
linear encoders can be intuitively put in analogy with the
concatenation of conditional probability distributions (see
Section~\ref{subsec:ConditionalProbability}).